\documentclass[a4paper,reqno,12pt]{amsart}

\usepackage[hmargin=2cm,vmargin=1.9cm]{geometry}

\usepackage{amsmath,amssymb,amsthm,mathtools,cancel,amsthm,amssymb,amsmath,mathrsfs,stmaryrd,lscape,rotating,adjustbox,
			mathdots,mathtools,color,graphicx,framed}

\mathtoolsset{showonlyrefs}

\usepackage[colorlinks=true, pdfstartview=FitV, linkcolor=black, citecolor=black, urlcolor=black]{hyperref}

\definecolor{shadecolor}{rgb}{0.9, 0.9, 0.86}

\def\Re{\mathrm {Re}\,}
\def\Im{\mathrm {Im}\,}
\def\wt{\widetilde}
\def\wh{\widehat}
\def\C{\mathbb{C}}
\def\R{\mathbb{R}}
\newcommand{\Q}{\mathbb{Q}}
\def\G{\Gamma}
\def\1{\mathbf{1}}
\def\H{\mathrm{H}}
\def\g{\gamma}
\def\b{\beta}
\def\a{\alpha}
\def\L{\Lambda}
\def\d{\mathrm d}
\def\O{\mathcal{O}}
\def\e{\mathrm{e}}
\def\i{\mathrm{i}}
\def\t{\mathbf{t}}
\def\pa{\partial}
\def\res{\mathop{\mathrm {res}}\limits_}
\def\tr{\mathrm{tr}\,}

\def\s{\sigma}
\def\E{\mathrm{E}}
\def\Mgn{\overline{\mathcal{M}}_{g,n}}
\def\ss{\mathbf{s}}
\def\barnes{\mathrm{G}}
\def\Vol{\mathrm{Vol}}
\def\0{\boldsymbol{0}}
\def\be{\begin{equation}}
\def\ee{\end{equation}}

\newtheorem{theorem}{Theorem}[section]
\newtheorem{example}[theorem]{Example}
\newtheorem{lemma}[theorem]{Lemma}
\newtheorem{remark}[theorem]{Remark}
\newtheorem{proposition}[theorem]{Proposition} 
\newtheorem{corollary}[theorem]{Corollary}

\newmuskip\pFqmuskip
\newcommand*\pFq[6][8]{
  \begingroup 
  \pFqmuskip=#1mu\relax
  \mathchardef\normalcomma=\mathcode`, 
  \mathcode`\,=\string"8000            
  \begingroup\lccode`\~=`\,
  \lowercase{\endgroup\let~}\pFqcomma  
  {}_{#2}F_{#3}{\left(\left.\genfrac..{0pt}{}{#4}{#5}\right|#6\right)}
  \endgroup
}
\newcommand{\pFqcomma}{{\normalcomma}\mskip\pFqmuskip}

\def\shrink#1{
\noindent
\adjustbox{width=1\textwidth}{
	\noindent\begin{minipage}{1\textwidth}
	{#1}
	\end{minipage}} 
}

\begin{document}

\title{Laguerre Ensemble: Correlators, Hurwitz Numbers and Hodge Integrals}

\author{Massimo Gisonni}
\address{M.~Gisonni:\newline SISSA, via Bonomea 265, 34136 Trieste, Italy}
\email{massimo.gisonni@sissa.it}

\author{Tamara Grava}
\address{T.~Grava:\newline School of Mathematics, Bristol University, UK
\newline SISSA, via Bonomea 265, 34136 Trieste, Italy}
\email{grava@sissa.it}

\author{Giulio Ruzza}
\address{G.~Ruzza:\newline 
SISSA, via Bonomea 265, 34136 Trieste, Italy
\newline IRMP, UCLouvain, Chemin du Cyclotron 2, 1348 Ottignies-Louvain-la-Neuve, Belgium}
\email{giulio.ruzza@uclouvain.be}

\pagestyle{myheadings}

\numberwithin{equation}{section}

\begin{abstract}
We consider the  Laguerre partition function, and derive explicit generating functions for connected correlators with arbitrary integer powers of traces in terms of products of Hahn polynomials.
It was recently proven in \cite{CDO2018} that correlators have a topological expansion in terms of weakly or strictly monotone Hurwitz numbers, that can be explicitly computed from our formul\ae. 
As a second result we identify the Laguerre partition function with only positive couplings and a special value of the parameter $\alpha=-1/2$ with the modified GUE partition function, which has recently been introduced  in \cite{DLYZ2016} as a generating function for Hodge integrals. 
This identification provides a direct and new link between monotone Hurwitz numbers and Hodge integrals.
\end{abstract}

\maketitle

\section{Introduction and Results}
\subsection{Laguerre Unitary Ensemble (LUE) and formul\ae\ for correlators}
The LUE is the statistical model on the cone $\H_N^+$ of positive definite hermitian matrices  of size $N$ endowed with the probability measure
\be
\label{laguerreweight}
\frac{1}{Z_N(\a;\0)}{\det}^\a X\exp\tr(-X)\d X,
\ee
$\d X$ being the restriction to $\H_N^+$ of the Lebesgue measure on the space $\H_N\simeq\R^{N^2}$ of hermitian matrices $X=X^\dagger$ of size $N$;
\be
\label{lebesgue}
\d X:=\prod_{1\leq i\leq N} \d X_{ii} \prod_{1\leq i<j\leq N}\d \Re X_{ij}\d \Im X_{ij}.
\ee
The normalizing constant $Z_N(\a;\0)$ in \eqref{laguerreweight} is computed explicitly as
\be
\label{normalizinglaguerre}
Z_N(\a;\0):=\int_{\H_N^+}{\det}^\a X\exp[\tr(-X)]\d X=\pi^{\frac{N(N-1)}2}\prod_{j=1}^N\G(j+\a).
\ee
The parameter $\a$ could be taken as an arbitrary complex number satisfying $\Re\a>-1$.
Writing $\alpha=M-N$, a random matrix $X$ distributed according the measure \eqref{laguerreweight} is called complex \emph{Wishart matrix} with parameter $M$; in particular, when $M$ is an integer there is the equality in law $X=\frac 1N WW^\dagger$ where $W$ is an $N\times M$ random matrix with independent identically distributed Gaussian entries \cite{F2010}.

Our first main result, Theorem \ref{thmmain} below, concerns explicit and effective formul\ae\ for \emph{correlators} of the LUE
\be
\left\langle \tr X^{k_1}\cdots\tr X^{k_r}\right\rangle:=\frac 1{Z_N(\a;\0)}\int_{\H_N^+} \tr X^{k_1}\cdots\tr X^{k_r} {\det}^\a X \exp\tr(-X)\d X
\ee
for \emph{arbitrary nonzero integers} $k_1,\dots,k_r\in\mathbb{Z}\setminus\{0\}$. Theorem \ref{thmmain} is best formulated in terms of \emph{connected} correlators
\be
\label{connectedcorrelators}
\left\langle \tr X^{k_1}\cdots\tr X^{k_r}\right\rangle_{\mathsf{c}}:=\sum_{\mathcal{P}\text{ partition of }\{1,\dots,r\}}(-1)^{|\mathcal{P}|-1}(|\mathcal{P}|-1)!\prod_{I\in\mathcal{P}}\left\langle\prod_{i\in I}\tr X^{k_i}\right\rangle,
\ee
e.g.
\be
\left\langle \tr X^{k_1}\right\rangle_{\mathsf{c}}:=\left\langle \tr X^{k_1}\right\rangle,
\qquad
\left\langle \tr X^{k_1}\tr X^{k_2}\right\rangle_{\mathsf{c}}:=\left\langle\tr X^{k_1}\tr X^{k_2}\right\rangle-\left\langle\tr X^{k_1}\right\rangle\left\langle\tr X^{k_2}\right\rangle.
\ee
The generating function for connected correlators
\be
\left\langle\tr\left(\frac{1}{x_1-X}\right)\tr\left(\frac{1}{x_2-X}\right)\cdots \tr\left(\frac{1}{x_r-X}\right)\right\rangle_{\mathsf{c}}
\ee
can be expanded near $x_j=\infty$ and/or $x_j=0$, yielding the following generating functions up to some irrelevant terms; for $r=1$
\be
\label{generatingonepoint}
\quad C_{1,0}(x):=\sum_{k\geq 1}\frac{1}{x^{k+1}}\left\langle\tr X^{k}\right\rangle,
\quad C_{0,1}(x):=-\sum_{k\geq 1}x^{k-1}\left\langle\tr X^{-k}\right\rangle,
\ee
for $r=2$
\begin{align}
\nonumber
C_{2,0}(x_1,x_2)&:=\sum_{k_1,k_2\geq 1}\frac{\left\langle\tr X^{k_1}\tr X^{k_2}\right\rangle_{\mathsf{c}}}{x_1^{k_1+1}x_2^{k_2+1}},\quad
C_{1,1}(x_1,x_2):=-\sum_{k_1,k_2,\geq 1}\frac{x_2^{k_2-1}}{x_1^{k_1+1}}\left\langle\tr X^{k_1}\tr X^{-k_2}\right\rangle_{\mathsf{c}},
\\ \label{2generating}
C_{0,2}(x_1,x_2)&:=\sum_{k_1,k_2\geq 1}x_1^{k_1-1}x_2^{k_2-1}\left\langle\tr X^{-k_1}\tr X^{-k_2}\right\rangle_{\mathsf{c}},
\end{align}
and, in general,
\be
\label{generatingmixed}
C_{r_+,r_-}(x_1,\dots,x_{r}):=(-1)^{r_-}\sum_{k_1,\dots,k_{r}\geq 1}\frac{1}{x_1^{\s_1k_1+1}\cdots x_r^{\s_rk_r+1}}\left\langle\tr X^{\s_1 k_1}\cdots \tr X^{\s_r k_r}\right\rangle_{\mathsf{c}},
\ee
where $r=r_++r_-$ and we introduce the signs
\be
\label{sign}
\s_1=\dots=\s_{r_+}=+,\qquad\s_{r_++1}=\dots=\s_r=-.
\ee

We obtain the  following explicit expression for these generating functions for correlators.

\begin{shaded}
\begin{theorem}
\label{thmmain}
Introduce the matrix-valued formal series
\begin{align}
\label{R+}
\!\!\!\!\!\!\!\!\!\!\!\!\!\!R_+(x)\!&:=\!\begin{pmatrix} 1 & 0 \\ 0 & 0 \end{pmatrix}+\sum_{\ell\geq 0}\frac 1{x^{\ell+1}}
\begin{pmatrix}
\ell A_{\ell}(N,N+\a) & B_\ell(N+1,N+\a+1)
\\
-N(N+\a)B_\ell(N,N+\a) & -\ell A_{\ell}(N,N+\a) 
\end{pmatrix}
\\
\label{R-}
\!\!\!\!\!\!\!\!\!\!\!\!\!\!R_-(x)\!&:=\!\begin{pmatrix} 1 & 0 \\ 0 & 0 \end{pmatrix}+\sum_{\ell\geq 0}\frac{x^{\ell}}{(\a-\ell)_{2\ell+1}}
\begin{pmatrix}
(\ell+1)A_{\ell}(N,N+\a) & -B_\ell(N+1,N+\a+1)
\\
N(N+\a)B_\ell(N,N+\a) & -(\ell+1)A_{\ell}(N,N+\a)
\end{pmatrix}
\end{align}
where, using a standard notation $(p)_j:=p(p+1)\cdots(p+j-1)$ for the \emph{rising factorial},
\be
\label{ABlMN}
A_\ell(N,M):=\begin{cases}
N, & \ell=0,
\\
\frac 1\ell\sum\limits_{j=0}^{\ell-1}(-1)^j\frac{(N-j)_{\ell}(M-j)_{\ell}}{j!(\ell-1-j)!}, &\ell\geq 1,
\end{cases}
\qquad
B_\ell(N,M):=\sum_{j=0}^\ell(-1)^j\frac{(N-j)_{\ell}(M-j)_{\ell}}{j!(\ell-j)!}.
\ee
Then the generating functions \eqref{generatingmixed} can be expressed as
\begin{align}
\nonumber
C_{1,0}(x)&=\frac 1 x \int_x^{\infty} \left[\left(R_+(y)\right)_{11}- 1\right]\d y, 
& C_{0,1}(x)&=\frac 1x \int_0^x \left[1-\left(R_-(y)\right)_{11}\right]\d y,
\\
\nonumber
C_{2,0}(x_1,x_2)&=\frac{\tr\left(R_+(x_1)R_+(x_2)\right)-1}{(x_1-x_2)^2},
&
C_{1,1}(x_1,x_2)&=\frac{\tr\left(R_+(x_1)R_-(x_2)\right)-1}{(x_1-x_2)^2},
\\
\label{2pointMain}
C_{0,2}(x_1,x_2)&=\frac{\tr\left(R_-(x_1)R_-(x_2)\right)-1}{(x_1-x_2)^2},
\end{align}
and, in general,
\be
\label{multipoint}
C_{r_+,r_-}(x_1,\dots,x_r)=-\sum_{(i_1,\dots,i_r)\in\mathcal{C}_r}\frac{\tr\left(R_{\s_{i_1}}(x_{i_1})\cdots R_{\s_{i_r}}(x_{i_r})\right)-\delta_{r,2}}{(x_{i_1}-x_{i_2})\cdots (x_{i_{r-1}}-x_{i_r})(x_{i_r}-x_{i_1})},
\ee
where $r=r_++r_-\geq 2$, the summation extends over   the $r-$cycles $(i_1,\dots,i_r)$ in the group of permutations of $\{1,\dots,r\}$, and we use the signs $\s_1,\dots,\s_r$ defined in \eqref{sign}.
\end{theorem}
\end{shaded}

The proof is given in Section \ref{subsecproof}. Theorem~\ref{thmmain} generalizes formul\ae\ for one-point correlators, since  the formul\ae\ for the generating series $C_{1,0}$ and $C_{0,1}$ boil down to the following identities
\be
\label{onepoint+-}
\left\langle\tr X^k\right\rangle=A_k(N,N+\a),
\qquad
\left\langle\tr X^{-k-1}\right\rangle=\frac{A_k(N,N+\a)}{(\a-k)_{2k+1}},
\qquad k\geq 0
\ee
which were already derived in the literature \cite{HSS1992,CMOS2019}.
From Theorem~\ref{thmmain} for example one can deduce compact expressions for correlators of the form
\begin{align}
\!\!\!\!\!\!\!\!\!\left \langle \tr X^{k} \tr X \right \rangle_{\mathsf{c}}
&= k A_k(N,N+\a),
&\!\!\!\!\!\!\!\!\!\!\!\!\!\!\!\!\!\!\!\!\left \langle \tr X^{-k} \tr X^{-1} \right \rangle_{\mathsf{c}}
&= \frac{k A_k(N,N+\a)}{\alpha (\alpha-k)_{2k+1}},
\nonumber
\\
\label{2pointCorr}
\!\!\!\!\!\!\!\!\!\left \langle \tr X^{k} \tr X^{-1} \right \rangle_{\mathsf{c}}
& = -\frac{k A_{k-1}(N,N+\a)}{\alpha}, 
&\!\!\!\!\!\!\!\!\!\!\!\!\!\!\!\!\!\!\!\!\left \langle \tr X^{-k} \tr X \right \rangle_{\mathsf{c}} &= -\frac{k A_{k-1}(N,N+\a)}{(\alpha-k+1)_{2k-1}},
\end{align}
for the derivation see Example \ref{2pointExp}.
For general positive moments, see also \cite{JKM2020}.

The entries $A_\ell(N,M)$, defined in \eqref{ABlMN}, are known to satisfy a three term recursion \cite{HT2003,CMSV2016}. We deduce this recursion together with a similar three term recursion for $B_\ell(N,M)$, in Lemma \ref{lemmarecursion}.
It was pointed out in \cite{CMOS2019} that the entries $A_\ell(N,M)$ are \emph{hypergeometric orthogonal polynomials} (in particular suitably normalized \emph{Hahn} and \emph{dual Hahn polynomials} \cite{KS1994,CMOS2019}), a fact which provides another interpretation of the same three term recursion; this interpretation extends to the entries $B_\ell(N,M)$, see Remark \ref{remarkhg}.
In Lemma \ref{lemmalargeN} we provide an alternative expression for the entries $A_\ell(N,M),B_\ell(N,M)$, which makes clear that they are polynomials in $N,M$ with integer coefficients.

Formul\ae\ of the same sort as \eqref{multipoint} have been considered in \cite{DY2017} for the Gaussian Unitary Ensemble, and already appeared in the Topological Recursion literature, see e.g. \cite{EO2007,CE2006,BE2009,BBE2015,EKR2015}.
Our approach is not directly based  on the Matrix Resolvent method \cite{DY2017} or the Topological Recursion \cite{CE2006}; in particular we provide a self-contained proof to Theorem \ref{thmmain} via orthogonal polynomials and their Riemann-Hilbert problem
\cite{IKF1990}.   

Insertion of \emph{negative} powers of traces in the correlators and computation of \emph{mixed} correlators are, to the best of our knowledge, novel aspects; as we shall see shortly, these general correlators have expansions with integer coefficients, a fact which generalizes results of  (see e.g.\ \cite{CDO2018}).
It would be interesting to implement this method to other invariant ensembles of random matrices \cite{DG2009,F2010}. 
With the aid of the formul\ae\ of Theorem~\ref{thmmain} we have computed several LUE connected correlators which are reported in the tables of App.\ \ref{apptable}.
Moreover, we can make direct use of the formul\ae\ of Theorem~\ref{thmmain} to prove (details in Section \ref{secproofprop13}) the following result, concerning the formal structure as large $N$ asymptotic series of \emph{arbitrary} correlators of the LUE in the scaling
\be
\label{scaling}
\alpha=(c-1)N,
\ee
corresponding to $M=cN$ in terms of the Wishart parameter $M$.

\begin{proposition}
\label{prop13}
Arbitrary rescaled LUE correlators admit an asymptotic expansion for $N\to\infty$ which is a series in $N^{-2}$ with coefficients polynomial in $c$ and $(c-1)^{-1}$ with integer coefficients. More precisely, for all $k_1,\dots,k_\ell \in \mathbb{Z}\setminus \{0\}$ we have
\be
\label{Texp+-}
N^{-[\ell\,\mathrm{mod}\,2]-\sum_{i=1}^\ell k_i}\left\langle \tr X^{k_1}\cdots\tr X^{k_\ell}\right\rangle_{\mathsf{c}}\sim \sum_{j\geq 0} \frac{f_j^{(k_1,\dots,k_\ell)}\left(c\right)}{N^{2j}},\qquad N\to\infty,
\ee
where $f_j^{(k_1,\dots,k_\ell)}(c)\in\mathbb{Z}\left[c,(c-1)^{-1}\right]$ for all $j\geq 0$; here we also denote $[\ell\,\mathrm{mod}\,2]\in\{0,1\}$ the parity of $\ell$.
\end{proposition}

From this result we infer that when $c=2$ (equivalently, $\a=N$) the coefficients of this large $N$ expansion are all integers.

From the tables in App.\ \ref{apptable} one easily conjectures that actually a stronger version of this result holds true, namely that that the asymptotic expansion for $N^{\ell-2-\sum_{i=1}^\ell k_i}\left\langle \tr X^{k_1}\cdots\tr X^{k_\ell}\right\rangle_{\mathsf{c}}$ (note the different power of $N$) as $N\to\infty$ is a series in $\mathbb Z[c,(c-1)^{-1}][[N^{-2}]]$. 
Such stronger property holds true when all the $k_j$'s have the same sign, see e.g. \cite{CDO2018,CMSV2016} and the section below.

\subsection{Topological expansions and Hurwitz numbers}\label{paragraphhurwitz}

It has been shown in \cite{ME2003,BG2013} that for matrix models with convex potentials, as in our case, correlators, suitably rescaled by a power of $N$, as in \eqref{Texp+-}, have a \emph{topological expansion}, by which we mean an asymptotic expansion in non-negative powers of $N^{-2}$.
As mentioned above, the topological expansion of the LUE correlators in the regime \eqref{scaling} was considered in \cite{CMSV2016,CDO2018} where the connection with Hurwitz numbers was made explicit.

Hurwitz numbers are very important combinatorial quantities, counting factorizations in the symmetric group; they were first studied in the end of the 19th century by Hurwitz and there has been a recent renewal of interest in view of the connection with integrable systems and random matrices \cite{O2000,HO2015}. The Hurwitz numbers related to this model \cite{CDO2018} are a variant of \emph{monotone} Hurwitz numbers \cite{GGPN2014,GGPN2016,GGPN2017,BGF2017,BCDGF2019} which can be defined as follows.
For $\mu,\nu$ partitions of the same integer $d=|\mu|=|\nu|$, define the \emph{strictly} (resp. \emph{weakly}) \emph{monotone double Hurwitz numbers} $h_g^>(\mu;\nu)$ (resp. $h_g^\geq(\mu;\nu)$) as the number of tuples $(\a,\tau_1,\dots,\tau_r,\beta)$ such that
\begin{itemize}
\item[(i)] $r=\ell+s+2g-2$ where $\ell$ is the length of $\mu$ and $s$ is the length of $\nu$,

\item[(ii)] $\a,\b$ are permutations of $\{1,\dots,d\}$ of cycle type $\mu,\nu$, respectively, and $\tau_1,\dots,\tau_r$ are transpositions such that $\a\tau_1\cdots\tau_r=\b$,

\item[(iii)] the subgroup generated by $\a,\tau_1,\dots,\tau_r$ acts transitively on $\{1,\dots,d\}$, and

\item[(iv)] writing $\tau_j=(a_j,b_j)$ with $a_j<b_j$ we have
$b_1<\cdots<b_r$ (resp. $b_1\leq\cdots\leq b_r$).
\end{itemize}

\begin{theorem}[\cite{CDO2018}]
\label{thmhurwitz}
The following asymptotic expansions as $N\to\infty$ hold true;
\begin{align}
\label{hurwitzstrict}
\!\!\!\!\!\!\!\!\!\!\!\!\!\!\!\!\!\!N^{\ell-|\mu|-2}\left\langle \tr X^{\mu_1}\cdots\tr X^{\mu_\ell}\right\rangle_{\mathsf{c}}&=\sum_{g\geq 0}\frac 1{N^{2g}}\sum_{s=1}^{1-2g+|\mu|-\ell}H_g^>(\mu;s)c^s,
&&
\!\!\!\!\!\!\!\!\!\!\!\!c>1-\frac 1N,
\\
\label{hurwitzweakly}
\!\!\!\!\!\!\!\!\!\!\!\!\!\!\!\!\!\!N^{\ell+|\mu|-2}\left\langle \tr X^{-\mu_1}\cdots\tr X^{-\mu_\ell}\right\rangle_{\mathsf{c}}&=\sum_{g\geq 0}\frac 1{N^{2g}}\sum_{s\geq 1}\frac {H_g^\geq(\mu;s)}{(c-1)^{2g-2+|\mu|+\ell+s}},
&&
\!\!\!\!\!\!\!\!\!\!\!\!c>1+\frac{|\mu|}N,
\end{align}
where we denote $|\mu|:=\mu_1+\cdots+\mu_\ell$, and the coefficients can be expressed as
\be
\label{eq122}
H_g^>(\mu;s)=\frac{z_\mu}{|\mu|!}\sum_{\nu\text{ of length }s}h_g^>(\mu;\nu),\qquad
H_g^\geq(\mu;s)=\frac{z_\mu}{|\mu|!}\sum_{\nu\text{ of length }s}h_g^\geq(\mu;\nu),
\ee
where $z_\mu:=\prod_{i\geq 1}\left(i^{m_i}\right)m_i!$, $m_i$ being the multiplicity of $i$ in the partition $\mu$.
\end{theorem}

From the structure of the formula \eqref{hurwitzweakly} it is clear that when $c=2$ (equivalently, $\a=N$) the coefficients in this expansion are all positive integers.

\begin{remark}
The type of Hurwitz numbers apparing in the expansions \eqref{hurwitzstrict} and \eqref{hurwitzweakly} can also be expressed in terms of the \emph{(connected) multiparametric weighted Hurwitz numbers} $\wt {H}_G^d(\mu)$, introduced and studied in \cite{HO2015,GPH2015,ACEH2020,BHR2020}, which depend on a single partition $\mu$ and are parametrized by a positive integer $d$ and by a sequence $g_1,g_2,\dots$ of complex numbers, the latter being encoded in the series $G(z)=1+\sum_{i\geq 1}g_iz^i$. To make the comparison precise, one has to identify
\be
d=2g-2-|\mu|-\ell(\mu)
\ee
and then we have
\begin{align}
\sum_{s=1}^{1-2g+|\mu|-\ell}H_g^{>}(\mu;s)c^s&=z_\mu c^{|\mu|-d}\wt H_G^d(\mu),& G(z)&=(1+cz)(1+z),
\\
\sum_{s\geq 1}\frac{H_g^{\geq}(\mu;s)}{(c-1)^s}&=\frac{z_\mu}{(c-1)^{|\mu|+d}}\wt H_G^d(\mu),& G(z)&=\frac{1+(c-1)z}{1-z},
\end{align}
where $z_\mu:=\prod_{i\geq 1}\left(i^{m_i}\right)m_i!$, $m_i$ being the multiplicity of $i$ in the partition $\mu$, as above.
\end{remark}

\subsection{Laguerre and modified GUE partition functions and Hodge integrals} 

Our arguments in the proof of Theorem \ref{thmmain} mainly revolve around the following generating function for correlators
\be
\label{Z}
Z_N(\a;\t_+,\t_-)=\int_{\H_N^+}{\det}^\a X\exp\tr\left(-X+\sum_{k\not=0}t_kX^k\right) \d X
\ee
which we call \emph{LUE partition function}. Here $\t_+=(t_{1},t_{2},\dots)$ and $\t_-=(t_{-1},t_{-2},\dots)$ are two independent infinite vectors of times and $\a$ is a complex parameter. For precise analytic details about the definition \eqref{Z} we refer to the beginning of Section \ref{secproof}.
Eventually we are interested in the formal expansion as $t_j\to 0$; more precisely, logarithmic derivatives of the LUE partition function at $\t_+=\t_-=\0$ recover the connected correlators \eqref{connectedcorrelators} as
\be
\label{connectedmoments}
\left.\frac{\pa^\ell\log Z_N(\a;\t_+,\t_-)}{\pa t_{k_1}\cdots \pa t_{k_\ell}}\right|_{\t_+=\t_-=\0}=\left\langle\tr X^{k_1}\cdots\tr X^{k_\ell}\right\rangle_{\mathsf{c}}.
\ee

It is known that $Z_N(\a;\t_+,\t_-)$ is a \emph{Toda lattice} tau function \cite{KM1996,AvM2001} separately in the times $\t_+$ and $\t_-$; this point is briefly reviewed in Section~\ref{paragraphtoda}. Our second main result is the identification (Theorem \ref{thm2} below) of the LUE partition function \eqref{Z} restricted to $\t_-=\0$ with another type of tau function, the \emph{modified Gaussian Unitary Ensemble (mGUE) partition function}, which has been introduced in \cite{DLYZ2016} as a generating function for Hodge integrals (see below), within the context of the \emph{Hodge-GUE correspondence} \cite{DY2017b,DLYZ2016,DY2019,Z2019,Z2019b,LYZZ2019}.

The mGUE partition function $\wt Z_N(\ss)$ is defined in \cite{DLYZ2016} starting from the \emph{even GUE} partition function
\be
\label{evenGUE}
Z_N^{\mathsf{even}}(\ss):=\int_{\H_N}\exp\tr \left(-\frac 12 X^2+\sum_{k\geq 1}s_k X^{2k}\right)\d X
,\qquad
\ss=(s_1,s_2,\dots)
\ee
which is the classical GUE partition function with couplings to odd powers set to zero. It is well known \cite{IKF1990} that \eqref{evenGUE} is a tau function of the \emph{discrete KdV} (also known as \emph{Volterra lattice}) hierarchy, which is a reduction of the Toda lattice hierarchy (see Section \ref{paragraphtoda} for a brief discussion of the Toda lattice hierarchy).
As far as only formal dependence on $N$ and on the times $\ss$ is concerned (see Section \ref{secformal} for more details) it is then argued in \cite{DLYZ2016} that the identity
\be
\label{factorization2}
\frac{Z_{N}^{\mathsf{even}}(\ss)}{(2\pi)^{N}\Vol(N)}=\wt Z_{N-\frac 12}(\ss)\wt Z_{N+\frac 12}(\ss),\qquad
\Vol(N):=\frac{\pi^{\frac{N(N-1)}2}}{\barnes(N+1)},
\ee
uniquely defines a function $\wt Z_N(\ss)$, termed mGUE partition function; in \eqref{factorization2} and throughout this paper, $\barnes(z)$ is the Barnes G-function, with the particular evaluation
\be
\barnes(N+1)=1!2!\cdots (N-1)!
\ee
for any integer $N>0$. With respect to the normalizations in \cite{DLYZ2016} we are setting $\epsilon\equiv 1$ for simplicity;  the dependence on $\epsilon$ can be restored by the scaling $N=x\epsilon$. In \cite{DY2019} a new type of tau function for the discrete KdV hierarchy is introduced and the mGUE partition function is identified with a particular tau function of this kind.

We have the following interpretation for the mGUE partition function.

\begin{shaded}
\begin{theorem}
\label{thm2}
The modified GUE partition function $\wt Z_N(\ss)$ in \eqref{factorization2}  is identified with the Laguerre partition function $ Z_N\left(\a;\t_+,\t_-\right)$ in  \eqref{Z} by the relation 
\be
\label{mGUEevenvsLaguerre}
\wt Z_{2N-\frac 12}(\ss)=C_N Z_N\left(\a=-\frac 12;\t_+,\t_-=\0\right)
\ee
where $\t_+,\ss$ are related by
\be
t_k=2^ks_k
\ee
and $C_N$ is an explicit constant depending on $N$ only\footnotemark;
\be
\label{constantmGUEevenvsLaguerre}
C_N=\frac{2^{N^2-\frac{3}{2}N+\frac{1}{4}}}{\pi^{\frac{N(N+1)}{2}}}\barnes(N+1).
\ee
\end{theorem}
\end{shaded}
\footnotetext{In the published version of this preprint the value of the constant $C_N$ is incorrect.}

The proof is given in Section \ref{secproofthm2}. 
Identity \eqref{mGUEevenvsLaguerre} can be recast as the following explicit relation;
\be
\wt Z_{2N-\frac 12}(\ss)=\frac{2^{-N+\frac{1}{4}}}{ \pi ^{\frac{N(N+1)}{2}}}\barnes(N+1)
\int_{H_N^+}\exp\tr\left(-\frac X2+\sum_{k\geq 1}s_kX^k\right)\frac{\d X}{\sqrt{\det X}}\;,
\ee
which is obtained from \eqref{mGUEevenvsLaguerre} by a change of variable $X\mapsto \frac X2$ in the LUE partition function.

Theorem \ref{thm2} provides a direct and new link (Corollary \ref{corollary2} below) between the monotone Hurwitz numbers in the expansion \eqref{hurwitzstrict} and \emph{special cubic Hodge integrals}. To state this result, let us denote $\Mgn$ the Deligne-Mumford moduli space of stable nodal Riemann surfaces, $\psi_1,\dots,\psi_n\in H^2\left(\Mgn,\Q\right)$ and $\kappa_j\in H^{2j}\left(\Mgn,\Q\right)$ ($j=1,2,\dots$) the Mumford-Morita-Miller classes, and $\L(\xi):=1+\lambda_1\xi+\cdots+\lambda_g\xi^g$ the Chern polynomial of the Hodge bundle, $\lambda_i\in H^{2i}\left(\Mgn,\Q\right)$. For the definition of these objects we refer to the literature, see e.g. \cite{Z2012} and references therein.

\begin{shaded}
\begin{corollary} \label{corollary2}
For any partition $\mu=(\mu_1,\dots,\mu_\ell)$ of length $\ell$ we have 
\be
\label{hodgemonotonehurwitzformula}
\sum_{g\geq 0}\epsilon^{2g-2} \mathscr{H}_{g,\mu}=2^\ell\sum_{\g\geq 0}\left(2\epsilon\right)^{2\g-2}\sum_{s=1}^{1-2\g+|\mu|-\ell}\left(\lambda+\frac \epsilon 2\right)^{2-2\g+|\mu|-\ell-s}\left(\lambda-\frac \epsilon 2\right)^s H_\g^>(\mu;s)
\ee
where
\begin{align}
\nonumber
\mathscr{H}_{g,\mu}:=\ & 2^{g-1} \sum_{m \geq 0} \frac{(\lambda - 1)^m}{m!} \int_{\overline{\mathcal{M}}_{g,\ell+m}} \Lambda^2(-1)\Lambda\left(\frac 12\right)\exp\left(-\sum_{d\geq 1}\frac{\kappa_d}d\right)\prod_{a=1}^\ell \frac{\mu_a\binom{2\mu_a}{\mu_a}}{1 - \mu_a\psi_a}
\\
\label{Hodgegmu}
&+\frac{\delta_{g,0}\delta_{\ell,1}}2 \left(\lambda-\frac {\mu_1}{\mu_1+1}\right){{2\mu_1}\choose \mu_1}+\frac{\delta_{g,0}\delta_{\ell,2}}2 \frac{\mu_1\mu_2}{\mu_1+\mu_2}{{2\mu_1}\choose \mu_1}{{2\mu_2}\choose \mu_2}.
\end{align}
\end{corollary}
\end{shaded}

The proof is given in Section \ref{secproofcorollaries}.
Note that $\mathscr{H}_{g,\mu}$ in \eqref{Hodgegmu} is a well defined formal power series in $\C[[\lambda-1]]$, as for dimensional reasons each coefficient of $(\lambda-1)^m$ in \eqref{Hodgegmu} is a finite sum of intersection numbers of Mumford-Morita-Miller and Hodge classes on the moduli spaces of curves.

Matching coefficients in \eqref{hodgemonotonehurwitzformula}, we obtain the following \emph{partial monotone ELSV-like formul\ae}, valid for all partitions $\mu=(\mu_1,\dots,\mu_\ell)$ of length $\ell$;
\begin{align}
\nonumber
\sum_{s\geq 1}H_0^>(\mu;s)={}&\frac 1{2^{\ell-1}}\int_{\overline{\mathcal{M}}_{0,\ell}}\exp\left(-\sum_{d \geq 1} \frac{\kappa_d}d\right)\prod_{a = 1}^{\ell} \frac{\mu_a{{2\mu_a}\choose{\mu_a}}}{1 - \mu_a\psi_a}
\\
\label{eqhurwitzg0} 
&+\delta_{\ell,1} \frac{1}{\mu_1+1}{{2\mu_1}\choose \mu_1}+\delta_{\ell,2} \frac{ \mu_1\mu_2}{\mu_1+\mu_2}{{2\mu_1}\choose \mu_1}{{2\mu_2}\choose \mu_2}
\end{align}
in genus zero (see also Example \ref{examplehodgeguegenus0}) and
\begin{align}
\nonumber
&\sum_{\g=0}^g 2^{4\g}\sum_{s\geq 1}\left[\sum_{p\geq 0}(-1)^p {{2-2\g+|\mu|-\ell-s} \choose p} {s \choose 2g-2\g-p}\right] H_\g^>(\mu;s)
\\
\label{eqhurwitzg1}
&\qquad
=2^{3g+1-\ell}
\int_{\overline{\mathcal{M}}_{g,\ell}} \Lambda^2(-1)\L\left(\frac 12\right)\exp\left(-\sum_{d \geq 1} \frac{\kappa_d}d\right)\prod_{a = 1}^{\ell} \frac{\mu_a{{2\mu_a}\choose{\mu_a}}}{1 - \mu_a\psi_a}
\end{align}
in higher genus $g\geq 1$. Note that the left  sides of \eqref{eqhurwitzg0} and \eqref{eqhurwitzg1} are finite sums.
The connection between Hurwitz numbers  and Hodge integrals, the so called ELSV formula,  was  introduced in  \cite{ELSV2001}, by T.Ekedahl, S. Lando, M. Shapiro, A. Vainshtein. 
Insertion of $\kappa$ classes in ELSV-type formul\ae\ for \emph{monotone} Hurwitz numbers have already been considered in the literature, e.g. in \cite{ALS2016} for single monotone Hurwitz numbers and in \cite{BGF2017} for \emph{orbifold} monotone Hurwitz numbers.

The relation between Hodge integrals and Hurwitz numbers expressed by Corollary \ref{corollary2} is obtained from Theorem~\ref{thm2} by re-expanding the topological expansion \eqref{hurwitzstrict}.
Indeed fixing $\a=-\frac 12$  implies that the parameter $c$  in \eqref{hurwitzstrict} is no longer independent of $N$ (\emph{soft-edge} limit) but actually scales as $c=1-\frac 1{2N}$ (\emph{hard-edge} limit). 
This explains why we cannot derive from the Hodge-GUE correspondence an expression in terms of Hodge integrals for each Hurwitz number in \eqref{hurwitzstrict}, but only an expression for a combination of Hurwitz numbers in different genera.

In particular, to obtain the formul\ae\ of Corollary \ref{corollary2} one has to re-expand the topological expansion \eqref{hurwitzstrict} in $N$ after the substitution $c=1-\frac 1{2N}$; that the result of this re-expansion, namely the right side of \eqref{hodgemonotonehurwitzformula}, involves only even powers of $\epsilon$ is a consequence of the invariance of positive LUE correlators under the involution $(N,\a)\mapsto(N+\a,-\a)$; this symmetry will be described below in Lemma \ref{lemmasymmetry}.
More concretely, this symmetry implies the symmetry of the positive LUE correlators under the involution $(N,c)\mapsto (Nc,c^{-1})$ which in view of \eqref{hurwitzstrict} is equivalent to the identity
\be
\label{symmetryhurwitz}
H_g^>(\mu;s)=H_g^>(\mu;2-2g+|\mu|-\ell(\mu)-s).
\ee
The above identity  implies that the small $\epsilon$ expansion on the right side of \eqref{hodgemonotonehurwitzformula} contains only even powers of $\epsilon$.
It is also possible to check the symmetry \eqref{symmetryhurwitz} by purely combinatorial arguments, see Rem.\ \ref{remcombinatorial}.

\begin{remark}
It is known that special cubic Hodge integrals are related to a $q$-deformation of the representation theory of the symmetric group \cite{OP2004}; it would be interesting to directly provide a link to the monotone Hurwitz numbers under consideration here.
\end{remark}

\subsection*{Organization of the paper}

In Section \ref{secproof} we prove Theorem \ref{thmmain}; a summary of the proof is given in the beginning of that section.
In Section \ref{secproofprop13} we analyze the formul\ae\ of Theorem \ref{thmmain} to prove Proposition \ref{prop13}.
In Section \ref{secLUE} we prove the identification of the mGUE and LUE partition functions, namely Theorem \ref{thm2}; then we recall the Hodge-GUE correspondence \cite{DLYZ2016} and we deduce Corollary \ref{corollary2}.
Finally, in the tables of App.\ \ref{apptable} we collect several connected correlators and weighted monotone double Hurwitz numbers, computed applying the formul\ae\ of Theorem \ref{thmmain}.

\section{Proof of Theorem \ref{thmmain}}\label{secproof}

In this section we prove our first main result, Theorem \ref{thmmain}. The proof combines two main ingredients; on one side the interpretation of the matrix integral \eqref{Z} as an \emph{isomonodromic tau function} \cite{BEH2006} and on the other side some algebraic manipulations of residue formul\ae\ introduced in \cite{BDY2016}. More in detail, we first introduce the relevant family of monic orthogonal polynomials and derive a compatible system of (\emph{monodromy-preserving}) ODEs in the parameters $\t$ (Proposition \ref{propdeformation}); throughout this section, in the interest of lighter notations, we set
\be
\t:=(\t_+,\t_-)=(\dots,t_{-2},t_{-1},t_1,t_2,\dots).
\ee
Such orthogonal polynomials reduce to monic Laguerre polynomials for $\t=\0$. With the aid of this system of deformations we then compute arbitrary derivatives of the LUE partition function \eqref{Z} in terms of formal residues of expressions that do not contain any derivative in $\t$ (Propositions \ref{propfirst}, \ref{propsecond} and \ref{propn}). Finally, the formul\ae\ of Theorem \ref{thmmain} are found by evaluation of these residues at $\t=\0$; the latter task is then to compute the asymptotic expansions of Cauchy transforms of Laguerre polynomials at zero and infinity (Propositions \ref{propexpansion+} and \ref{propexpansion-}). It is worth stressing at this point that the two formal series $R_\pm$ of \eqref{R+}-\eqref{R-} in Theorem \ref{thmmain} are actually asymptotic expansions of the \emph{same} analytic function at two different points.

As a preliminary to the proof, let us comment on the definition \eqref{Z} of the LUE partition function. Even though a formal approach is sufficient to make sense of the LUE partition function as a generating function, we shall also regard it as genuine analytic function of the times $\t$. In this respect let us point out that to make strict non-formal sense of \eqref{Z} one can assume that the vector of times is finite, namely that
\be
\label{truncate}
t_k\not=0\iff K_-\leq k\leq K_+,
\ee
and then, to ensure convergence of the matrix integral, that $\Re t_{K_-}<0$ for $K_-<0$ and $\Re t_{K_+}<\delta_{K_+,1}$ for $K_+>0$.

Though we have to assume in our computations that we have chosen such an arbitrary truncation of the times, this is inconsequential in establishing the formul\ae\ of Theorem \ref{thmmain}. More precisely, such truncation implies that \eqref{connectedmoments} holds true only as long as $K_+,K_-$ are large enough, and the formal generating functions $C_{r_+,r_-}$ (as it follows from our arguments, see Section \ref{subsecproof}) are manifestly independent of $K_\pm$ and are therefore obtained by a well-defined inductive limit $K_+\to\infty,K_-\to-\infty$.

Moreover, in \eqref{Z} the parameter $\a$ has to satisfy $\Re\a>-1$; even worse, in \eqref{connectedmoments} we have to assume that $\Re\a>-\sum_{i=1}^rk_i-1$ to enforce convergence of the matrix integral at $X=0$. This restriction can be lifted, if $\a$ is not an integer, by taking a suitable deformation of the contour of integration.
This caveat is crucial to us, as we shall need the formal expansion of the matrix $R(x)$ at all orders near $x=0$, compare with \eqref{R-}; the coefficients of this expansion are in general ill-defined for integer $\a$ (although truncated expansions are well defined if $\a$ is confined to suitable right half-planes). It is clear how to overcome these issues by the aforementioned analytic continuation, hence we do not dwell further on this point.

\subsection{Orthogonal polynomials and deformation equations}\label{secdeformation}

\subsubsection{Orthogonal polynomials}

Let $\pi_\ell^{(\a)}(x;\t)=x^\ell+\dots$ ($\ell\geq 0$) be the family of \emph{monic} orthogonal polynomials, uniquely defined by the property
\be
\label{orthogonality}
\int_0^{+\infty}\pi_\ell^{(\a)}(x;\t)\pi_{\ell'}^{(\a)}(x;\t)\e^{-V_\alpha(x;\t)}\d x=\delta_{\ell,\ell'}h_\ell(\t),\qquad \ell,\ell'\geq 0,
\ee
where 
\be
\label{Valpha}
V_\alpha(x;\t):=x-\alpha\log(x)-\sum_{k\not=0}t_k x^k,\qquad x>0.
\ee
For $\t=\0$ they essentially reduce to the \emph{generalized Laguerre polynomials} $L_\ell^{(\a)}(x)$; more precisely, denoting $\pi_{\ell}^{(\a)}(x):=\pi_{\ell}^{(\a)}(x;\t=\0)$ we have the identity
\be
\label{laguerrepolynomial}
\pi_\ell^{(\a)}(x):=(-1)^\ell\ell!L_\ell^{(\a)}(x)=\sum_{j=0}^\ell\frac{(-1)^{\ell-j}(\ell-j+1)_j(j+1+\a)_{\ell-j}}{j!}x^j,\qquad \ell\geq 0.
\ee
Using \emph{Rodrigues formula}
\be
\label{rodrigues}
\pi_\ell^{(\a)}(x)=(-1)^\ell x^{-\a}\e^x\left(\frac{\d^\ell}{\d x^\ell}\left(\e^{-x}x^{\a+\ell}\right)\right)
\ee
and integration by parts we obtain
\be
\int_0^{+\infty}x^k\pi_\ell^{(\a)}(x)\e^{-x}x^\a\d x=\int_0^{+\infty}\left(\frac{\d^\ell}{\d x^\ell}x^k\right)\e^{-x}x^{\a+\ell}\d x=
\begin{cases}
0, & k<\ell,
\\
\ell!\, \G(\a+\ell+1), & k=\ell.
\end{cases}
\ee
Hence the orthogonality property \eqref{orthogonality} for $\t=\0$ reads as
\be
\label{h0}
\int_0^{+\infty}\pi_\ell^{(\a)}(x)\pi_{\ell'}^{(\a)}(x)x^{\a}\e^{-x}\d x=h_\ell\delta_{\ell,\ell'},
\qquad h_\ell=\ell!\, \G(\a + \ell+1),
\ee
where $h_\ell=h_\ell(\t=\0)$.
For general $\t$ instead, the monic orthogonal polynomials $\pi_0^{(\a)}(\t)$, $\pi_1^{(\a)}(\t)$,  $\dots$, $\pi_{L-1}^{(\a)}(\t)$ exist whenever the moment matrix
\be
\left(m_{i+j}\right)_{i,j=0}^{L-1}, \qquad m_{\ell}:=\int_0^{+\infty}x^{\ell}\e^{-V_\alpha(x;\t)}\d x,
\ee
is non-degenerate. In the present case, their existence is ensured for real $\t$ by the fact that the moment matrix $\left(m_{i+j}\right)_{i,j=0}^{L-1}$ is positive definite.

By standard computations we have the following identity
\be
\label{h}
Z_N(\a;\t)=\frac{\pi^{\frac {N(N-1)}2}}{\barnes(N+1)} \prod_{\ell=0}^{N-1}h_\ell(\t)
\ee
where $h_\ell(\t)$ are defined by \eqref{orthogonality}.

\subsubsection{Connection with Toda lattice hierarchy}\label{paragraphtoda}
It is well known that the monic orthogonal polynomials $\pi_\ell^{(\a)}(x;\t)$ satisfy a three term recurrence relation
\be
x\pi_\ell^{(\a)}(x;\t)=\pi_{\ell+1}^{(\a)}(x;\t)+v_\ell^{\alpha}(\t)\pi_{\ell}^{(\a)}(x;\t)+w_\ell^\alpha(\t) \pi_{\ell-1}^{(\a)}(x;\t).
\ee
That is, the orthogonal polynomials are eigenvectors of the second order difference operator
\be
\left( L\, \psi\right)_\ell = \psi_{\ell+1} +v^{\alpha}_\ell \psi_\ell + w^{\alpha}_\ell \psi_{\ell-1}.
\ee
The corresponding half-infinite tri-diagonal matrix, also denoted $L=\left( L_{ij}\right)$, $i,j\geq 0$, takes the form 
\be
L=\left(
\begin{array}{ccccc}
v^\a_0 & 1 & 0 & 0 &\cdots
\\
w^\a_1 & v^\a_1 & 1 & 0 & \cdots
\\
0 & w^\a_2 & v^\a_2 & 1 &\cdots
\\
0 & 0 & w^\a_3 & v^\a_3 &\cdots
\\
\vdots & \vdots & \vdots & \vdots & \ddots
\end{array}
\right).
\ee
It is a standard fact that $L$, and therefore the coefficients $v^{\alpha}_n(\t)$ and $w^{\alpha}_n(\t)$ evolve with respect to positive times $\t_+=(t_1,t_2,\dots)$, for any fixed $\t_-=(t_{-1},t_{-2},\dots)$, according to the \emph{Toda lattice hierarchy} \cite{DLT1989,KM1996,AvM2001,DY2017,CdlI2018} 
\be
\frac{\pa L}{\pa t_k}=\left[\left(L^k\right)_+,L\right]
\ee
where for any matrix $P$, $P_+$ denotes the lower triangular part of $P$, i.e.\ the matrix with entries
\be
(P_+)_{ij}:=\begin{cases}P_{ij} & {\rm if}~i\geq j \\ 0 & \text{\rm if}~i<j \end{cases}
\ee
where $P_{ij}$ are the entries of $P$. Setting $\t_-=\0$, we can also write the initial data of the Toda hierarchy as
\be
v^{\alpha}_\ell(\t_+=\t_-=\0)=2\ell+1+\alpha,\quad w^{\alpha}_\ell(\t_+=\t_-=\0)=\ell(\ell+\alpha)
\ee
that are the recurrence coefficients for the monic generalized Laguerre polynomials \eqref{laguerrepolynomial}. Moreover, it is well known, see loc.\ cit., that $Z_N(\a;\t_+,\t_-=\0)$ is the Toda lattice tau function corresponding to this solution.

It can be observed that the evolution with respect to the negative times $\t_-=(t_{-1},t_{-2},\dots)$ is also described by a Toda lattice hierarchy and a simple shift in $\a$. More precisely, we claim that $Z_N(\a-2N,\t_+=\0,\t_-)$ is also a Toda lattice tau function, with a \emph{different} initial datum; namely, it is associated with the tri-diagonal matrix $\wt L$ satisfying the Toda hierarchy
\be
\label{todainverse}
\frac{\pa \wt L}{\pa t_{-k}}=\left[\left(\wt L^k\right)_+,\wt L\right]
\ee
constructed as above from the three term recurrence of monic orthogonal polynomials, this time with respect to the measure
\be
\label{measureinverse}
\exp\left(-\frac 1x+\sum_{k>0}t_{-k}x^k\right)\frac{\d x}{x^\a}
\ee
on $(0,+\infty)$. To see it, let us rewrite
\begin{align}
\nonumber
Z_N(\a;\t_+=\0,\t_-)&=\int_{\H_N^+}{\det}^\a X\exp\tr\left(-X+\sum_{k<0}t_kX^k\right) \d X
\\
&=\int_{\H_N^+}{\det}^{-\a} \wt X\exp\tr\left(-\wt X^{-1}+\sum_{k>0}t_{-k}\wt X^{k}\right) \d\left(\wt X^{-1}\right)
\end{align}
where we perform the change of variable $\wt X=X^{-1}$, which is a diffeomorphism of $\H_N^+$.
The Lebesgue measure \eqref{lebesgue} can be rewritten (on the full-measure set of semisimple matrices) as 
\be
\d X=\d U\prod_{i<j}(x_i-x_j)^2\d x_1\cdots\d x_N
\ee
where $\d U$ is a suitably normalized Haar measure on $\mathrm{U}(N)/(\mathrm{U}(1))^N$ and $x_1,\dots,x_N$ are the eigenvalues of $X$. Therefore the measure transforms as
\begin{align}
\nonumber
\d \wt X&=\d U \prod_{i<j}\left(\frac 1{x_i}-\frac 1{x_j}\right)^2\d\left(\frac 1{x_1}\right)\cdots\d\left(\frac 1{x_N}\right)
\\
&
=\frac {\d U}{(x_1\cdots x_N)^{2N}}\prod_{i<j}(x_i-x_j)^2\d x_1\cdots\d x_N=\frac{\d X}{{\det}^{2N} X},
\end{align}
yielding
\be
\d X=\frac{\d \wt X}{{\det}^{2N}\wt X}.
\ee
Summarizing, we have
\be
\label{finalinverse}
Z_N(\a;\t_+=\0,\t_-)=\int_{\H_N^+}\frac{\exp\tr\left(-\wt X^{-1}+\sum_{k>0}t_{-k}\wt X^{k}\right)}{{\det}^{\a+2N} \wt X} \d \wt X
\ee
and the standard arguments of loc.\ cit.\ now apply to the matrix integral $Z_N(\a-2N;\t_+=\0,\t_-)$ to show that it is indeed the Toda lattice tau function associated with the solution $\wt L$.

For our purposes, we need to describe the simultaneous dependence on $\t_+$ and $\t_-$; this is achieved by the zero-curvature condition \eqref{zerocurvatureA} of the system of compatible ODEs \eqref{deformation} which we now turn our attention to.

\subsubsection{Cauchy transform and deformation equations}
Let us denote by
\be
\label{cauchytransform}
\wh\pi_\ell^{(\a)}(x;\t):=\frac 1{2\pi\i}\int_0^{+\infty}\pi_\ell^{(\a)}(\xi;\t)\e^{-V_{\alpha}(\xi;\t)}\frac{\d \xi}{\xi-x}
\ee
the \emph{Cauchy transforms} of the orthogonal polynomials $\pi_\ell^{(\a)}(x;\t)$.  
Then, for fixed $N$ introduce the following $2\times 2$ matrix
\be
\label{Y}
Y(x;\t):=
\begin{pmatrix}
\pi^{(\a)}_N(x;\t) & \wh\pi^{(\a)}_N(x;\t)\\
-\frac{2\pi\i}{h_{N-1}(\t)}\pi_{N-1}^{(\a)}(x;\t) & -\frac{2\pi\i}{h_{N-1}(\t)}\wh\pi_{N-1}^{(\a)}(x;\t)
\end{pmatrix}
\ee
where, for the interest of clarity, we drop the dependence on $N,\a$. The matrix $Y(x;\t)$ was introduced in the seminal paper \cite{IKF1990} to study  the general connection between orthogonal polynomials and random matrix models. The rest of this section follows from \cite{IKF1990}. The matrix \eqref{Y} solves the following Riemann-Hilbert problem for orthogonal polynomials; it is analytic for $x\in\C\setminus[0,\infty)$ and continuous up to the boundary $(0,\infty)$ where it satisfies the  \emph{jump condition}
\be
\label{jumpY}
Y_+(x;\t)=Y_-(x;\t)
\begin{pmatrix}
1 & \e^{-V_\a(x;\t) }
\\
0 & 1
\end{pmatrix}, \qquad x\in (0,\infty)
\ee
where $Y_\pm(x;\t)=\lim_{\epsilon\to 0_+}Y(x\pm \i\epsilon;\t)$. Moreover, at the endpoints $x=\infty,0$ we have
\begin{align}
\label{expinfty}
Y(x;\t)&\sim\left(\1+\O(x^{-1})\right)x^{N \s_3},&& \!\!\!\!\!\!\!\!\!\!\!\!\!\!\!\!\!\!\!\! x\to\infty,
\\
\label{exp0}
Y(x;\t)&\sim G_0(\t)\left(\1+\O(x)\right),	  	&& \!\!\!\!\!\!\!\!\!\!\!\!\!\!\!\!\!\!\!\! x\to 0,
\end{align}
within the sector $0<\arg x<2\pi$; the matrix $G_0(\t)$ in \eqref{exp0} is independent of $x$ and it is invertible (actually it has unit determinant, as we now explain).

The jump matrix in \eqref{jumpY} has unit determinant, hence $\det Y(x;\t)$ is analytic for all complex $x$ but possibly for isolated singularities at $x=0,\infty$; however, $\det Y(x;\t)\sim 1$ when $x\to \infty$, see \eqref{expinfty}, and is bounded as $x\to 0$, see \eqref{exp0}; therefore we conclude by the Liouville theorem that $\det Y(x;\t)\equiv 1$ identically.

Introduce the $2\times 2$ matrix
\be
\label{Psi}
\Psi(x;\t):=Y(x;\t)\exp\left( -V_{\alpha}(x;\t)\frac{\s_3}2\right).
\ee
Here we choose the branch of the logarithm appearing in $V_\a(x;\t)$, see \eqref{Valpha}, analytic for $x\in\C\setminus[0,\infty)$ satisfying $\lim_{\epsilon\to 0_+}\log(x+\i\epsilon)\in\R$; to be consistent with \eqref{Valpha}, we shall identify $V_\a(x;\t)$, without further mention, with $V_{\a,+}(x;\t)=\lim_{\epsilon\to 0_+}V_\a(x+\i\epsilon;\t)$ whenever $x>0$.

Accordingly, $\Psi(x;\t)$ is analytic for $x\in\C\setminus[0,\infty)$.

\begin{proposition}\label{propdeformation}
The matrix $\Psi$ in \eqref{Psi} satisfies a compatible system of linear $2\times 2$ matrix ODEs with rational coefficients;
\be
\label{deformation}
\frac{\pa\Psi(x;\t)}{\pa x}=\mathcal{A}(x;\t)\Psi(x;\t),\qquad \frac{\pa\Psi(x;\t)}{\pa t_k}=\Omega_k(x;\t)\Psi(x;\t),\quad k\not=0.
\ee
In particular, for $k>0$, the matrices $\Omega_k(x;\t)$ are polynomials in $x$ of degree $k$, whilst for $k<0$ they are polynomials in $x^{-1}$ of degree $|k|$ without constant term; more precisely, they admit the representations
\be
\label{Omega}
\Omega_k(x;\t)=\res\xi \left(Y(\xi;\t)\frac{\s_3}2 Y^{-1}(\xi;\t)\frac{\xi^k\d\xi}{x-\xi}\right)
\ee
where $\res\xi$ denotes $\res{\xi=\infty}$ when $k>0$ and $\res{\xi=0}$ when $k<0$. On the other hand, $\mathcal{A}(x;\t)$ is a Laurent polynomial in $x$, provided times are truncated according to \eqref{truncate}.
\end{proposition}

\begin{proof}
We note that \eqref{jumpY} implies the following jump condition for the matrix $\Psi$, with a \emph{constant} jump matrix;
\be
\Psi_+(x;\t)=\Psi_-(x;\t)\begin{pmatrix}
\e^{-\i\pi\a} & \e^{-\i\pi\a}
\\
0 & \e^{\i\pi\a}
\end{pmatrix}
,\qquad x\in (0,\infty).
\ee
Here $\Psi_\pm(x;\t)=\lim_{\epsilon\to 0_+}\Psi(x\pm \i\epsilon;\t)$; to prove this relation we observe that the branch of the logarithm we are using satisfies $\log_+(x)=\log_-(x)-2\pi\i$ for $x\in(0,\infty)$ and so $V_{\a,+}(x;\t)=V_{\a,-}(x;\t)+2\i\pi\a$, with a similar notation for the $\pm$-boundary values along $(0,\infty)$.
Hence all derivatives of $\Psi$ satisfy the same jump condition, with the same jump matrix. 
It follows that the ratios $\mathcal{A}:=\frac{\pa\Psi}{\pa x}\Psi^{-1}$ and $\Omega_k:=\frac{\pa\Psi}{\pa t_k}\Psi^{-1}$ (for all $k\not=0$) are regular along the positive real axis; however they may have isolated singularities at $x=0$ and at $x=\infty$.
Let us start from $\Omega_k$ for $k>0$. In such case, it follows from \eqref{expinfty} and \eqref{exp0} that $\Omega_k$ has a polynomial growth at $x=\infty$ and it is regular at $x=0$:
\be
\Omega_k=\frac{\pa Y(x;\t)}{\pa t_k}Y^{-1}(x;\t)+Y(x;\t)\frac{\s_3}2 Y^{-1}(x;\t)x^k\sim
\begin{cases}
\frac{\s_3}2x^k+\O\left(x^{k-1}\right),& x\to\infty\\
\O(1),& x\to 0.
\end{cases}
\ee
From the Liouville theorem we conclude that $\Omega_k$ for $k>0$ is a polynomial, which therefore equals the polynomial part of its expansion at $x=\infty$, which is computed as in \eqref{Omega}, since at $x=\infty$ the term $\frac{\pa Y}{\pa t_k}Y^{-1}=\O(x^{-1})$ does not contribute to the polynomial part of the expansion. The statement for $\Omega_k$ for $k<0$ follows along similar lines. 
Likewise, $\mathcal{A}(x;\t)$ in \eqref{deformation} has a polynomial growth at $x=\infty$ and a pole at $x=0$ and therefore it is a Laurent polynomial.
\end{proof}

The compatibility of \eqref{deformation} is ensured by the existence of the solution $\Psi(x;\t)$. In particular this implies the \emph{zero curvature equations}
\be
\label{zerocurvatureA}
\frac{\pa\mathcal{A}}{\pa t_k}-\frac{\pa \Omega_k}{\pa x}=[\Omega_k,\mathcal{A}],\qquad k\not=0.
\ee

\begin{remark}\label{remtraceless}
Since the determinants of $Y(x;\t)$ and $\Psi(x;\t)$ are identically equal to $1$, it follows that $\Omega_k(x;\t)$ and $\mathcal{A}(x;\t)$, introduced in \eqref{deformation}, are traceless.
\end{remark}

We end this paragraph by considering the restriction $\t=\0$.
The matrix $\Psi(x):=\Psi(x;\t=\0)$ is obtained from the Laguerre polynomials \eqref{laguerrepolynomial}.
The matrix $\mathcal{A}(x)=\frac{\pa\Psi(x)}{\pa x}\Psi(x)^{-1}$ takes the form
\be
\label{Amatrix}
\mathcal{A}(x):=\mathcal{A}(x;\t=\0)=-\frac 12 \s_3+
\frac 1x
\begin{pmatrix}
N+\frac\a 2 & -\frac{h_N}{2\pi\i}
\\ 
\frac{2\pi\i}{h_{N-1}} & -N-\frac\a 2
\end{pmatrix}
\ee
which has a \emph{Fuchsian} singularity at $x=0$ and an \emph{irregular} singularity of \emph{Poincar\'{e} rank} $1$ at $x=\infty$. 

\begin{remark}\label{remfuchsian}
The \emph{Frobenius indices} of \eqref{Amatrix} at $x=0$ are $\pm\frac{\a}2$, and so the Fuchsian singularity $x=0$ is \emph{non-resonant} if and only if $\a$ is not an integer. It is worth pointing out that the monodromy matrix $\frac{\a}2\s_3$ at $x=0$ is preserved under the $\t$-deformation \eqref{deformation}.
\end{remark}

\subsection{Residue formul\ae\ for correlators}\label{sectioncorrelations}

\subsubsection{One-point correlators}
The general type of formul\ae\ of Proposition \ref{propfirst} below first appeared in \cite{BEH2006}, where the authors consider a very general case. Such formul\ae\ identify the LUE partition function with the \emph{isomonodromic tau function} \cite{JMU1981} of the \emph{monodromy-preserving} deformation system \eqref{deformation}.
The starting point for the following considerations is the representation \eqref{h} for the LUE partition function \eqref{Z}.

\begin{proposition}
\label{propfirst}
Logarithmic derivatives of the LUE partition function admit the following expression in terms of \emph{formal} residues;
\be
\frac{\pa\log Z_N(\a;\t)}{\pa t_k}=-\res{x}\tr\left(Y^{-1}(x;\t) \frac{\pa Y(x;\t)}{\pa x} \frac{\s_3}2\right)x^k\d x
\ee
where the symbol $\res{x}$ denotes $\res{x=\infty}$ when $k>0$ and $\res{x=0}$ when $k<0$.
\end{proposition}

\begin{proof}
For the proof we follow the lines of \cite{CGM2015}. First, differentiate the orthogonality relation \eqref{orthogonality}
\be
\frac{\pa h_\ell(\t)}{\pa t_k}=\int_0^{+\infty}\left(\pi_\ell^{(\a)}(x;\t)\right)^2x^k\e^{-V_{\alpha}(x;\t)}\d x
\ee
and recall the \emph{confluent Christoffel-Darboux} formula for orthogonal polynomials
\begin{align}
\nonumber
\sum_{\ell=0}^{N-1}\frac{\left(\pi_\ell^{(\a)}(x;\t)\right)^2}{h_\ell}&=
\frac 1{h_{N-1}}\left(\pi_{N-1}^{(\alpha)} (x;\t)\frac{\pa\pi_{N}^{(\a)}(x;\t)}{\pa x}-\frac{\pa\pi^{(\a)}_{N-1}(x;\t)}{\pa x}\pi_{N}^{(\a)}(x;\t)\right)
\\
&=\frac 1{2\pi\i}\left(Y^{-1}(x;\t)\frac{\pa Y(x;\t)}{\pa x}\right)_{21}
\end{align}
where in the last step one uses $\det Y(x;\t)\equiv 1$.
Now, omitting the dependence on $x,\t$ in the rest of the proof for the sake of brevity, it can be checked that the jump relation \eqref{jumpY} implies 
\be
\tr\left(Y_{+}^{-1}\frac{\pa Y_{+}}{\pa x}\frac{\s_3}2 \right)=\tr\left(Y_{-}^{-1}\frac{\pa Y_{-}}{\pa x}\frac{\s_3}2\right)
-\left(Y^{-1}\frac{\pa Y}{\pa x}\right)_{21}\e^{-V_\a}.
\ee
Therefore, starting from \eqref{h}, we compute
\begin{align}
\nonumber
\frac{\pa\log Z_N(\a;\t)}{\pa t_k}=\sum_{\ell=0}^{N-1}\frac{1}{h_\ell}\frac{\pa h_\ell}{\pa t_k}
&=\sum_{\ell=0}^{N-1}\int_0^{+\infty}\frac{\left(\pi_\ell^{(\a)}\right)^2}{h_\ell}x^k\e^{-V_{\alpha}}\d x
\\
\label{2}
&=\frac 1{2\pi\i}\int_0^{+\infty}\tr\left[\left(Y_{-}^{-1}\frac{\pa Y_{-}}{\pa x}-Y_{+}^{-1}\frac{\pa Y_{+}}{\pa x}\right)\frac{\s_3}2 \right]x^{k}\d x.
\end{align}
Such an integral of a jump can be performed by a residue computation. First of all, note that despite $x^k\tr\left(Y^{-1}\frac{\pa Y}{\pa x}\frac{\sigma_3}{2}\right)$ is not analytic at $x=\infty$, it has a large $x$ asymptotic expansion given by
\be
x^k\tr\left(Y^{-1}\frac{\pa Y}{\pa x}\frac{\sigma_3}{2}\right)=\sum_{j=-1}^{k-2}c_jx^j+\O\left(\frac{1}{x^2}\right)
\ee
for any $k\in \mathbb{N}$, where $-c_{-1}$ is, by definition, the \emph{formal} residue at infinity of $ x^k\tr\left(Y^{-1}\frac{\pa Y}{\pa x}\frac{\sigma_3}{2}\right)$. Then, recalling our choice for the branch of the logarithm and using contour deformation, we can express \eqref{2} as
\be
\frac{\pa\log Z_N(\a;\t)}{\pa t_k}
=-\res{x=0}\tr \left(Y^{-1}\frac{\pa Y}{\pa x}\frac{\s_3}2\right)  x^{k}\d x-\res{x=\infty}\tr \left(Y^{-1}\frac{\pa Y}{\pa x}\frac{\s_3}2\right) x^{k}\d x,
\ee
the residues being intended in the formal sense explained above. Finally, the proof is complete by noting that for $k>0$ (resp. $k<0$) the formal residue at $x=0$ (resp. $x=\infty$) vanishes. 
\end{proof}

For later convenience let us slightly rewrite the result of the above proposition. To this end introduce the matrix
\be
\label{R}
R(x;\t):=Y(x;\t)\E_{11}Y^{-1}(x;\t),
\ee
denoting $\E_{11}:=\begin{pmatrix} 1 & 0 \\ 0 & 0 \end{pmatrix}$ from now on.

\begin{corollary}\label{corollaryfirst}
We have
\be
\label{312}
\frac{\pa\log Z_N(\a;\t)}{\pa t_k}=-\res{x}\left(\tr\left(\mathcal{A}(x;\t)R(x;\t)\right)+\frac 12\dfrac{\partial}{\partial x} V_{\alpha}(x;\t)\right)x^k\d x,
\ee
where $R(x;\t)$ is introduced in \eqref{R} and again $\res{x}$ denotes $\res{x=\infty}$ when $k>0$ and $\res{x=0}$ when $k<0$.
\end{corollary}

\begin{proof}
We have  from \eqref{Psi} and \eqref{deformation}
\be
\frac{\pa}{\pa x}Y(x;\t)=\mathcal{A}(x;\t)Y(x;\t)+Y(x;\t)\frac {\s_3}2\frac{\pa}{\pa x}V_{\alpha}(x;\t)
\ee
so that
\begin{align}
\nonumber
\tr\left(Y^{-1}(x;\t) \frac{\pa Y(x;\t)}{\pa x}\frac{\s_3}2\right)&=\tr\left(Y^{-1}(x;\t)\mathcal{A}(x;\t)Y(x;\t)\frac{\s_3}2\right)+\frac 12 \frac{\pa}{\pa x}V_{\alpha}(x;\t)
\\
&=\tr\left(\mathcal{A}(x;\t)R(x;\t)\right)+\frac 12\frac{\pa}{\pa x}V_{\alpha}(x;\t),
\end{align}
where in the last step we have used that
\be
\tr\left(Y^{-1}(x;\t)\mathcal{A}(x;\t)Y(x;\t)\frac{\s_3}2\right)=\tr\left(Y^{-1}(x;\t)\mathcal{A}(x;\t)Y(x;\t)\E_{11}\right)
=\tr\left(\mathcal{A}(x;\t)R(x;\t)\right),
\ee
where the first equality follows from $\tr\mathcal{A}(x;\t)=0$ and the second one from the cyclic property of the trace and the definition \eqref{R}.
\end{proof}

\subsubsection{Multipoint connected correlators}

We first consider two-point connected correlators.

\begin{proposition}
\label{propsecond}
For every nonzero integers $k_1,k_2$ we have
\be
\frac{\pa^2\log Z_N(\a;\t)}{\pa t_{k_2}\pa t_{k_1}}=\res{x_1}\res{x_2}\frac{\tr(R(x_1;\t)R(x_2;\t))-1}{(x_1-x_2)^2}x_1^{k_1}x_2^{k_2}\d x_1\d x_2,
\ee
where the symbol $\res{x_i}$ denotes $\res{x_i=\infty}$ (resp. $\res{x_i=0}$) if $k_i>0$ (resp. $k_i<0$).
\end{proposition}

\begin{proof}
From \eqref{312} we have
\be
\frac{\pa\log Z_N(\a;\t)}{\pa t_{k_1}}=-\res{x_1}\left(\tr\left(\mathcal{A}(x_1;\t)R(x_1;\t)\right)+\frac 12\dfrac{\partial}{\partial x_1} V_{\alpha}(x_1;\t)\right)x_1^k\d x_1.
\ee
Let us take one more time-derivative
\be
\label{equazionezero}
\frac{\pa^2\log Z_N(\a;\t)}{\pa t_{k_2}\pa t_{k_1}}=-\res{x_1}\left(\mathrm{tr}\left(\frac{\pa\mathcal{A}(x_1;\t)}{\pa t_{k_2}}R(x_1;\t)+\frac{\pa R(x_1;\t)}{\pa t_{k_2}}\mathcal{A}(x_1;\t)\right)-\frac 12 k_2x^{k_2-1}_1\right)x_1^{k_1}\d x_1
\ee
and note that, using \eqref{zerocurvatureA} and $\pa_{t_k}R(x;\t)=[\Omega_k(x;\t),R(x;\t)]$
\be
\label{equazione}
\tr\left(\frac{\pa \mathcal{A}(x_1;\t)}{\pa t_{k_2}}R(x_1;\t)+\frac{\pa R(x_1;\t)}{\pa t_{k_2}}\mathcal{A}(x_1;\t)\right)
=\tr\left(\frac{\pa \Omega_{k_2}(x_1;\t)}{\pa x_1}R(x_1;\t)\right).
\ee
Now let us write $\Omega_{k_2}$ from \eqref{Omega} as
\be
\Omega_{k_2}(x_1;\t)=\res{x_2}\left(Y(x_2;\t)\frac{\s_3}2Y^{-1}(x_2;\t)\frac{x_2^{k_2}}{x_1-x_2}\right)\d x_2=\res{x_2}R(x_2;\t)\frac{x_2^{k_2}\d x_2}{x_1-x_2}-\frac \1 2 \res{x_2}\frac{x_2^{k_2}\d x_2}{x_1-x_2}
\ee
yielding
\be
\label{equazione2}
\frac{\pa \Omega_{k_2}(x_1;\t)}{\pa x_1}=-\res{x_2}R(x_2;\t)\frac{x_2^{k_2}\d x_2}{(x_1-x_2)^2}+\frac \1 2 \res{x_2}\frac{x_2^{k_2}\d x_2}{(x_1-x_2)^2}.
\ee
Finally, the identity
\be
\label{equazione3}
-k_2x_1^{k_2-1}=\res{x_2}\frac{x_2^{k_2}}{(x_1-x_2)^2}\d x_2
\ee
holds true irrespectively of the sign of $k_2$, and the proof is completed by inserting \eqref{equazione}, \eqref{equazione2} and \eqref{equazione3} in \eqref{equazionezero}, along with $\tr R(x;\t)\equiv 1$.
\end{proof}

To compute higher order logarithmic derivatives of the LUE partition function, let us introduce the functions
\be
\label{S}
S_r(x_1,\dots,x_r;\t):=-\sum_{(i_1,\dots,i_r)\in\mathcal{C}_r}\frac{\tr\left(R(x_{i_1};\t)\cdots R(x_{i_r};\t)\right)}{(x_{i_1}-x_{i_2})\cdots(x_{i_{r-1}}-x_{i_r})(x_{i_r}-x_{i_1})}-\frac{\delta_{r,2}}{(x_1-x_2)^2},
\ee
where, as explained in the statement of Theorem \ref{thmmain}, the sum extends over cyclic permutations of $\{1,\dots,r\}$. Due to the cyclic invariance of the trace and of the polynomial $(x_1-x_2)\cdots(x_r-x_1)$, it follows that $S_r(x_1,\dots,x_r;\t)$ is symmetric in $x_1,\dots,x_r$.

The following proof is reported for the sake of completeness; it has appeared in the literature several times, e.g. see \cite{BDY2016,DY2017,BDY2018,BR2019}. The only slight difference here is that we consider two different set of times and correspondingly the residues are taken at two different points.

\begin{proposition}
\label{propn}
For every $r\geq 2$ we have
\be
\label{residuenpoint}
\frac{\pa^r\log Z_N(\a;\t)}{\pa t_{k_r}\cdots\pa t_{k_1}}=(-1)^{r}\res{x_1}\cdots\res{x_r}S_r(x_1,\dots,x_r;\t)x_1^{k_1}\cdots x_r^{k_r}\d x_1\cdots\d x_r
\ee
where, as above, the symbol $\res{x_i}$ denotes $\res{x_i=\infty}$ (resp. $\res{x_i=0}$) if $k_i>0$ (resp. $k_i<0$).
\end{proposition}

\begin{proof}
We have
\be
\frac{\pa}{\pa t_k}R(x;\t)=[\Omega_k(x;\t),R(x;\t)]=\res{\xi}\frac{[R(\xi;\t),R(x;\t)]}{x-\xi}\xi^k\d\xi
\ee
where we have used \eqref{Omega} and $\res{\xi}$ denotes the formal residue at $\xi=\infty$ if $k>0$ or the formal residue at $\xi=0$ if $k<0$. Hence we compute
\be
\frac{\pa S_r(x_1,\dots,x_r;\t)}{\pa t_k}=-\res{\xi}\sum_{(i_1,\dots,i_r)\in\mathcal{C}_r}\sum_{j=1}^r\frac{\tr\left(R(x_{i_1};\t)\cdots[R(\xi;\t),R(x_{i_j};\t)]\cdots R(x_{i_r};\t)\right)}{(x_{i_1}-x_{i_2})\cdots(x_{i_r}-x_{i_1})(x_{i_j}-\xi)}\xi^k\d\xi.
\ee
Expanding the commutator $[R(\xi;\t),R(x_{i_j};\t)]=R(\xi;\t)R(x_{i_j};\t)-R(x_{i_j};\t)R(\xi;\t)$, we note that each term involving the expression
\be
\tr\left(R(x_{i_1};\t)\cdots R(\xi;\t)R(x_{i_j};\t)\cdots R(x_{i_r};\t)\right)
\ee
appears twice, but with different denominators; collecting these terms gives
\begin{align}
\nonumber
&-\res{\xi}\sum_{(i_1,\dots,i_r)\in\mathcal{C}_r}\sum_{j=1}^r\frac{\tr\left(R(x_{i_1};\t)\cdots R(\xi;\t)R(x_{i_j};\t)\cdots R(x_{i_r};\t)\right)}{(x_{i_1}-x_{i_2})\cdots(x_{i_r}-x_{i_1})}\left(\frac 1{x_{i_j}-\xi}-\frac 1{x_{i_{j-1}}-\xi}\right)\xi^k\d\xi
\\
\nonumber
&\qquad=\res{\xi}\sum_{(i_1,\dots,i_r)\in\mathcal{C}_r}\sum_{j=1}^r\frac{\tr\left(R(x_{i_1};\t)\cdots R(\xi;\t)R(x_{i_j};\t)\cdots R(x_{i_r};\t)\right)}{(x_{i_1}-x_{i_2})\cdots(x_{i_{j-1}}-\xi)(\xi-x_{i_j}) \cdots (x_{i_r}-x_{i_1})}\xi^k\d \xi
\\
&\qquad=-\res{\xi}S_{r+1}(x_1,\dots,x_r,\xi)\xi^k\d\xi,
\end{align}
where the index $j$ in the internal summation is taken ${\rm mod}\, r$, namely $i_0:=i_r$.

Summarizing, we have shown that for all $r\geq 2$
\be
\frac{\pa S_r(x_1,\dots,x_r;\t)}{\pa t_k}=-\res{\xi}S_{r+1}(x_1,\dots,x_r,\xi)\xi^k\d\xi
\ee
and the proof now follows by induction on $r\geq 2$, the base $r=2$ being established in Proposition \ref{propsecond}.
\end{proof}

\begin{remark} \label{Regularity}
The functions $S_r(x_1,\dots,x_r)$ are regular along the diagonals $x_i=x_j$. In the case $r=2$ this can be seen from the fact that
\be
\tr(R^2(x;\t))\equiv 1,
\ee
hence the function $\tr\left(R(x_1;\t)R(x_2;\t)\right)-1$ is symmetric in $x_1$ and $x_2$  and vanishes for $x_1=x_2$. Therefore  the zero on the diagonal  $x_1=x_2$ is of order at least $2$ and so $S_2(x_1,x_2)$ is regular at $x_1=x_2$. For $r\geq 3$ instead we can reason as follows; since $S_r$ is symmetric, we can focus on the case $x_{r-1}=x_r$, and the only addends in $S_r$ which are singular at $x_{r-1}=x_r$ are those coming from the $r$-cycles $(i_1,\dots,i_{r-2},r-1,r)$ and $(i_1,\dots,i_{r-2},r,r-1)$, namely the terms
\begin{align}
\nonumber
\sum_{(i_1,\dots,i_{r-2},r-1,r)}&\frac{\tr\left(R(x_{i_1};\t)\cdots R(x_{i_{r-2}};\t)R(x_{r-1};\t)R(x_{r};\t)\right)}{(x_{i_1}-x_{i_2})\cdots(x_{i_{r-2}}-x_{r-1})(x_{r-1}-x_r)(x_r-x_{i_1})}
\\
&\qquad
+\sum_{(i_1,\dots,i_{r-2},r,r-1)}\frac{\tr\left(R(x_{i_1};\t)\cdots R(x_{i_{r-2}};\t)R(x_r;\t)R(x_{r-1};\t)\right)}{(x_{i_1}-x_{i_2})\cdots(x_{i_{r-2}}-x_{r})(x_r-x_{r-1})(x_{r-1}-x_{i_1})}
\end{align}
and this expression is manifestly regular at $x_{r-1}=x_r$.

In particular, the order in which residues are carried out in \eqref{residuenpoint} is immaterial.
\end{remark}

Finally we remark that it would be interesting to extend the above formulation to other matrix ensembles like the GOE, see e.g. \cite{DG2009}.

\subsection{Asymptotic expansions and proof of Theorem \ref{thmmain}}\label{subsecproof}

To compute LUE correlators we have to set $\t=\0$ in the residue formul\ae\ of Corollary \ref{corollaryfirst}, and of Propositions \ref{propsecond} and \ref{propn}. To this end we now consider
\be
R(x):=R(x;\t=\0)
\ee
where $R(x;\t)$ is introduced in \eqref{R} and compute explicitly series expansions as $x\to\infty,0$. We start with the expansion as $x\to\infty$.

\begin{proposition}\label{propexpansion+}
The matrix $R(x)$ admits the asymptotic expansion
\be
\label{expR+}
TR(x)T^{-1}\sim R_+(x),\qquad x\to\infty
\ee
uniformly within the sector $0<\arg x<2\pi$. Here $R_+$ is the formal series introduced in the beginning of this paper, see \eqref{R+}, and $T$ is defined as
\be
\label{T}
T:=
\begin{pmatrix}
1 & 0\\ 0& \frac{h_N}{2 \pi\i}
\end{pmatrix}
\ee
where $h_N=N!\G(N+\a+1)$ as in \eqref{h0}. 
\end{proposition}

\begin{remark}
The matrix $T$ is \emph{independent of $x$} and is introduced for convenience as it simplifies the coefficients in the expansions. This simplification does not affect the residue formul\ae\ of the previous paragraph, as it involves a \emph{constant} conjugation of $R(x)$.
\end{remark}

\begin{proof}
First off, we recall that
\be
Y(x):=Y(x;\t=\0)=
\renewcommand*{\arraystretch}{1.4}
\left(
\begin{array}{cc}
\pi_N^{(\a)}(x) & \wh\pi_N^{(\a)}(x)
\\
-\frac{2\pi\i}{h_{N-1}}\pi_{N-1}^{(\a)}(x) & -\frac{2\pi\i}{h_{N-1}}\wh\pi_{N-1}^{(\a)}(x)
\end{array}
\right)
\ee
where the polynomials $\pi_\ell^{(\a)}(x)$ and their Cauchy transforms $\wh\pi_\ell^{(\a)}(x)$ have been given in \eqref{laguerrepolynomial} and \eqref{cauchytransform} respectively, while $h_\ell$ is in \eqref{h0}. We can expand $\wh\pi_\ell^{(\a)}$ as $x\to\infty$ as
\begin{align}
\nonumber
\wh\pi_\ell^{(\a)}(x)&=\frac{1}{2\pi\i}\int_0^{+\infty}\pi_\ell^{(\a)}(\xi)\xi^\a\e^{-\xi}\frac{\d\xi}{\xi-x}
\\
\nonumber
&\sim-\frac{1}{2\pi\i }\sum_{j\geq 0}\frac 1{x^{j+1}}\int_0^{+\infty}\pi_\ell^{(\a)}(\xi)\xi^{\a+j}\e^{-\xi}\d\xi
\\
\nonumber
&=-\frac{1}{2\pi\i } \sum_{j\geq 0}\frac 1{x^{j+\ell+1}}\int_0^{+\infty}\pi_\ell^{(\a)}(\xi)\xi^{\a+j+\ell}\e^{-\xi}\d\xi
\\
\nonumber
&=-\frac{1}{2\pi\i } \sum_{j\geq 0}\frac 1{x^{j+\ell+1}}\int_0^{+\infty} (-1)^\ell \left(\frac{\d^\ell}{\d\xi^\ell}(\e^{-\xi}\xi^{\a+\ell})\right)\xi^{j+\ell}\d\xi
\\
\nonumber
&=-\frac{1}{2\pi\i }\sum_{j\geq 0}\frac 1{x^{j+\ell+1}}\int_0^{+\infty} \left(\frac{\d^\ell}{\d \xi^\ell}\xi^{j+\ell}\right)\xi^{\a+\ell}\e^{-\xi}\d\xi
\\
\label{formalexpcauchytransf}
&=-\frac{1}{2\pi\i } \sum_{j\geq 0} \frac{(j+1)_\ell\G(j+\ell+1+\a)}{x^{\ell+j+1}},
\end{align}
where we have used the orthogonality property to shift the sum in the first place, then Rodrigues formula \eqref{rodrigues} and integration by parts.
The expansion \eqref{formalexpcauchytransf} is formal; however, it has an analytic meaning of asymptotic expansion as $x\to\infty$. Indeed, for any $J\geq 0$ the difference between the Cauchy transform and its truncated formal expansion is
\be
\wh\pi^{(\a)}_\ell(x)+\frac 1 {2\pi\i} \sum_{j=0}^{J-1}\frac 1{x^{j+1}}\int_0^{+\infty}\!\pi_\ell^{(\a)}(\xi)\xi^{\a+j}\e^{-\xi}\d \xi=\frac 1{2\pi\i x^J}\int_0^{+\infty}\!\pi_\ell^{(\a)}(\xi)\xi^{\a+J}\e^{-\xi}\frac{\d \xi}{\xi-x}=\O\left(\frac 1 {x^{J+1}}\right)
\ee
where the last step holds as $x\to\infty$, uniformly in $\C\setminus[0,+\infty)$.
Hence, using \eqref{laguerrepolynomial} and \eqref{formalexpcauchytransf},
\be
\label{expansionY}
Y(x)\sim
\sum_{j\geq 0}\frac 1{j!x^j}
\begin{pmatrix}
(-1)^j(N-j+1+\a)_j(N-j+1)_j
&
-\frac {h_{N-1}}{2\pi\i x}(N+\a)_{j+1} (N)_{j+1}
\\
-\frac{2\pi\i}{h_{N-1}x}(-1)^j(N-j+\a)_j(N-j)_j
&
(N+\a)_j(N)_j
\end{pmatrix}x^{N\s_3}
\ee
as $x\to \infty$ within the sector $0<\arg x<2\pi$. Since $\det Y(x)\equiv 1$ we have
\be
\label{TRTinverse}
TR(x)T^{-1}=TY(x)\E_{11}Y^{-1}(x)T^{-1}=\left(
\begin{array}{cc}
1+Y_{21}(x)Y_{12}(x) & -\frac{2\pi\i}{h_N}Y_{11}(x)Y_{12}(x)
\\
\frac{h_N}{2\pi\i}Y_{21}(x)Y_{22}(x) & -Y_{21}(x)Y_{12}(x)
\end{array}
\right)
\ee
from which the expansion at $x=\infty$ can be computed as follows. For the $(1,1)$-entry we have
\be
Y_{21}(x)Y_{12}(x)\sim\sum_{\ell\geq 0}\frac 1{x^{\ell+2}}\sum_{j=0}^\ell\frac{(-1)^j(N+\a)_{\ell-j+1} (N)_{\ell-j+1}(N-j+\a)_{j}(N-j)_{j}}{j!(\ell-j)!}
\ee
and noting a trivial simplification of rising factorials
\be
\label{ID1}
(N+\a)_{\ell-j+1}(N-j+\a)_{j}=(N-j+\a)_{\ell+1},
\
(N)_{\ell-j+1}(N-j)_{j}=(N-j)_{\ell+1}
\ee
it follows that as $x\to \infty$
\be
(TR(x)T^{-1})_{11}\sim 1+\sum_{\ell\geq 0}\frac 1{x^{\ell+2}}(\ell+1)A_{\ell+1}(N,N+\alpha)= 1+\sum_{\ell\geq 0}\frac 1{x^{\ell+1}}\ell A_{\ell}(N,N+\alpha)=(R_+)_{11}(x)
\ee
with $A_{\ell}(N,M)$ as in \eqref{ABlMN}.
In a similar way we compute the $(1,2)$-entry
\begin{align}
\nonumber
&\text{\small{$-\frac{2\pi\i}{h_N}Y_{11}(x)Y_{12}(x)\sim\frac{1}{N(N+\alpha)}\sum_{\ell\geq 0}\frac{1}{x^{\ell+1}}\sum_{j=0}^\ell (-1)^j\frac{(N-j+1+\alpha)_j(N-j+1)_j(N+\alpha)_{\ell-j+1}(N)_{\ell-j+1}}{j!(\ell-j)!}$}}
\\
&\qquad \qquad \qquad \quad \ 
=\sum_{\ell\geq 0}\frac{1}{x^{\ell+1}}\sum_{j=0}^\ell (-1)^j\frac{(N-j+1+\alpha)_\ell (N-j+1)_\ell}{j!(\ell-j)!},
\end{align}
where in the second relation we use a similar version of \eqref{ID1}, and therefore from the above relation and \eqref{TRTinverse} we conclude that
\be
(TR(x)T^{-1})_{12}\sim\sum_{\ell\geq 0}\frac 1{x^{\ell+1}} B_{\ell}(N+1,N+1+\alpha)=(R_+)_{12}(x) 
\ee
with $B_{\ell}(N,M)$ as in \eqref{ABlMN}. 
Finally, the $(2,1)$-entry of the expansion of $TR(x)T^{-1}$ is computed in a similar way as
\begin{align}
\nonumber
\frac{2\pi\i}{h_N}Y_{21}(x)Y_{22}(x)&\sim
-N(N+\a)\sum_{\ell\geq 0}\frac{1}{x^{\ell+1}}\sum_{j=0}^\ell(-1)^j\frac{(N-j+\alpha)_j(N-j)_j(N+\alpha)_{\ell-j}(N)_{\ell-j}}{j!(\ell-j)!}
\\
\nonumber
&=-N(N+\a)\sum_{\ell\geq 0}\frac{1}{x^{\ell+1}}\sum_{j=0}^\ell(-1)^j\frac{(N-j+\alpha)_\ell (N-j)_\ell}{j!(\ell-j)!}
\\
&=-N(N+\a)\sum_{\ell\geq 0}\frac{1}{x^{\ell+1}}B_\ell(N,N+\a),
\end{align}
and the proof is complete.
\end{proof}

Let us note a recurrence property of the coefficients $A_\ell(N,M)$ and $B_\ell(N,M)$ defined in \eqref{ABlMN} entering the expansion \eqref{expR+}.

\begin{lemma}
\label{lemmarecursion}
The entries $A_\ell(N,M),B_\ell(N,M)$ ($\ell\geq 0$) defined in \eqref{ABlMN} satisfy the following three term recursions
\begin{align} \nonumber
\!\!\!\!\!\!\!\!\!\!\!\!(\ell+2)A_{\ell+1}(N,M)&=(2\ell+1)(N+M)A_\ell(N,M)+(\ell-1)(\ell^2-(M-N)^2)A_{\ell-1}(N,M),
\\  \label{RecursionAB}
\!\!\!\!\!\!\!\!\!\!\!\!(\ell+1)B_{\ell+1}(N,M)&=(2\ell+1)(N+M-1)B_\ell(N,M)+\ell(\ell^2-(M-N)^2)B_{\ell-1}(N,M),
\end{align}
for $\ell\geq 1$, with initial data given as
\be
A_0(N,M)=N,\quad A_1(N,M)=NM,\qquad B_0(N,M)=1,\quad B_1(N,M)=N+M-1.
\ee
\end{lemma}
\begin{proof}
Introduce the matrices
\be
\label{sl2basis}
\s_3=\left(
\begin{array}{cc}
1 & 0 \\ 0 &-1
\end{array}
\right),\qquad
\s_+=\left(
\begin{array}{cc}
0 & 1 \\ 0 &0
\end{array}
\right),\qquad
\s_-=\left(
\begin{array}{cc}
0 & 0 \\ 1 &0
\end{array}
\right),
\ee 
and write
\begin{equation}
\label{R_decom}
TR(x)T^{-1}=\frac 12 \1+r_3\s_3+r_+\s_++r_-\s_-,
\end{equation} where we use that $\tr R\equiv 1$; hereafter we omit the dependence on $x$ for brevity. Recalling the first equation in \eqref{deformation} we infer that
\be
\label{commutator}
\frac{\pa}{\pa x}R(x)=[\mathcal{A}(x),R(x)]\Rightarrow \frac{\pa}{\pa x}\left(TR(x)T^{-1}\right)=[T\mathcal{A}(x)T^{-1},TR(x)T^{-1}]
\ee
and writing
\be
T\mathcal{A}(x)T^{-1}=
-\frac 12\s_3+\frac 1x
\left(
\begin{array}{cc}
N+\frac \a 2 & -1
\\
N(N+\a) & -N-\frac \a 2
\end{array}
\right)
=a_3\s_3+a_+\s_++a_-\s_-
\ee
using \eqref{Amatrix}, we deduce from \eqref{commutator} the system of linear ODEs
\be \label{syseq}
\pa_x r_3=a_+r_--a_-r_+,\qquad
\pa_x r_+=2(a_3r_+-a_+r_3),\qquad
\pa_x r_-=2(a_-r_3-a_3r_-)
\ee
which in turn implies the following \emph{decoupled} third order equations for $\pa_xr_3,r_+,r_-$,
\begin{align}
\label{equation3}
\!\!\!\!\!\!\!\!\!\!\!\!3 (2N+\a-x)\pa_x r_3+(4-\a^2+2(2N+\a)x-x^2)\pa_x^2 r_3+5 x \pa_x^3 r_3+x^2 \pa_x^4 r_3&=0,
\\
\label{equatioN+-}
\!\!\!\!\!\!\!\!\!\!\!\!(2N+\a\pm 1-x)r_\pm+(1-\a^2+2 (2N+\a\pm 1)x- x^2)\pa_xr_\pm+3x\pa_x^2r_\pm+x^2\pa_x^3r_\pm&=0.
\end{align}
Finally, using the Wishart parameter $M=N+\a$, we substitute the expansion at $x=\infty$ given by \eqref{R+} into the ODEs \eqref{equation3} and \eqref{equatioN+-} to obtain the claimed recursion relations.
\end{proof}

\begin{remark}
\label{remarkhg}
Let us remark that the recursion for $A_\ell(N,M)$ in Lemma \ref{lemmarecursion} is also deduced, by different means, in \cite{HT2003}.
In \cite{CMOS2019} it is pointed out that such three term recursion is a manifestation of the fact that $A_\ell(N,M)$ is expressible in terms of \emph{hypergeometric orthogonal polynomials}; this property extends to the entries $B_\ell(N,M)$, as we now show.
Introducing the  \emph{generalized hypergeometric function} $_3F_2$
\be
\pFq{3}{2}{p_1,p_2,p_3}{q_1,q_2}{\zeta}:=\sum_{j\geq 0}\frac{(p_1)_j(p_2)_j(p_3)_j}{(q_1)_j(q_2)_j}\frac{\zeta^j}{j!}
\ee
we can rewrite the coefficients $A_\ell(N,M)$ and $B_\ell(N,M)$ in the form
\begin{align}
A_\ell(N,M)&:=\frac{(N)_{\ell}(M)_{\ell}}{\ell!}\pFq{3}{2}{1-N,1-M,1-\ell}{1-N-\ell,1-M-\ell}{1},
\\
B_\ell(N,M)&:=\frac{(N)_{\ell}(M)_{\ell}}{\ell!}\pFq{3}{2}{1-N,1-M,-\ell}{1-N-\ell,1-M-\ell}{1}.
\end{align}
Alternatively, introducing the \emph{Hahn} and \emph{dual Hahn polynomials} \cite{KS1994,CMOS2019}
\begin{align}
Q_j(x;\mu,\nu,k) &:= 
\pFq{3}{2}{-x, j+\mu+\nu+1,-j}{-k, \mu+1}{1},
\\
R_j(\lambda(x);\gamma,\delta,k) &:= \pFq{3}{2}{-j, x+\gamma+\delta+1, -x}{-k, \gamma+1}{1}, \qquad  \lambda(x)=x(x+\gamma+\delta+1)
\end{align}
the coefficients $A_\ell(N,M)$ and $B_\ell(N,M)$ can be rewritten in the form
\begin{align}
\frac{\ell!A_\ell(N,M)}{(N)_{\ell}(M)_{\ell}} &=  Q_{\ell-1}(N-1;-M-\ell,1,N+\ell-1) = R_{N-1}(\ell-1;-M-\ell,1,N+\ell-1),
\\
\frac{\ell!B_\ell(N,M)}{(N)_\ell(M)_\ell} &= Q_\ell(N-1;-M-\ell, 0, N+\ell-1) = R_{N-1}(\ell;-M-\ell,0,N+\ell-1).
\end{align}

Let us note that the first differential equation in \eqref{syseq} implies, at the level of the coefficients of the power series $r_3, r_-, r_+$, the following relation
\be \label{AvsB}
\ell(\ell+1)A_\ell(N,M)= NM \left ( B_\ell(N+1,M+1)-B_\ell(N,M) \right),
\ee
which will be used in Example \ref{2pointExp} to prove formul\ae \eqref{2pointCorr}.
\end{remark}

Let us now consider the asymptotic expansion as $x\to 0$.

\begin{proposition}\label{propexpansion-}
The matrix $R(x)$ admits the asymptotic expansion
\be
\label{expR-}
TR(x)T^{-1}\sim R_-(x),\qquad x\to 0
\ee
uniformly in $\C\setminus[0,+\infty)$. Here $R_-$ is the formal series introduced in the beginning of this paper, see \eqref{R-}, and $T$ is defined in \eqref{T}. 
\end{proposition}

\begin{proof}
First we observe that by arguments which are entirely analogous to those employed in the proof of Proposition \ref{propexpansion+}, the matrices $Y(x)$ and (consequently) $R(x)$ possess asymptotic expansions in integer powers of $x$ as $x\to 0$, which are uniform in $\C\setminus[0,+\infty)$.
The first coefficients of these expansions at $x=0$ can be computed from
\begin{align}
\pi_\ell^{(\a)}&=(-1)^\ell\left((\a+1)_\ell-\ell(\a+2)_{\ell-1}x+\O\left(x^2\right)\right)
\\
\wh\pi_\ell^{(\a)}&\sim\frac {(-1)^\ell}{2\pi\i}\left(\ell!\G(\a)+(\ell+1)!\G(\a-1)x+\O\left(x^2\right)\right)
\end{align}
where the former is found directly from \eqref{laguerrepolynomial} and the latter by a computation analogous to \eqref{formalexpcauchytransf}; hence recalling the definition \eqref{Y} we have
\begin{align}
\nonumber
Y(x)\sim&(-1)^N
\renewcommand*{\arraystretch}{1.4}
\begin{pmatrix}
(\a+1)_N & \frac{N!\G(\a)}{2\pi\i} \\
\frac{2\pi\i}{h_{N-1}}(\a+1)_{N-1} & \frac{(N-1)!\G(\a)}{h_{N-1}}
\end{pmatrix}
\\
&
+(-1)^{N}
\renewcommand*{\arraystretch}{1.4}
\begin{pmatrix}
-N(\a+2)_{N-1} & \frac{(N+1)!\G(\a-1)}{2\pi\i} \\
-\frac{2\pi\i}{h_{N-1}}(N-1)(\a+2)_{N-2} & \frac{N!}{h_{N-1}}\G(\a-1)
\end{pmatrix}x+\O\left(x^2\right)
\end{align}
as $x\to 0$ within $0<\arg x<2\pi$; this implies that in the same regime we have
\begin{align}
\nonumber
\!\!\!\!\!\!\!\!\!\!\!\!TR(x)T^{-1}\sim&\begin{pmatrix} 1 & 0 \\ 0 & 0 \end{pmatrix}+\frac 1\a\begin{pmatrix}
N & -1 \\
N(N+\a) & -N
\end{pmatrix}
\\
\label{firstterms}
&+\begin{pmatrix}
2N(N+\a) & -2N-\a-1 \\
N(N+\a)(2N+\a-1) & -2N(N+\a)
\end{pmatrix}\frac x{(\a-1)\a(\a+1)}+\O\left(x^2\right).
\end{align}
Therefore, our goal is just to show that the coefficients of the latter expansion are related to those of the expansion at $x=\infty$ as stated in the formul\ae\ \eqref{R+} and  \eqref{R-}. To this end let us write, in terms of the decomposition \eqref{R_decom},
\be
\label{ansatz}
r_3(x)\sim\frac 12+\sum_{\ell\geq 0}(\ell+1) \wt A_\ell(N,N+\a)\frac{x^\ell}{(\a-\ell)_{2\ell+1}}
,\qquad
r_\pm(x)\sim \sum_{\ell\geq 0}\wt B_\ell^\pm(N,N+\a)\frac{x^\ell}{(\a-\ell)_{2\ell+1}}
\ee
for some, yet undetermined coefficients $\wt A_\ell(N,M),\wt B_\ell^\pm(N,M)$. 
From \eqref{firstterms} we read the first coefficients $\wt A_{\ell}(N,M)$, $\wt B_\ell^\pm(N,M)$ in \eqref{ansatz} as
\begin{align}
\nonumber
\wt A_0(N,M)&=N=A_0(N,M),& \wt A_1(N,M)&=NM=A_1(N,M),\\
\nonumber
\wt B_0^+(N,M)&=-1=-B_0(N+1,M+1),& \wt B_1^+(N,M)&=-N-M-1=-B_1(N+1,M+1),\\
\nonumber
\wt B_0^-(N,M)&=NM=NMB_0(N,M),& \wt B_1^-(N,M)&=NM(N+M-1)=NMB_1(N,M).
\\
\label{initialdatainverse}
\end{align}
Finally, it can be checked that inserting \eqref{ansatz} in \eqref{equation3} and \eqref{equatioN+-} we obtain, again using $M=N+\a$, the recursions
\begin{align}
\nonumber
(\ell+2)\wt A_{\ell+1}(N,M)&=(2\ell+1)(N+M)\wt A_\ell(N,M)+(\ell-1)(\ell^2-(M-N)^2)\wt A_{\ell-1}(N,M),
\\
\nonumber
(\ell+1)\wt B_{\ell+1}^+(N,M)&=(2\ell+1)(N+M+1)\wt B_\ell^+(N,M)+\ell(\ell^2-(M-N)^2)\wt B_{\ell-1}^+(N,M),
\\
\nonumber
(\ell+1)\wt B_{\ell+1}^-(N,M)&=(2\ell+1)(N+M-1)\wt B_\ell^-(N,M)+\ell(\ell^2-(M-N)^2)\wt B_{\ell-1}^-(N,M),
\\
\label{recursioninverse}
\end{align}
for $\ell\geq 1$.
In view of Lemma \ref{lemmarecursion}, the linear recursions \eqref{recursioninverse} with initial data \eqref{initialdatainverse} are uniquely solved as
\be
\wt A_\ell(N,M)=A_\ell(N,M),\quad \wt B_\ell^+(N,M)=-B_\ell(N+1,M+1),\quad \wt B_\ell^-(N,M)=NMB_\ell(N,M).
\ee
Therefore from \eqref{R_decom}, \eqref{ansatz} and the above relation we obtain 
\be
TR(x)T^{-1}\sim\begin{pmatrix} 1 & 0 \\ 0 & 0 \end{pmatrix}+\sum_{\ell\geq 0}\frac{x^\ell}{(\a-\ell)_{2\ell+1}}
\begin{pmatrix}(\ell+1)A_\ell(N,N+\a)&-B_\ell(N+1,N+1+\a)
\\
N(N+\alpha)B_\ell(N,N+\a)&-(\ell+1) A_\ell(N,N+\a) 
\end{pmatrix}
\ee
with $\alpha=M-N$ and $A_\ell(N,M)$ and $B_{\ell}(N,M)$ as in \eqref{ABlMN}. The proof is complete
\end{proof}

We are finally ready to prove Theorem \ref{thmmain}.

\begin{proof}[{\bf Proof of Theorem \ref{thmmain}}]
Let us first consider the one-point generating functions $C_{1,0}(x)$ and $C_{0,1}(x)$. It is convenient to introduce 
the scalar function
\be
S_1(x):=\tr(\mathcal{A}(x)R(x)).
\ee
Indeed from \eqref{312} we see that for all $k\not =0$ we have
\begin{align}
\nonumber
\left.\frac{\pa Z_N(\a;\t)}{\pa t_k}\right|_{\t=\0}
&=-\res{x}\tr(\mathcal{A}(x)R(x))x^k\d x+\res{x}\left(\frac \a{2x}-\frac 12\right)x^k\d x
\\
\nonumber
&=-\res{x}
(xS_1(x))x^{k-1}\d x-\frac 12 \delta_{k,-1}
\\
\nonumber
&=\res{x}\frac{\pa (xS_1(x))}{\pa x}\frac{x^k}{k}\d x-\frac 12 \delta_{k,-1}
\\
\label{312bis}
&=\res{x}\left(\frac{\pa (xS_1(x))}{\pa x}+\frac 12\right)\frac{x^k}{k}\d x.
\end{align}
We now claim that
\be
\label{claim}
\frac{\pa}{\pa x} (xS_1(x))=\frac 12 - R_{11}(x).
\ee
Indeed we have
\be
\label{last}
\pa_x S_1(x)=\tr\left((\pa_x \mathcal{A}(x))R(x)\right)+\tr\left(\mathcal{A}(x)(\pa_x R(x))\right).
\ee
and noting the following identities
\be
\pa_x\mathcal{A}(x)=-\frac 1x \left(\mathcal{A}(x)+\frac 12 \s_3\right),\qquad \pa_x R(x)=[\mathcal{A}(x),R(x)]
\ee
we can rewrite \eqref{last} as
\be
\pa_x S_1(x)=-\frac 1x\tr(\mathcal{A}(x)R(x))-\frac{1}{2x} \tr(\s_3R(x))+\tr(\mathcal{A}(x)[\mathcal{A}(x),R(x)])
\ee
and \eqref{claim} follows noting $\tr(\mathcal{A}(x)[\mathcal{A}(x),R(x)])=\tr([\mathcal{A}(x),\mathcal{A}(x)R(x)])=0$ and 
\be
\frac 12 \tr(\s_3R(x))=\tr(\E_{11}R(x))-\frac 12\tr R(x)=R_{11}(x)-\frac 12
\ee
as $\tr R(x)\equiv 1$. Hence, inserting \eqref{claim} into \eqref{312bis} we obtain, irrespectively of the sign of $k$,
\be
\label{ooooonepoint}
\left\langle\tr X^k\right\rangle\stackrel{\eqref{connectedmoments}}{=}\left.\frac{\pa Z_N(\a;\t)}{\pa t_k}\right|_{\t=\0}
=-\frac 1k \res{x} \left(R_{11}(x)-1\right)x^k\d x.
\ee
At the level of generating functions, for $C_{1,0}(x)$ we have
\be
\pa_x(xC_{1,0}(x))\stackrel{\eqref{generatingonepoint}}{=}-\sum_{k\geq 1}\frac {k\left\langle\tr X^k\right\rangle}{x^{k+1}}\stackrel{\eqref{ooooonepoint}}{=}\sum_{k\geq 1}\frac{1}{x^{k+1}}\res{\xi=\infty}\left(R_{11}(\xi)-1\right)\xi^k\d \xi=-(R_+(x))_{11}+1
\ee
which, after integration, is the formula in the statement of Theorem \ref{thmmain}; in the last step of the last chain of equalities, we have to observe that $(R_+)_{11}=1+\O\left(x^{-2}\right)$ as $x\to\infty$, see \eqref{TRTinverse}.

Similarly, for $C_{0,1}(x)$ we have
\be
\pa_x(xC_{0,1}(x))\stackrel{\eqref{generatingonepoint}}{=}-\sum_{k\geq 1}kx^{k-1}\left\langle\tr X^{-k}\right\rangle\stackrel{\eqref{ooooonepoint}}{=}-\sum_{k\geq 1}x^{k-1}\res{\xi=0}\left(R_{11}(\xi)-1\right)\xi^{-k}\d \xi
=-(R_-(x))_{11}+1
\ee
which, after integration, is the formula in the statement of Theorem \ref{thmmain}. Here we have noted that $(R(x))_{11}=(TR(x)T^{-1})_{11}$ since $T$ is diagonal, see \eqref{T}.

The formul\ae\ for $r\geq 2$ are proven instead by the following computation;
\begin{align}
\nonumber
C_{r_+,r_-}(x_1,\dots,x_r)&\stackrel{\mathclap{\eqref{generatingmixed}}}{\ =\ }\sum_{k_1,\dots,k_{r}\geq 1}\frac {(-1)^{r_-}}{x_1^{\s_1k_1+1}\cdots x_r^{\s_rk_r+1}}\left\langle\tr X^{\s_1 k_1}\cdots \tr X^{\s_r k_r}\right\rangle_{\mathsf{c}}
\\
\nonumber
&\stackrel{\mathclap{\eqref{connectedmoments}}}{\ =\ }\sum_{k_1,\dots,k_{r}\geq 1}\frac {(-1)^{r_-}}{x_1^{\s_1k_1+1}\cdots x_r^{\s_rk_r+1}}\left.\frac{\pa^r\log Z_N(\a;\t)}{\pa t_{\s_1k_1}\cdots\pa t_{\s_rk_r}}\right|_{\t=\0}
\\
\nonumber
&\stackrel{\mathclap{\eqref{residuenpoint}}}{\ =\ }\sum_{k_1,\dots,k_{r}\geq 1}\frac{(-1)^{r_+}\res{\xi_1}\cdots\res{\xi_r}S_r(\xi_1,\dots,\xi_r;\t=\0)\xi_1^{\s_1k_1}\cdots \xi_r^{\s_rk_r}\d \xi_1\cdots\d \xi_r}{x_1^{\s_1k_1+1}\cdots x_r^{\s_rk_r+1}}
\\
&\stackrel{\mathclap{\eqref{S},\eqref{expR+},\eqref{expR-}}}{\ =\ }\ \ \ \quad - \sum_{(i_1,\dots,i_r)\in\mathcal{C}_r}\frac{\tr\left(R_{\s_{i_1}}(x_{i_1})\cdots R_{\s_{i_r}}(x_{i_r})\right)-\delta_{r,2}}{(x_{i_1}-x_{i_2})\cdots(x_{i_r}-x_{i_1})},
\end{align}
where we have noted that the transformation $R\mapsto TRT^{-1}$ leaves the expression $S_r$ invariant, and therefore we are free to use the expansions $R_\pm$ of Propositions \ref{propexpansion+} and \ref{propexpansion-}; the signs $\s_i$ are those defined in \eqref{sign}. The proof is complete.
\end{proof}

\begin{example}\label{2pointExp}
As an application of Theorem \ref{thmmain}, let us show how to prove formul\ae\ \eqref{2pointCorr}. Combining \eqref{2generating} and \eqref{2pointMain} gives
\be \label{eqq1}
\left\langle\tr X^{k} \tr X\right\rangle_{\mathsf{c}} =\res{x_1=\infty}\res{x_2=\infty}\frac{\tr(R_+(x_1)R_+(x_2))-1}{(x_1-x_2)^2}x_1^{k}x_2\ \d x_1\d x_2.
\ee
Let us write the matrix $R_+(x)$ as
\be
R_+(x)=E_{11}+\sum_{\ell\geq 0}\frac {R^+_{\ell}}{x^{\ell+1}}, \qquad R^+_{\ell}=
\begin{pmatrix}
\ell A_{\ell}(N,M) & B_\ell(N+1,M+1)
\\
-N M B_\ell(N,M) & -\ell A_{\ell}(N,M)
\end{pmatrix}
\ee
and expand the denominator in $1/(x_1-x_2)^2$ as a geometric series (the order we carry out the expansions in $x_1,x_2$ is irrelevant, as explained in Remark \ref{Regularity}) to rewrite the right side of \eqref{eqq1} as
\be
\frac 1{x_2^2}\sum_{h_1,h_2 \ge 0} \frac{x_1^{h_1+h_2}}{x_2^{h_1+h_2}}\left(\sum_{\ell_1 \ge 0} \frac{\tr(E_{11}R_{\ell_1}^+)}{x_1^{\ell_1+1}}+\sum_{\ell_2 \ge 0} \frac{\tr(E_{11}R_{\ell_2}^+)}{x_2^{\ell_2+1}}+\sum_{\ell_1, \ell_2 \ge 0}\frac{\tr(R_{\ell_1}^+R_{\ell_2}^+)}{x_1^{\ell_1+1}x_2^{\ell_2+1}}\right).
\ee
Finally, the residues extract the coefficient in front of $x_1^{-k-1}x_2^{-2}$, yielding
\be
\left\langle\tr X^{k} \tr X\right\rangle_{\mathsf{c}} = \tr(E_{11}R^+_{k}) = k A_k(N,M).
\ee

In a similar way, from  the relation 
\be \label{eqq2}
\left\langle\tr X^{-k} \tr X^{-1}\right\rangle_{\mathsf{c}} =\res{x_1=0}\res{x_2=0}\frac{\tr(R_-(x_1)R_-(x_2))-1}{(x_1-x_2)^2}x_1^{-k}x_2^{-1} \d x_1\d x_2
\ee
and 
\be
R_-(x)= \frac{1}{\a}\begin{pmatrix} M & -1 \\ N M & -N \end{pmatrix}+\sum_{\ell\geq 1}\frac{x^{\ell}}{(\a-\ell)_{2\ell+1}}R^-_{\ell}, \
R^-_{\ell}=
\begin{pmatrix}
(\ell+1)A_{\ell}(N,M) & -B_\ell(N+1,M+1)
\\
N M B_\ell(N,M) & -(\ell+1)A_{\ell}(N,M)
\end{pmatrix}
\ee
we obtain
\begin{align}
\nonumber
\left\langle\tr X^{-k} \tr X^{-1} \right\rangle_{\mathsf{c}} &= \frac{1}{\a}\tr\left[\begin{pmatrix} M & -1 \\ N M & -N \end{pmatrix} R^-_{k+1}\right] \\
\nonumber 
&= \frac{(k+2)(N+M)A_{k+1}(N,M) - NM \left( B_{k+1}(N+1,M+1)+B_{k+1}(N,M)\right)}{\a(\a-k-1)_{2k+3}} 
\\
&= \frac{k A_k(N,M)}{\a(\a-k)_{2k+1}}.
\end{align}
The last equality follows from the recursion relations \eqref{RecursionAB} and the formula \eqref{AvsB}. 
The computations of  $\left\langle\tr X^{k} \tr X^{-1} \right\rangle_{\mathsf{c}}$ and $\left\langle\tr X^{-k} \tr X^{1} \right\rangle_{\mathsf{c}}$ follow in a similar way.
\end{example}

\section{Proof of Proposition \ref{prop13}}\label{secproofprop13}

In this section we prove Proposition \ref{prop13} by means of the explicit formul\ae\ for the matrices $R_{\pm}(x)$ of Theorem \ref{thmmain}.
The proof follows from two main lemmas; the first one explains why rescaled correlators can be written as series in even powers of $N$ only.
We recall that we are working in the regime $\a=(c-1)N$, i.e. $M=cN$, with $c$ independent of $N$.
From \eqref{generatingmixed} we can write generating functions for the rescaled correlators appearing in \eqref{Texp+-} as
\be
\label{generatingmixedrescaled}
\sum_{k_1,\dots,k_{r}\geq 1}\frac{ N^{-\sum_{i=1}^r \s_ik_i}(-1)^{r_-}}{x_1^{\s_1k_1+1}\cdots x_r^{\s_rk_r+1}}\left\langle\tr X^{\s_1 k_1}\cdots \tr X^{\s_r k_r}\right\rangle_{\mathsf{c}}
=
N^rC_{r_+,r_-}\left(Nx_1,\dots,Nx_{r}\right),
\ee
where we use the signs in \eqref{sign}.

Let us preliminarily observe two properties of the formula \eqref{multipoint}, which are crucial to our proof of Proposition \ref{prop13}. First, such formula is invariant under replacing the matrices $R_\pm(x)$ with $G R_\pm(x) G^{-1}$ for some \emph{constant} non-degenerate matrix $G$, and second it is invariant (up to a simple modification for the two-point function) under replacing $R_\pm(x)$ with $R_\pm(x)+\gamma\1$ for any constant $\gamma\in\C$. While the first property is trivial, the second one requires few lines  of explanation. When $r=2$ one can exploit the fact that $\tr R_\pm(x)\equiv 1$ to write
\be
\frac{\tr \left(R_\pm(x_1)R_\pm(x_2)\right)-1}{(x_1-x_2)^2}=\frac{\tr \left((R_\pm(x_1)+\gamma\1)(R_\pm(x_2)+\gamma\1)\right)-(1+2\gamma+2\gamma^2)}{(x_1-x_2)^2}.
\ee
When $r\geq 3$ instead we reason as follows. Let us write every $r$-cycle $(i_1,\dots,i_r)$ with $i_r=r$, namely
\be
\sum_{(i_1,\dots,i_r)\in\mathcal{C}_r}\frac{\tr\left(R_{i_1}\cdots R_{i_r}\right)}{(x_{i_1}-x_{i_2})\cdots (x_{i_r}-x_{i_1})}
=
\sum_{(i_1,\dots,i_{r-1},r)\in\mathcal{C}_r}\frac{\tr\left(R_{i_1}\cdots R_{i_{r-1}}R_r\right)}{(x_{i_1}-x_{i_2})\cdots (x_{i_{r-1}}-x_r)(x_r-x_{i_1})},
\ee
where for the purpose of this explanation we adopt a short notation $R_i:=R_{\s(i)}(x_i)$; we point out that the role of the ``fixed" matrix $R_r$ is completely arbitrary, as the function \eqref{multipoint} is symmetric. Let us now show that this expression is invariant under the transformation $R_r\mapsto R_r+\gamma\1$; indeed the difference between the two expressions is computed from the last formula to be proportional to
\begin{align}
\nonumber
&\sum_{(i_1,\dots,i_{r-1},r)\in\mathcal{C}_r}\frac{\tr\left(R_{i_1}\cdots R_{i_{r-1}}\right)}{(x_{i_1}-x_{i_2})\cdots (x_{i_{r-1}}-x_r)(x_r-x_{i_1})}
\\ &\qquad
=\sum_{(i_1,\dots,i_{r-1},r)\in\mathcal{C}_r}\frac{\tr\left(R_{i_1}\cdots R_{i_{r-1}}\right)}{(x_{i_1}-x_{i_2})\cdots (x_{i_{r-1}}-x_{i_1})}\left(\frac 1{x_{i_{r-1}}-x_r}-\frac 1{x_{i_1}-x_r}\right)=0.
\end{align}
It follows that in \eqref{multipoint} one may inductively substitute all $R_i$'s by $R_i+\gamma\1$ (in principle, even with different $\gamma$'s for each $R_i$, but we do not need such freedom) without affecting the formula \eqref{multipoint}.

\begin{lemma} \label{Parity}
$C_{r_+,r_-}(N x_1,\dots,N x_r)$ is an even function of $N$ for every $r_+,r_-$.
\end{lemma}
\begin{proof}
Using formula \eqref{multipoint} in Theorem \ref{thmmain} we have
\be
\label{scaledmultipoint}
C_{r_+,r_-}(N x_1,\dots,N x_r)=-\frac 1{N^r} \sum_{(i_1,\dots,i_r)\in\mathcal{C}_r}\frac{\tr\left(R_{\s_{i_1}}(N x_{i_1})\cdots R_{\s_{i_r}}(N x_{i_r})\right)-\delta_{r,2}}{(x_{i_1}-x_{i_2})\cdots (x_{i_{r-1}}-x_{i_r})(x_{i_r}-x_{i_1})}.
\ee
After the considerations exposed just before this lemma, it is clear that we are done if we find a matrix $G$ such that the matrices $G R_{\pm}(N x) G^{-1}-\frac 12\1$ are both odd in $N$. We claim that the matrix
\be
G=\begin{pmatrix} \sqrt{c} & N^{-1}\\ -\sqrt{c} N & 1 \end{pmatrix}
\ee
serves this purpose. The proof of this claim is a computation that we now perform; we have
\begin{align}\nonumber
G R_+(N x) G^{-1} &= \frac 12\begin{pmatrix}
1 & -N^{-1} \\ -N & 1
\end{pmatrix}
\\ 
\label{ultima1}
&\phantom{=}{}
+ \frac 12 \sum_{\ell\geq 0}\frac{1}{x^{\ell+1}}\begin{pmatrix} D_\ell(c,N) & N^{-2}\left(\ell E_\ell(c,N)+F_\ell(c,N)\right) \\ \ell E_\ell(c,N)-F_\ell(c,N) & -D_\ell(c,N) \end{pmatrix}
\end{align}
where
\begin{align}
D_\ell(c,N)&:=\frac {\sqrt{c}}{N^\ell}\left(B_\ell(N+1,cN+1)-B_\ell(N,cN)\right)
\\
E_\ell(c,N)&:=-\frac{2}{N^\ell} A_\ell(N,cN)
\\
F_\ell(c,N)&:=\frac{\sqrt{c}}{N^{\ell-1}}\left(B_\ell(N+1,cN+1)+B_\ell(N,cN)\right).
\end{align}
Therefore our claim is equivalent to the statement that $D_\ell,E_\ell,F_\ell$ are odd functions of $N$. This is easily seen from the linear recursions of Lemma \ref{lemmarecursion}. 
For the coefficients $E_\ell$ the initial datum of the recursion is
\be
E_0(c,N)=-2N,\qquad E_1(c,N)=-2cN,
\ee
and the recursion reads
\be
N^2 (\ell+2)E_{\ell+1}(c,N)=N^2(2\ell+1)(c+1)E_\ell(c,N)+(\ell-1)(\ell^2-(c-1)^2N^2)E_{\ell-1}(c,N)
\ee
and the claim follows by induction, as the initial datum is odd and the recursion is even in $N$. Similarly, for the coefficients $D_\ell,F_\ell$, the initial datum of the recursion is odd in $N$
\be
D_0(c,N)=0,\ \  D_1(c,N)=\frac{2\sqrt{c}}{N}, \qquad F_0(c,N)=2 \sqrt{c} N,\ \ F_1(c,N)=2 N \sqrt{c}(c+1),
\ee
and the recursion is even in $N$
\be
\label{RecursionD}
\text{\small $\!\!\!\!\!\!\!\!
N^2 (\ell+1)D_{\ell+1}(c,N)=N^2 (c+1) (2\ell+1)D_\ell(c,N)+(2 \ell+1)F_\ell +\ell (\ell^2-N^2(c-1)^2)D_{\ell-1}(c, N)
$}
, 
\ee
\be
\label{RecursionF}
\text{\small $
N^2 (\ell+1)F_{\ell+1}(c,N)=N^2 (c+1) (2\ell+1)F_\ell(c,N)+N^2 (2 \ell+1)D_\ell +\ell (\ell^2-N^2(c-1)^2)F_{\ell-1}(c, N)
$}
.
\ee
The same claim for $R_-(Nx)$ is proven exactly in the same way, as we have
\begin{align}
\nonumber
&\!\!\!\!\!\!\!\!\!G R_-(N x) G^{-1} =\frac 12\begin{pmatrix}
1 & -N^{-1} \\ -N & 1
\end{pmatrix} 
\\
\label{ultima2}
&\!\!\! + \frac 12\sum_{\ell\geq 0}\frac{N^{2\ell+1}x^\ell}{(\a-\ell)_{2\ell+1}}\begin{pmatrix} -D_\ell(c,N) & N^{-2}\left((\ell+1)E_\ell(c,N)-F_\ell(c,N)\right) \\ (\ell+1) E_\ell(c,N)+F_\ell(c,N) & D_\ell(c,N) \end{pmatrix}
\end{align}
and, since $\a=(c-1)N$,
\be
\label{denom}
\frac{N^{2\ell+1}}{(\a-\ell)_{2\ell+1}} = \frac{N^{2\ell}}{(c-1)} \prod_{j=1}^\ell \frac{1}{N^2(c-1)^2-j^2}
\ee
which is even in $N$.
\end{proof}

The second lemma regards integrality of the coefficients.

\begin{lemma}\label{lemmalargeN}
The functions $A_\ell(N,M)$ and $B_\ell(N,M)$ in \eqref{ABlMN} admit the alternative expressions
\begin{align}
\label{IntAlNM}
A_\ell(N,M)&= \sum_{\begin{smallmatrix}a,b\geq 0 \\ a+b\leq {\ell-1} \end{smallmatrix}}\frac{\ell!(\ell-1)!(N-a)_{a+1}(M-b)_{b+1}}{(a+1)!(b+1)!a!b!(\ell-1-a-b)!},&\ell\geq 1,
\\
\label{IntBlNM}
B_\ell(N,M)&= \sum_{\begin{smallmatrix}a,b\geq 0 \\ a+b\leq \ell \end{smallmatrix}}\frac{\ell!(\ell-1)!(N-a)_a(M-b)_b}{a!^2b!^2(\ell-a-b)!},&\ell\geq 0.
\end{align}
\end{lemma}

\begin{proof}
Using the identity
\be
(\beta)_k=\beta(\beta+1)\cdots(\beta+k-1)=\left.\frac{\pa^k}{\pa x^k}x^{\beta+k-1}\right|_{x=1}
\ee
we rewrite \eqref{ABlMN}, for $\ell\geq 1$, as
\be
A_\ell(N,M)=\frac 1{\ell!}\sum_{k=0}^{\ell-1}(-1)^k{{\ell-1}\choose k}(N-k)_{\ell}(M-k)_{\ell}=\left.\frac{\pa^{\ell}}{\pa x^{\ell}}\frac{\pa^{\ell}}{\pa y^{\ell}}\frac{x^Ny^{M}(xy-1)^{\ell-1}}{\ell!}
\right|_{x=1,y=1}
\ee
and then we rewrite this expression, by a change of variable $1+\xi=x,1+\eta=y$, as
\be
\left.\frac{\pa^{\ell}}{\pa \xi^{\ell}}\frac{\pa^{\ell}}{\pa\eta^{\ell}}\frac{(1+\xi)^N(1+\eta)^{M}(\xi\eta+\xi+\eta)^{\ell-1}}{\ell!}
\right|_{\xi=0,\eta=0}=\!\sum_{\begin{smallmatrix}a,b\geq 0 \\ a+b\leq {\ell-1} \end{smallmatrix}}\!
\frac{\ell!(\ell-1)!(N-a)_{a+1}(M-b)_{b+1}}{(a+1)!(b+1)!a!b!(\ell-1-a-b)!}
\, .
\ee
Similarly, for all $\ell\geq 0$ we have
\begin{align}
\nonumber
B_\ell(N,M)&=\frac 1{\ell!}\sum_{k=0}^{\ell}(-1)^k{\ell\choose j}(N-k)_{\ell}(M-k)_{\ell}=\left.\frac{\pa^{\ell}}{\pa x^{\ell}}\frac{\pa^{\ell}}{\pa y^{\ell}}\frac{x^{N-1}y^{M-1}(xy-1)^{\ell}}{\ell!}
\right|_{x=1,y=1}
\\
\nonumber
&=\left.\frac{\pa^{\ell}}{\pa \xi^{\ell}}\frac{\pa^{\ell}}{\pa\eta^{\ell}}\frac{(1+\xi)^{N-1}(1+\eta)^{M-1}(\xi\eta+\xi+\eta)^{\ell}}{\ell!}
\right|_{\xi=0,\eta=0}
\\
&=\sum_{\begin{smallmatrix}a,b\geq 0 \\ a+b\leq {\ell-1} \end{smallmatrix}}\frac{\ell!(\ell-1)!(N-a)_a(M-b)_b}{a!^2b!^2(\ell-a-b)!}
\end{align}
and the proof is complete.
\end{proof}

The expression \eqref{IntAlNM} is also derived, in a different way, in \cite{HT2003}.

It can be checked that the coefficients $\frac{\ell!(\ell-1)!}{(a+1)!(b+1)!a!b!(\ell-1-a-b)!}$ in \eqref{IntAlNM} are integers within the range of summation $a,b\geq 0$, $a+b\leq \ell-1$; indeed if $a+b\leq \ell-2$ one can write such coefficient as
\be
\frac{\ell!(\ell-1)!}{(a+1)!(b+1)!a!b!(\ell-1-a-b)!}={\ell\choose{a+1}}{{\ell-1}\choose b}{{\ell-b-1}\choose a} (b+2)_{\ell-a-b-2}
\ee
which is manifestly integer, while if $a+b=\ell-1$ the same coefficient is written as
\be
\frac{\ell!(\ell-1)!}{(a+1)!(b+1)!a!b!(\ell-1-a-b)!}=\frac 1\ell {\ell\choose a}{\ell\choose {a+1}}
\ee
which  is also manifestly integer since $a\leq\ell-1$.
Similarly, the coefficients $\frac{\ell!(\ell-1)!}{a!^2b!^2(\ell-a-b)!}$ in \eqref{IntBlNM} are integers within the range of summation $a,b\geq 0$, $a+b\leq \ell$.

\begin{proof}[{\bf Proof of Proposition \ref{prop13}}]
Lemma \ref{lemmalargeN} implies that $A_\ell(N, c N)$ and $B_\ell(N, c N)$ are polynomials in $N$ and $c$ with integer coefficients. Then the dependence on $N^2$ follows from Lemma \ref{Parity} and the expansion of \eqref{denom} as series in $N$ and $(c-1)$ with integer coefficients as
\be
\frac 1{N(c-1)} \prod_{j=1}^\ell \frac 1{N^2(c-1)^2-j^2} 	
= \frac{1}{(N(c-1))^{2\ell+1}} \sum_{k_1, \dots, k_\ell \ge 0}\frac{1^{2 k_1} 2^{2 k_2} \cdot \dots \cdot \ell^{2 k_\ell}}{(N(c-1))^{2k_1 + \dots + 2k_\ell}},
\ee
provided $\ell<N(c-1)$.

Finally we note, e.g. from the recursions, that $D_\ell\sim \O(N^{-1}),E_\ell,F_\ell=\O(N)$, as $N\to\infty$; hence from \eqref{ultima1} and \eqref{ultima2} we see that $GR_\pm(Nx)G^{-1}={\rm diag}(1,N)\O(1){\rm diag}(1,N^{-1})$ where $\O(1)$ refers to the behavior as $N\to\infty$. We conclude that \eqref{generatingmixedrescaled} is $\O(1)$ as $N\to\infty$, and has the same parity in $N\mapsto -N$ as $r$ (Lemma \ref{Parity}), completing the proof.
\end{proof}

\begin{example}\label{GenusZeroEx}
Here we  obtain  the formul\ae\ of Theorem \ref{thmmain} in genus zero for one- and two-point correlators.
In these cases, formul\ae\ of the same kind   have already appeared in the literature \cite{FRW2010,Z2019,CV2015}.
In the regime $\a=N(c-1)$ with $N\to\infty$ we have
\begin{align}
\lim_{N\to\infty}\frac{A_\ell(N,cN)}{N^{\ell+1}}&=\frac 1\ell\sum_{b=0}^{\ell-1}{{\ell}\choose {b+1}}{{\ell}\choose {b}}c^{b+1}, 
\\
\lim_{N\to\infty}\frac{B_\ell(N,cN)}{N^{\ell}}&=\sum_{b=0}^{\ell}{\ell\choose b}^2c^b.
\end{align}
The above relations follow from Lemma \ref{lemmalargeN} and the trivial estimate ${N\choose k}\sim\frac{N^k}{k!}$.
In particular due to \eqref{onepoint+-}, in the regime $N\to\infty$ with $\a=N(c-1)$  we have
\be
\label{planar}
\lim_{N\to\infty}\frac{\left\langle\tr X^\ell\right\rangle}{N^{\ell+1}}=\sum_{s=1}^{\ell}\mathcal{N}_{\ell,s}c^{s}
\ee
where
\be
\mathcal{N}_{\ell,s}:=\frac 1\ell{\ell\choose s}{\ell\choose s-1},\qquad \ell\geq 1,\ s=1,\dots,\ell
\ee
are the \emph{Narayana numbers}. Formula \eqref{planar} agrees with Wigner's computation of positive moments of the Laguerre equilibrium measure $
\rho(x)=\frac {\sqrt{(x_+-x)(x-x_-)}}{2\pi c x}1_{x\in (x_-,x_+)}$ where $ x_\pm:=(1\pm\sqrt{c})^2$, see \cite{F2010}.
From the one-point function we obtain the weighted strictly monotone and weakly monotone double Hurwitz numbers of genus zero with partition $\mu=(k)$ and $\nu$ of length $s$ as
\begin{align}
&H_{g=0}^>((k);s)=\frac 1{(k-1)!}\sum_{\nu\text{ of length }s}h_{g=0}^{>}((k);\nu)=\mathcal{N}_{k,s}=\frac 1k { {k}\choose  s-1} {{k}\choose  s},\\
\label{GenusZeroPlanar}
&H_{g=0}^\geq((k);s)=\frac 1{(k-1)!}\sum_{\nu\text{ of length }s}h_{g=0}^{\geq}((k);\nu)={{k-1}\choose  k-s} \frac{(s+1)_{k-2}}{(k-1)!}.
\end{align}
Similarly, for all two-point generating functions, we obtain the \emph{planar limit} $g=0$ as
\begin{align}
\lim_{N\to\infty} N^2C_{2,0}(Nx_1,Nx_2)&=\lim_{N\to\infty} N^2C_{0,2}(Nx_1,Nx_2)=\frac{\phi(x_1,x_2)-\sqrt{\phi(x_1,x_1)\phi(x_2,x_2)}}{2\sqrt{\phi(x_1,x_1)\phi(x_2,x_2)}(x_1-x_2)^2}
\\
\lim_{N\to\infty} N^2C_{1,1}(Nx_1,Nx_2)&=-\frac{\phi(x_1,x_2)+\sqrt{\phi(x_1,x_1)\phi(x_2,x_2)}}{2\sqrt{\phi(x_1,x_1)\phi(x_2,x_2)}(x_1-x_2)^2}
\end{align}
where
\be
\phi(x_1,x_2):=c^2-c(2+x_1+x_2)+(x_1-1)(x_2-1).
\ee
The two-point planar limit is strictly related  \cite{EO2007} to the so called canonical symmetric bi-differential (called also Bergman kernel)  associated to the spectral curve
$x^2y^2=(x-x_+)(x-x_-)=c^2-2c(x+1)+(x-1)^2$.
\end{example}

\section{Hodge-LUE Correspondence}\label{secLUE}

\subsection{Factorization of matrix models with even potential}

For the purposes of the present section, let us  introduce two sequences of monic orthogonal polynomials;
$p^{\mathsf{even}}_n(x)=x^n+\dots$ satisfying
\be
\int_{-\infty}^{+\infty}p^{\mathsf{even}}_n(x)p^{\mathsf{even}}_m(x)\e^{-V(x^2)}\d x=h_{n}^{\mathsf{even}}\delta_{n,m}
\ee
and, for $\Re \a>-1$, $p_n^{(\a)}(x)=x^n+\dots$ satisfying
\be
\int_0^{+\infty}p^{(\a)}_n(x)p_m^{(\a)}(x)x^\a\e^{-V(x)}\d x= h_{n}^{(\a)}\delta_{n,m},
\ee
where $V(x)$ is an arbitrary potential for which the polynomials are well defined.
The following lemma is elementary and the proof can be found in \cite{CGM2015}.

\begin{lemma}
\label{lemmaelementary}
For all $n\geq 0$ we have
\be
p^{\mathsf{even}}_{2n}(x)=p_n^{\left(-\frac 12\right)}(x^2),\qquad p^{\mathsf{even}}_{2n+1}(x)=xp_n^{\left(\frac 12\right)}(x^2)
\ee
and
\be\label{rel_hn}
h_{2n}^{\mathsf{even}}=h_n^{\left(-\frac 12\right)},\quad h_{2n+1}^{\mathsf{even}}=h_n^{\left(\frac 12\right)}.
\ee
\end{lemma}

Next we recall the relation between matrix integrals and the norming constants of the  above  orthogonal polynomials 
\begin{align}
\frac{1}{\Vol(N)}\int_{\H_N}\exp\tr(-V(X^2))\d X&= h^\mathsf{even}_0h^\mathsf{even}_1\cdots h^{\mathsf{even}}_{N-1},
\\
\frac{1}{\Vol(N)}\int_{\H_N^+}{\det}^\a X\exp\tr(-V(X))\d X&= h_0^{(\a)}h_1^{(\a)}\cdots h_{N-1}^{(\a)},
\end{align}
where $\Vol(N)$ is defined in \eqref{factorization2}.

Using the above relations and \eqref{rel_hn} in the case $V(x)=\frac x2 - \sum_{k\geq 1}s_k x^k$, we obtain the following identity  between the GUE partition function $Z_{2N}^{\mathsf{even}}(\ss)$ in \eqref{evenGUE} and the Laguerre partition function $Z_N\left(\pm\frac 12; \t_+\right)$ in \eqref{Z} with $\t_-=\0$
\be
\label{factorization1}
\frac{Z_{2N}^{\mathsf{even}}(\ss)}{Z_{2N}^{\mathsf{even}}(\0)}=
\frac{Z_N\left(-\frac 12; \t_+\right)}{Z_N\left(-\frac 12;\0\right)}
\frac{Z_N\left(\frac 12;\t_+\right)}{Z_N\left(\frac 12 ; \0\right)}
,\qquad t_k:=2^k s_k,
\ee
where  $Z_N^{\mathsf{even}}(\0)$ is given in \eqref{normalizingGUE} and $Z_N\left(\pm\frac 12 ; \0\right)$  in  \eqref{normalizinglaguerre}. There is a similar, slightly more involved, factorization for the matrix model $Z_{2N+1}^{\mathsf{even}}$, but we do not need its formulation for our present purposes.

\subsection{Formal matrix models and mGUE partition function}\label{secformal}

In this section we review the definition of mGUE partition function. First, the logarithm of the even GUE partition function can be considered as a formal Taylor expansion for small $s_k$ as
\be
\label{evenGUEformal}
\log Z_N^{\mathsf{even}}(\ss):=\log Z_N^{\mathsf{even}}(\0)+\sum_{r\geq 1}\sum_{k_1,\dots,k_r\geq 1}\frac{s_{k_1}\cdots s_{k_r}}{r!}\left\langle\tr X^{k_1}\cdots\tr X^{k_r}\right\rangle_{\mathsf{c}}^{\mathsf{even}}
\ee
where the connected even GUE correlators are introduced as in \eqref{connectedmoments}
\be
\left\langle\tr X^{k_1}\cdots\tr X^{k_r}\right\rangle_{\mathsf{c}}^{\mathsf{even}}:=\left.\frac{\pa^r \log Z_N^{\mathsf{even}}(\ss)}{\pa s_{k_1}\cdots\pa s_{k_r}}\right|_{\ss=\0}
\ee
and the normalizing constant $Z_N^{\mathsf{even}}(\0)$ is 
\be
\label{normalizingGUE}
Z_N^{\mathsf{even}}(\0)=\sqrt{2^N\pi^{N^2}}.
\ee
The infinite sum in \eqref{evenGUEformal} can be given a rigorous formal meaning in the algebra $\C[N,\a][[\ss]]$; introducing the grading $\deg s_k:=k$, the latter algebra is obtained taking the inductive limit $K\to\infty$ from the algebras of polynomials in $\ss$ of degree $<K$, with coefficients in $\C[N,\a]$. Equivalently, this grading can be encoded, up to an inessential shift, by a (small) variable $\epsilon$ via the transformation $s_k\mapsto \epsilon^{k-1} s_k$, which is the same as considering the matrix model $\int_{\H_N}\exp\left[-\frac 1\epsilon\left(\frac {X^2}2-\sum_{k\geq 1}s_k X^{2k}\right)\right]\d X$. For simplicity we have preferred to avoid the explicit $\epsilon$-dependence, even though we shall restore it for the statement of the Hodge-GUE/LUE correspondence (Theorem \ref{thmhodgegue}, Corollary \ref{corollary1}).

It must be stressed that \eqref{evenGUEformal} makes sense for any complex $N$, and not just for positive integers as it would be required by the genuine matrix integral interpretation; indeed the correlators are polynomials in $N$.

For the purposes of this section it is convenient to apply the same arguments to the Laguerre partition function (with $\t_-=\0$) and similarly identify the latter with the formal series
\be
\label{formalLUE}
\log Z_N(\a;\t_+)=\log Z_N(\a;\0)+\sum_{r\geq 1}\sum_{k_1,\dots,k_r\geq 1}\frac{t_{k_1}\cdots t_{k_r}}{r!}\left\langle\tr X^{k_1}\cdots\tr X^{k_r}\right\rangle_{\mathsf{c}}
\ee
where $Z_N(\a;\0)$ is given in \eqref{normalizinglaguerre} and the correlators are as in \eqref{connectedmoments}; using the last expression provided in \eqref{normalizinglaguerre} and the fact that the correlators are polynomials in $N,\a$ the expression \eqref{formalLUE} makes sense  also for  $N$  complex. This remark is crucial for a correct understanding of formul\ae\ \eqref{symmetryZ} and \eqref{symmetry12} below.

Let us finally recall from the introduction and \cite{DLYZ2016} that the mGUE partition function is introduced by \eqref{factorization2}, the left side of which being interpreted formally as in \eqref{evenGUEformal}. Of course in the identification of Theorem \ref{thm2}, the right side must be interpreted formally as in \eqref{formalLUE}.

\subsection{Proof of Theorem \ref{thm2}}\label{secproofthm2}

The proof of Theorem \ref{thm2} relies on two main ingredients; on one side the factorization property \eqref{factorization1}, and a symmetry property of the formal positive LUE partition function \eqref{formalLUE}, which we now describe.

\begin{lemma}
\label{lemmasymmetry}
The LUE connected correlator $\langle\tr X^{k_1}\cdots\tr X^{k_r}\rangle_\mathsf{c}$ with $k_1,\dots, k_r>0$ is a polynomial in $N,\a$, and it is invariant under the involution $(N,\a)\mapsto (N+\a,-\a)$.
\end{lemma}

\begin{proof}
It follows directly from Theorem \ref{thmmain}, as the coefficients of $R_+$, defined in \eqref{R+}, are polynomials in $N,\a$ which are manifestly symmetric under the aforementioned transformation.
Indeed from \eqref{ABlMN} we see that all the coefficients $A_\ell(N,M),B_\ell(N,M)$ are symmetric in $N,M$ but $A_0(N,M)=N$; however $R_+$ only contains the combination $\ell A_\ell(N,M)$, which is always symmetric in $N,M$.
\end{proof}

\begin{remark}
\label{remcombinatorial}
As already pointed out in the introduction, the symmetry of the above lemma is equivalent, by \eqref{hurwitzstrict}, to the symmetry property \eqref{symmetryhurwitz} of Hurwitz numbers. An alternative, purely combinatorial derivation of the symmetry in the latter formulation can be given as follows. In the group algebra of the group of permutations of $\{1,\dots,d\}$ (consisting of formal $\C$-linear combinations of permutations of $\{1,\dots,d\}$) we consider the distinguished elements
\begin{itemize}
\item $\mathcal{C}_\lambda$ for any partition $\lambda$ of $d$, which is the sum of all permutations of cycle type $\lambda$, and
\item $\mathcal{J}_m:=(1,m)+\dots+(m-1,m)$ (\emph{Jucys-Murphy elements}) for $m=2,\dots,d$.
\end{itemize}
Such elements commute among themselves and are related by
\be
(1+\xi\mathcal{J}_2)\cdots(1+\xi\mathcal{J}_d)=\sum_{|\nu|=d}\xi^{d-\ell(\nu)}\mathcal{C}_\nu.
\ee
From this relation we deduce that for every partition $\mu$ of $d$ of length $\ell$ we have
\be
\label{identitygroupalgebra}
\mathcal{C}_\mu\sum_{|\nu|=d}y^{d-\ell(\nu)} \mathcal{C}_\nu(1+z\mathcal{J}_2)\cdots (1+z\mathcal{J}_d)=
\mathcal{C}_\mu\sum_{|\lambda|=d}z^{d-\ell(\lambda)}(1+y\mathcal{J}_2)\cdots(1+y\mathcal{J}_d)\mathcal{C}_\lambda.
\ee
From the definition of Hurwitz numbers, recalled in Section \ref{paragraphhurwitz}, the coefficient in front of $\mathcal{C}_{[1^d]}y^{d-s}z^{2g-2+\ell+s}$ on the left side of \eqref{identitygroupalgebra} is $H^>_g(\mu;s)$ (up to the normalization factor $\frac{z_\mu}{d!}$); the coefficient of the same element $\mathcal{C}_{[1^d]}y^{d-s}z^{2g-2+\ell+s}$ on the right side of \eqref{identitygroupalgebra} is $H^>_g(\mu;2-2g+d-\ell-s)$ (up to the same normalization factor $\frac{z_\mu}{d!}$), yielding \eqref{symmetryhurwitz}.
\end{remark}

Let us restate Lemma \ref{lemmasymmetry}, in view of the formal expansion \eqref{formalLUE}, as the following identity
\be
\label{symmetryZ}
\frac{Z_N(\a;\t_+)}{Z_N(\a;\0)}=\frac{Z_{N+\a}(-\a;\t_+)}{Z_{N+\a}(-\a;\0)}.
\ee
The following special case ($\a=\frac 12$) of \eqref{symmetryZ}
\be
\label{symmetry12}
Z_{N+\frac 12}\left(-\frac 12;\mathbf{t}_+\right)=\frac{\pi^{\frac 38+\frac N2}\barnes(N+1)}{\barnes\left(N+\frac 32\right)}Z_N\left(\frac 12;\mathbf{t}_+\right)
\ee
plays a key role in the proof of Theorem \ref{thm2}, which we are now ready to give.

\begin{proof}[{\bf Proof of Theorem \ref{thm2}}]
We use the uniqueness of the decomposition \eqref{factorization2} which defines the mGUE partition function; rewriting it under the substitution $N\mapsto 2N$ we have
\be
\label{factorizationA}
\frac{Z_{2N}^{\mathsf{even}}(\ss)}{(2\pi)^{2N}\Vol(2N)}=\wt Z_{2N-\frac 12}(\ss)\wt Z_{2N+\frac 12}(\ss).
\ee
On the other hand, from \eqref{factorization1} we have
\be
\label{factorizationB}
Z_{2N}^{\mathsf{even}}(\ss)=D_N Z_{N}\left(-\frac 12;\t_+\right)Z_{N}\left(\frac 12;\t_+\right)
\ee
where here and below we are identifying $t_k=2^ks_k$.  The proportionality constant $D_N$ is explicitly evaluated from \eqref{normalizinglaguerre} and \eqref{normalizingGUE} as
\be
\label{proportionalityA}
D_N=\frac{Z_{2N}^{\mathsf{even}}(\0)}{Z_N\left(-\frac 12;\0\right)Z_N\left(\frac 12;\0\right)}=\frac{2^N \pi^{N^2 + N + \frac 12} \barnes(\frac 12)^2 }{ \barnes(N + \frac 12)\barnes(N + \frac 32) }.
\ee
It is then enough to show that the two factorizations \eqref{factorizationA} and \eqref{factorizationB} are consistent once we identify  $\wt Z_{2N-\frac 12}(\ss)=C_N Z_{N}\left(-\frac 12;\t_+\right)$   with $C_N$ a constant depending on $N$ only. Such consistency follows from the chain of equalities
\begin{align}
\frac{Z_{2N}^{\mathsf{even}}(\ss)}{(2\pi)^{2N}\Vol(2N)}&=
\wt Z_{2N-\frac 12}(\ss)\wt Z_{2N+\frac 12}(\ss)
\nonumber
\\
&=\wt Z_{2N-\frac 12}(\ss)\wt Z_{2\left(N+\frac 12\right)-\frac 12}(\ss)
\nonumber
\\
&=C_N Z_{N}\left(-\frac 12;\t_+\right)C_{N+\frac 12}Z_{N+\frac 12}\left(-\frac 12;\t_+\right)
\nonumber
\\
&=C_NC_{N+\frac 12}\frac{\pi^{\frac 38+\frac N2}\barnes(N+1)}{\barnes\left(N+\frac 32\right)}Z_{N}\left(-\frac 12;\t_+\right)Z_{N}\left(\frac 12;\t_+\right)
\label{proportionalityB}
\end{align}
where we have used the symmetry property \eqref{symmetry12}. This shows that the two factorizations \eqref{factorizationA} and \eqref{factorizationB} are consistent, provided we also identify the proportionality constants \eqref{proportionalityA} and \eqref{proportionalityB}
\be
\label{eq:eqforCN}
C_NC_{N+\frac 12}\frac{\pi^{\frac 38+\frac N2}\barnes(N+1)}{\barnes\left(N+\frac 32\right)}=\frac{D_N}{(2\pi)^{2N}\Vol(2N)}=
4^{N(N-1)}\pi^{-N(N+1)} \barnes(N+1)^2,
\ee
where in the last step we use the duplication formula for the Barnes G-function in the form
\be
\barnes(2N+1)=\frac{2^{N (2 N-1)} \pi ^{-N-\frac{1}{2}}}{\barnes\left(\frac{1}{2}\right)^2}\barnes\left(N+\frac 12\right)\barnes(N+1)^2\barnes\left(N+\frac 32\right).
\ee
Equation \eqref{eq:eqforCN} fixes the constant to be
\be
C_N=2^{N^2-\frac{3}{2}N+\frac{1}{4}} \pi ^{-\frac{N(N+1)}{2}} \barnes(N+1),
\ee
as stated in \eqref{constantmGUEevenvsLaguerre}.
\end{proof}

We conclude this section with a couple of remarks.

First, the identification of the mGUE and LUE partition functions is manifest also from the \emph{Virasoro constraints} of the two models. Indeed, Virasoro constraints for the modified GUE partition function have been derived in \cite{DLYZ2016}, directly from those of the GUE partition function, and they assume the form $\wt {\mathcal{L}}_n\wt Z_N(\ss)=0$, for $n\geq 0$, where
\be
\label{virasoromGUE}
\wt {\mathcal{L}}_n:=\begin{cases}
\sum_{k\geq 1}k \left(s_k-\frac 12\delta_{k,1}\right)\frac{\pa}{\pa s_k}+\frac{N^2}4-\frac 1{16}, & n=0,
\\
\sum_{k=1}^{n-1}\frac{\pa^2}{\pa s_k\pa s_{n-k}}+\sum_{k\geq 1}k \left(s_k-\frac 12\delta_{k,1}\right)\frac{\pa}{\pa s_{k+n}}+N\frac{\pa}{\pa s_n}, & n\geq 1.
\end{cases}
\ee
On the other hand, it is well known \cite{HH1993,AvM1995} that the LUE partition function with only positive couplings $\t_+$ satisfies the Virasoro constraints $\mathcal{L}_n^{(\a)} Z_N(\a;\t_+)=0$, for $n\geq 0$, where
\be
\label{virasoroLUE}
\mathcal{L}_n^{(\a)}:=\begin{cases}
\sum_{k\geq 1}k \left(t_k-\delta_{k,1}\right)\frac{\pa}{\pa t_k}+N\left(N+\a\right), & n=0,
\\
\sum_{k=1}^{n-1}\frac{\pa^2}{\pa t_k\pa t_{n-k}}+\sum_{k\geq 1}k \left(t_k-\delta_{k,1}\right)\frac{\pa}{\pa t_{k+n}}+(2N+\a)\frac{\pa}{\pa t_n},
& n\geq 1.
\end{cases}
\ee
The Virasoro constraints $\wt{\mathcal{L}}_n=\wt {\mathcal{L}}_n(N,\ss)$ in \eqref{virasoromGUE} and $=\mathcal{L}_n^{(\a)}=\mathcal{L}_n^{(\a)}(N,\t_+)$ in \eqref{virasoroLUE} satisfy
\be
2^n\wt{\mathcal{L}}_n\left(2N-\frac 12,\ss\right)=\mathcal{L}_n^{\left(-\frac 12\right)}(N,\t_+)
\ee 
under the identification $t_k=2^ks_k$, in agreement with Theorem \ref{thm2}.

Second, in \cite{DY2019} formul\ae\ of similar nature as those of Theorem \ref{thmmain} are derived for the modified GUE partition function. It can be checked that such formul\ae\ match with those of Theorem \ref{thmmain} restricted to $\a=-\frac 12$ under the identifications of times made explicit in the statement of Theorem \ref{thm2}.

\subsection{Proof of Corollary \ref{corollary2}}\label{secproofcorollaries}

From Theorem \ref{thm2} and the Hodge-GUE correspondence of \cite{DLYZ2016}, which we now recall, we are able to deduce a \emph{Hodge-LUE correspondence}; to state this result (Corollary \ref{corollary1}) let us introduce the generating function
\be
\mathcal{H}(\mathbf{p};\epsilon)=\sum_{g\geq 0}\epsilon^{2g-2}\sum_{n\geq 0}\sum_{k_1,\dots,k_n\geq 0}\frac{p_{k_1}\cdots p_{k_n}}{n!}\int_{\Mgn}\L^2(-1)\L\left(\frac 12\right)\prod_{i=1}^n\psi_i^{k_i},
\ee
for \emph{special cubic Hodge integrals} (with the standard notations recalled before the statement of Corollary \ref{corollary2}); here $\mathbf{p}=(p_0,p_1,\dots)$.

\begin{theorem}[Hodge-GUE correspondence \cite{DLYZ2016}]
\label{thmhodgegue}
Introduce the formal series
\be
\label{A(x,s)}
A(\lambda,\ss) := \frac 14 \sum_{j_1,j_2\geq 1}\frac{j_1j_2}{j_1+j_2}{{2j_1}\choose j_1}{{2j_2}\choose j_2}s_{j_1}s_{j_2} +  \frac 12\sum_{j\geq 1}\left(\lambda-\frac j{j+1}\right){{2j}\choose j}s_j,
\ee
and a transformation of an infinite vector of times $\ss=(s_1,s_2,\dots)\mapsto\mathbf{p}=(p_0,p_1,\dots)$ depending on a parameter $\lambda$ as
\be
\label{changep}
p_k(\lambda,\ss):=\sum_{j\geq 1}j^{k+1}{{2j}\choose j}s_j+\delta_{k,1}+\lambda\delta_{k,0}-1,\qquad k\geq 0.
\ee
Then we have
\be
\label{hodgegueformula}
\mathcal{H}\left(\mathbf{p}\left(\lambda,\ss\right);\sqrt{2}\epsilon\right)+\epsilon^{-2} A(\lambda, \mathbf{s}) =\log \wt Z_{\frac \lambda\epsilon}\left( (s_1,\epsilon s_2,\epsilon^2 s_3,\dots)\right)+B(\lambda,\epsilon)
\ee
where   $B(\lambda,\epsilon)$ is a constant depending on $\lambda$ and $\epsilon$ only and $\wt Z_{\frac \lambda\epsilon}$  is the mGUE partition function   in \eqref{factorization2}.
\end{theorem}

\begin{corollary}[Hodge-LUE correspondence]
\label{corollary1}
Let $\mathcal{H}\left(\mathbf{p}\left(\lambda,\ss\right);\sqrt{2}\epsilon\right)$ as in \eqref{hodgegueformula} and $Z_N\left(-\frac 12;\t_+\right)$ the Laguerre partition function \eqref{Z} with parameter $\alpha=-\frac{1}{2}$ and times $\t_+$ and $\t_-=0$. We have
\be
\mathcal{H}\left (\mathbf{p}(\lambda,\ss);\sqrt{2}\epsilon\right)+ \epsilon^{-2} A \left (\lambda,\ss\right )= \log Z_N\left(-\frac 12;\t_+\right) + C(N,\epsilon),
\ee
where we identify
\be
\lambda=\epsilon\left( 2 N -\frac 12 \right),\qquad t_k =  2^k\epsilon^{k-1} s_k,
\ee
and $A(\lambda, \ss)$ is defined in \eqref{A(x,s)}, $\mathbf{p}(\lambda,\ss)$ is defined in \eqref{changep}, and $C(N,\epsilon)$ is a constant depending on $N$ and $\epsilon$ only.
\end{corollary}

\begin{proof}
It follows from \eqref{hodgegueformula} upon the substitution $\lambda \mapsto \epsilon\left( 2 N -\frac 12 \right)$ and applying Theorem \ref{thm2} for the set of times $\epsilon^{k-1}s_k$, $k\geq 1$.
\end{proof}

It would be interesting to construct the Double Ramification  hierarchy (see \cite{B2015,BR2016}) for cubic Hodge integrals, and then check in this case the conjecture formulated in \cite{BDGR2018} by which the logarithm of the corresponding tau function should coincide with the LUE partition function, after the change of variables described in \cite{BDGR2020}.

Finally, Corollary \ref{corollary2} is obtained matching the coefficients in \eqref{hodgegueformula} using \eqref{hurwitzstrict}.

\begin{proof}[{\bf Proof of Corollary~\ref{corollary2}}]
We apply $\left.\frac{\pa^\ell}{\pa s_{\mu_1}\cdots\pa s_{\mu_\ell}}\right|_{\ss=\0}$, for $\ell > 0$, on both sides of \eqref{hodgegueformula}. On the right side we get, in view of Theorem \ref{thm2}
\begin{align}
\nonumber
&\left.\frac{\pa^\ell}{\pa s_{\mu_1}\cdots\pa s_{\mu_\ell}}\right|_{\ss=\0}
\log \wt Z_{\frac \lambda\epsilon}\left((s_1,\epsilon s_2,\epsilon^2 s_3,\dots)\right)
=\epsilon^{|\mu|-\ell}2^{|\mu|}\left\langle\tr X^{\mu_1}\cdots\tr X^{\mu_\ell}\right\rangle_{\mathsf{c}}\big|_{N=\frac \lambda{2\epsilon}+\frac 14,\, \a=-\frac 12}
\\
&\qquad
=\epsilon^{|\mu|-\ell}2^{|\mu|}\sum_{\g\geq 0}\sum_{s=1}^{1-2\g+|\mu|-\ell}\left(\frac{\lambda+\frac\epsilon 2}{2\epsilon}\right)^{2-2\g+|\mu|-\ell}\left(\frac{\lambda-\frac\epsilon 2}{\lambda+\frac \epsilon 2}\right)^sH_\g^>(\mu;s)
\end{align}
where in the last step we have used \eqref{hurwitzstrict}; we also note that the substitutions $2N-\frac 12=\frac \lambda\epsilon$, $\a=-\frac 12$, from Theorem \ref{thm2}, yield $N=\frac {\lambda+\frac \epsilon 2}{2\epsilon}$, $c=\frac{\lambda-\frac \epsilon 2}{\lambda+\frac \epsilon 2}$. On the other side we get
\be
\left.\frac{\pa^\ell}{\pa s_{\mu_1}\cdots\pa s_{\mu_\ell}}\right|_{\ss=\0}\mathcal{H}(\mathbf{p}(\lambda,\mathbf{s});\sqrt{2} \epsilon)+\epsilon^{-2}\left.\frac{\pa^\ell}{\pa s_{\mu_1}\cdots\pa s_{\mu_\ell}}\right|_{\ss=\0} A(\lambda,\mathbf{s}).
\ee
The contributions from the last term is directly evaluated from \eqref{A(x,s)} and give the second line of \eqref{Hodgegmu}. For the first term we recall the \emph{affine} change of variable \eqref{changep} and compute
\begin{align}
\nonumber
&\text{\small $
\frac{\pa^\ell}{\pa s_{\mu_1}\cdots\pa s_{\mu_\ell}}\mathcal{H}\left(\mathbf{p}(\lambda,\mathbf{s});\sqrt{2}\epsilon\right)=
\sum_{i_1,\dots,i_\ell\geq 0}\prod_{b=1}^\ell\mu_b^{i_b+1}{{2\mu_b}\choose \mu_b}\frac{\pa^\ell}{\pa p_{i_1}\cdots\pa p_{i_\ell}}\mathcal{H}\left(\mathbf{p}(\lambda,\mathbf{s});\sqrt{2}\epsilon\right)$}
\\
\label{eq}
&\text{\small $=\sum_{g,n\geq 0}\left(\sqrt 2 \epsilon\right)^{2g-2}\sum_{\begin{smallmatrix}k_1,\dots,k_n\geq 0 \\ i_1,\dots,i_\ell\geq 0 \end{smallmatrix}}\frac{p_{k_1}(\lambda,\ss)\cdots p_{k_n}(\lambda;\ss)}{n!}
\int_{\overline{\mathcal{M}}_{g,n+\ell}}\L^2(-1)\L\left(\frac 12\right)\prod_{a=1}^n\psi_a^{k_a}\prod_{b=1}^\ell \mu_b^{i_b+1}{{2\mu_b}\choose \mu_b}\psi_{n+b}^{i_b}$}.
\end{align}
Evaluation at $\ss=\0$ corresponds to $p_k=\delta_{k,1}+\lambda\delta_{k,0}-1$; thus, in the previous expression, we set $n=m+r$, where $m$ is the number of $k_a$'s equal to zero, and the remaining $k_1,\dots,k_r$'s are all $\geq 2$ (we are evaluating at $p_1=0$), and so the evaluation of the \eqref{eq} at $p_k=\delta_{k,1}+\lambda\delta_{k,0}-1$ reads
\begin{align}
\nonumber
&\text{\small $\sum_{g,m,r\geq 0}\left(\sqrt 2 \epsilon\right)^{2g-2}\sum_{\begin{smallmatrix}k_1,\dots,k_r\geq 2 \\ i_1,\dots,i_\ell\geq 0 \end{smallmatrix}}\frac{(\lambda-1)^m(-1)^r}{m!r!}
\int_{\overline{\mathcal{M}}_{g,\ell+m+r}}\L^2(-1)\L\left(\frac 12\right)\prod_{a=1}^r\psi_a^{k_a}\prod_{b=1}^\ell \mu_b^{i_b+1}{{2\mu_b}\choose \mu_b}\psi_{m+r+b}^{i_b}$}
\\
\nonumber
&\text{\small $=\sum_{g,m,r\geq 0}\left(\sqrt 2 \epsilon\right)^{2g-2}\sum_{d_1,\dots,d_r\geq 1}\frac{(\lambda-1)^m(-1)^r}{m!r!}
\int_{\overline{\mathcal{M}}_{g,\ell+m+r}}\L^2(-1)\L\left(\frac 12\right)\prod_{a=1}^r\psi_a^{d_a+1}\prod_{b=1}^\ell \frac{\mu_b{{2\mu_b}\choose \mu_b}}{1-\mu_b\psi_{m+r+b}}$},
\\
\label{eqeq}
\end{align}
\normalsize
where in the last step we rename $k_a=d_a+1$, $d_a\geq 1$.

We can trade the $\psi_1,\dots,\psi_r$ classes in \eqref{eqeq} for a suitable combination of Mumford $\kappa$ classes, following ideas from \cite{BGF2017}.
Let $\pi:\overline{\mathcal M}_{g,\ell+m+r}\to\overline{\mathcal M}_{g,\ell+m}$ be the map forgetting the first $r$ marked points (and contracting the resulting unstable components), then we have the following iterated version of the \emph{dilaton equation}
\be
\pi_*\left((\pi^*\mathcal{X}) \prod_{a=1}^r\psi_a^{d_a+1}\right)=\mathcal{X}\sum_{\s\in\mathfrak{S}_r}\prod_{\g\in\mathrm{Cycles}(\s)}\kappa_{\sum_{a\in\g}d_a},\qquad d_a,\dots,d_r\geq 1,
\ee
for any $\mathcal{X}\in H^\bullet\left(\overline{\mathcal M}_{g,\ell+m},\Q\right)$. Here and below, $\mathfrak{S}_r$ is the group of permutations of $\{1,\dots,r\}$ and $\mathrm{Cycles}(\s)$ is the set of disjoint cycles in the permutation $\s$, $\s=\prod_{\g\in\mathrm{Cycles}(\s)}\g$.
In our case it is convenient to set
\be
\mathcal{X}=\L^2(-1)\L\left(\frac 12\right)\prod_{b=1}^\ell \frac{\mu_b{{2\mu_b}\choose \mu_b}}{1-\mu_b\psi_{m+b}},\qquad\pi^*\mathcal{X}=\L^2(-1)\L\left(\frac 12\right)\prod_{b=1}^\ell \frac{\mu_b{{2\mu_b}\choose \mu_b}}{1-\mu_b\psi_{m+r+b}},
\ee
so that the sum over $r\geq 0$ and $d_1,\dots,d_r\geq 1$ in \eqref{eqeq} can be expressed as
\be
\text{\small $\sum_{r\geq 0}\frac{(-1)^r}{r!}\sum_{d_1,\dots,d_r\geq 1}\int_{\overline{\mathcal M}_{g,\ell+m+r}}(\pi^*X)\prod_{a=1}^r\psi_a^{d_a+1}
=\sum_{r\geq 0}\frac{(-1)^r}{r!}\sum_{d_1,\dots,d_r\geq 1}\int_{\overline{\mathcal M}_{g,\ell+m}} X\sum_{\s\in\mathfrak{S}_r}\prod_{\g\in\mathrm{Cycles}(\s)}\kappa_{\sum_{a\in\g}d_a}$}.
\ee

Let us now recall that for any set of variables $F_1,F_2,\dots$, we have the identity of symmetric functions
\be
\exp\left(\sum_{r\geq 1}\frac{\xi^r}r F_r\right)=\sum_{\nu}\frac{\xi^{|\nu|}}{z_\nu}F_{\nu_1}\cdots F_{\nu_{\ell(\nu)}}
\ee
where the sum on the right side extends over the set of all partitions $\nu=(\nu_1,\dots,\nu_{\ell(\nu)})$, $|\nu|=\nu_1+\cdots+\nu_{\ell(\nu)}$, and $z_\nu$ has the same definition as above, namely $z_\nu:=\prod_{i\geq 1}\left(i^{m_i}\right)m_i!$, $m_i$ being the multiplicity of $i$ in the partition $\nu$. Applying this relation to
\be
F_r=\sum_{d_1,\dots,d_r\geq 1}\kappa_{\sum_{a=1}^td_a}=\sum_{d\geq r}{{d-1}\choose{r-1}}\kappa_d,\qquad \xi=-1,
\ee
since for any partition $\nu$ of $r$ the quantity $r!/z_\nu$ is the cardinality of the conjugacy class labeled by $\nu$ in $\mathfrak{S}_r$, we deduce that
\be
\sum_{r\geq 0}\frac {(-1)^r}{r!}\sum_{d_1,\dots,d_r\geq 1}\int_{\overline{\mathcal M}_{g,\ell+m}}\mathcal{X}\sum_{\s\in\mathfrak{S}_r}\prod_{\g\in{\mathrm Cycles}(\s)}\kappa_{\sum_{a\in\g}d_a}
=\int_{\overline{\mathcal M}_{g,\ell+m}}\mathcal{X}\exp\left(-\sum_{d\geq 1}\frac{\kappa_d}d\right)
\ee
where we also use the identity $\sum_{r\geq 1}\frac{(-1)^r}r{{d-1}\choose{r-1}}=-\frac 1d$. The proof is complete.
\end{proof}

\begin{example}\label{examplehodgeguegenus0}

Comparing the coefficients of $\epsilon^{-2}$ on  both sides of \eqref{hodgemonotonehurwitzformula} we obtain the following relation in genus zero
\be
\label{hodgehurwitzgenus0}
\mathscr{H}_{0,\mu}=2^{\ell-2}\lambda^{|\mu|+2-\ell}\sum_{s=1}^{|\mu|+1-\ell}H_0^>(\mu;s)
\ee
valid for any partition $\mu$ of length $\ell$. One can check that \eqref{hodgehurwitzgenus0} is consistent with the computations of Hurwitz numbers in genus zero performed in Example \ref{GenusZeroEx}.

E.g. for $\ell=1$ we compute the first terms in the $(\lambda-1)$-expansion of the left side of \eqref{hodgehurwitzgenus0}, directly from \eqref{hodgemonotonehurwitzformula},
\be
\label{LHSgenus0}
\mathscr{H}_{0,\mu=(\mu_1)}=\frac 12 \frac {1}{\mu_1+1}{{2\mu_1}\choose \mu_1}+\frac{(\lambda-1)}2{{2\mu_1}\choose \mu_1}+\frac{(\lambda-1)^2}4\mu_1{{2\mu_1}\choose\mu_1}+\O\left((\lambda-1)^3\right).
\ee
On the other hand, the right side of \eqref{hodgehurwitzgenus0} is computed as
\begin{align}
\nonumber
\frac 12 \lambda^{\mu_1+1}\sum_{s=1}^{\mu_1}H_0^>(\mu=(\mu_1);s)&=\frac 1{2\mu_1}\lambda^{\mu_1+1}\sum_{s=1}^{\mu_1}{{\mu_1}\choose s}{{\mu_1}\choose{s-1}}
\\
\nonumber
&=\frac 1{2\mu_1}\lambda^{\mu_1+1}{{2\mu_1}\choose {\mu_1-1}}
\\
\label{RHSgenus0}
&=\frac 1{2(\mu_1+1)}{{2\mu_1}\choose{\mu_1}}\sum_{b=0}^{\mu_1+1}{{\mu_1+1}\choose b}(\lambda-1)^b
\end{align}
where we use \eqref{GenusZeroPlanar} and the identity
\be
\sum_{s=1}^{\mu_1} {{\mu_1} \choose {s-1}}{{\mu_1} \choose s} = \sum_{s=0}^{\mu_1-1} {\mu_1 \choose s}{\mu_1 \choose {\mu_1-1-s}} = {{2\mu_1} \choose {\mu_1-1}},
\ee
which follows from the \emph{Chu-Vandermonde identity} $\sum_{s=0}^{k-1}{a\choose s}{b\choose{k-1-s}}={{a+b}\choose{k-1}}$ for $a=b=k=\mu_1$. Expressions \eqref{LHSgenus0} and \eqref{RHSgenus0} match.
\end{example}

\appendix

\section{Numerical Tables}\label{apptable}
\subsection{Tables of some weighted strictly monotone double Hurwitz numbers,  $H_{g}^>(\mu;s)$}
We recall that $H_{g}^>(\mu;s)=\frac{z_\mu}{|\mu|!}\sum\limits_{\nu\text{ of length }s}h_g^>(\mu;\nu)$, where $h_g^>(\mu;\nu)$ is the strictly monotone double Hurwitz number with partitions $\mu$ and $\nu$; see \eqref{eq122}.

\shrink{
\footnotesize
\begin{align*}
&\begin{array}[t]{|c|c|c|}
\hline \mu=(3,1) & g=0 & g=1\\
\hline s=1 & 3 & 3 \\
\hline s=2 & 9 & 0 \\
\hline s=3 & 3 & 0 \\
\hline
\end{array}
\quad
\begin{array}[t]{|c|c|c|}
\hline \mu=(3,2) & g=0 & g=1 \\
\hline s=1 & 6 & 18 \\
\hline s=2 & 30 & 18 \\
\hline s=3 & 30 & 0 \\
\hline s=4 & 6 & 0 \\
\hline
\end{array}
\quad
\begin{array}[t]{|c|c|c|c|}
\hline \mu=(3,3) & g=0 & g=1 & g=2 \\
\hline s=1 & 9 & 75 & 36 \\
\hline s=2 & 72 & 198 & 0 \\
\hline s=3 & 138 & 75 & 0 \\
\hline s=4 & 72 & 0 & 0 \\
\hline s=5 & 9 & 0 & 0 \\
\hline
\end{array}
\\
&\begin{array}[t]{|c|c|c|c|c|}
 \hline\mu=(4,4) & g=0 & g=1 & g=2 & g=3 \\
 \hline s=1 & 16 & 616 & 3304 & 1104 \\
 \hline s=2 & 264 & 4636 & 8132 & 0 \\
 \hline s=3 & 1200 & 8496 & 3304 & 0 \\
 \hline s=4 & 1940 & 4636 & 0 & 0 \\
 \hline s=5 & 1200 & 616 & 0 & 0 \\
 \hline s=6 & 264 & 0 & 0 & 0 \\
 \hline s=7 & 16 & 0 & 0 & 0 \\
 \hline
\end{array}
\quad
\begin{array}[t]{|c|c|c|c|c|}
 \hline \mu=(6, 3) & g=0 & g=1 & g=2 & g=3 \\
 \hline s=1 & 18 & 1428 & 16002 & 22872 \\
 \hline s=2 & 414 & 15120 & 70938 & 22872 \\
 \hline s=3 & 2598 & 43680 & 70938 & 0 \\
 \hline s=4 & 6210 & 43680 & 16002 & 0 \\
 \hline s=5 & 6210 & 15120 & 0 & 0 \\
 \hline s=6 & 2598 & 1428 & 0 & 0 \\
 \hline s=7 & 414 & 0 & 0 & 0 \\
 \hline s=8 & 18 & 0 & 0 & 0 \\
\hline
\end{array}
\\
&\begin{array}[t]{|c|c|}
 \hline \mu=(2,1,1) & g=0\\
\hline s=1 & 6 \\
\hline s=2 & 6 \\
\hline
\end{array}
\quad
\begin{array}[t]{|c|c|c|}
 \hline \mu=(2,2,1) & g=0 & g=1 \\
 \hline s=1 & 16 & 8 \\
 \hline s=2 & 40 & 0 \\
 \hline s=3 & 16 & 0 \\
\hline
\end{array}
\quad
\begin{array}[t]{|c|c|c|}
 \hline \mu=(2,2,2) & g=0 & g=1 \\
\hline s=1 & 40 & 80 \\
\hline s=2 & 176 & 80 \\
\hline s=3 & 176 & 0 \\
\hline s=4 & 40 & 0 \\
\hline
\end{array}
\end{align*}
}

\shrink{
\footnotesize
\begin{align*}
&\begin{array}[t]{|c|c|c|c|c|c|}
 \hline \mu=(4, 4, 4) & g=0 & g=1 & g=2 & g=3 & g=4 \\
 \hline s=1 & 704 & 89760 & 2631552 & 18161440 & 19033344 \\
 \hline s=2 & 21312 & 1568640 & 24587904 & 75241920 & 19033344 \\
 \hline s=3 & 204480 & 8507520 & 66562944 & 75241920 & 0 \\
 \hline s=4 & 843648 & 18934080 & 66562944 & 18161440 & 0 \\
 \hline s=5 & 1673856 & 18934080 & 24587904 & 0 & 0 \\
 \hline s=6 & 1673856 & 8507520 & 2631552 & 0 & 0 \\
 \hline s=7 & 843648 & 1568640 & 0 & 0 & 0 \\
 \hline s=8 & 204480 & 89760 & 0 & 0 & 0 \\
 \hline s=9 & 21312 & 0 & 0 & 0 & 0 \\
 \hline s=10 & 704 & 0 & 0 & 0 & 0 \\
\hline \end{array}
\\
&\begin{array}[t]{|c|c|c|c|c|}
 \hline \mu=(4, 3, 2, 1) & g=0 & g=1 & g=2 & g=3 \\
 \hline s=1 & 1728 & 54432 & 235872 & 70848 \\
 \hline s=2 & 26136 & 379512 & 570672 & 0 \\
 \hline s=3 & 111024 & 680832 & 235872 & 0 \\
 \hline s=4 & 175824 & 379512 & 0 & 0 \\
 \hline s=5 & 111024 & 54432 & 0 & 0 \\
 \hline s=6 & 26136 & 0 & 0 & 0 \\
 \hline s=7 & 1728 & 0 & 0 & 0 \\
\hline \end{array}
\quad
\begin{array}[t]{|c|c|c|c|}
 \hline \mu=(2, 2, 2, 2) & g=0 & g=1 & g=2 \\
 \hline s=1 & 672 & 3360 & 1008 \\
 \hline s=2 & 4464 & 8016 & 0 \\
 \hline s=3 & 7872 & 3360 & 0 \\
 \hline s=4 & 4464 & 0 & 0 \\
 \hline s=5 & 672 & 0 & 0 \\
\hline
\end{array}
\\
&\begin{array}[t]{|c|c|c|c|c|c|c|}
 \hline \mu=(5, 4, 4, 2) & g=0 & g=1 & g=2 & g=3 & g=4 & g=5 \\
 \hline s=1 & 29120 & 7047040 & 444924480 & 8434666240 & 42317475200 & 35974149120 \\
 \hline s=2 & 1212800 & 180513600 & 6829912320 & 71893480000 & 168041817600 & 35974149120 \\
 \hline s=3 & 16616960 & 1529449920 & 33913376640 & 186374568640 & 168041817600 & 0 \\
 \hline s=4 & 103248000 & 5796138240 & 72317482560 & 186374568640 & 42317475200 & 0 \\
 \hline s=5 & 331189440 & 11030467200 & 72317482560 & 71893480000 & 0 & 0 \\
 \hline s=6 & 584935680 & 11030467200 & 33913376640 & 8434666240 & 0 & 0 \\
 \hline s=7 & 584935680 & 5796138240 & 6829912320 & 0 & 0 & 0 \\
 \hline s=8 & 331189440 & 1529449920 & 444924480 & 0 & 0 & 0 \\
 \hline s=9 & 103248000 & 180513600 & 0 & 0 & 0 & 0 \\
 \hline s=10 & 16616960 & 7047040 & 0 & 0 & 0 & 0 \\
 \hline s=11 & 1212800 & 0 & 0 & 0 & 0 & 0 \\
 \hline s=12 & 29120 & 0 & 0 & 0 & 0 & 0 \\
\hline \end{array}
\\
&\begin{array}[t]{|c|c|c|}
 \hline \mu=(2, 2, 2, 1, 1) & g=0 & g=1 \\
 \hline s=1 & 1680 & 3360 \\
 \hline s=2 & 7392 & 3360 \\
 \hline s=3 & 7392 & 0 \\
 \hline s=4 & 1680 & 0 \\
\hline
\end{array}
\quad
\begin{array}[t]{|c|c|c|c|c|}
 \hline \mu=(3, 3, 2, 2, 2) & g=0 & g=1 & g=2 & g=3 \\
 \hline s=1 & 71280 & 2661120 & 18461520 & 18722880 \\
 \hline s=2 & 1206144 & 23973840 & 75182256 & 18722880 \\
 \hline s=3 & 6314976 & 63697968 & 75182256 & 0 \\
 \hline s=4 & 13791600 & 63697968 & 18461520 & 0 \\
 \hline s=5 & 13791600 & 23973840 & 0 & 0 \\
 \hline s=6 & 6314976 & 2661120 & 0 & 0 \\
 \hline s=7 & 1206144 & 0 & 0 & 0 \\
 \hline s=8 & 71280 & 0 & 0 & 0 \\
\hline
\end{array}
\end{align*}
}

\subsection{Tables of some weighted weakly monotone double Hurwitz numbers $H_g^\geq(\mu;s)$}

We recall that  $H_g^\geq(\mu;s) =\frac{z_\mu}{|\mu|!}\sum\limits_{\nu\text{ of length }s}h_g^\geq(\mu;\nu),
$ where $h_g^\geq(\mu;\nu)$ is the weakly monotone double Hurwitz number with partitions $\mu$ and $\nu$; compare with \eqref{eq122}.

In general, $H_{g}^\ge(\mu;s) \not = 0$ for every $s \leq |\mu|$ and $g \ge 0$. We calculate $H_{g}^\ge(\mu;s)$ for the first few values of $g$.

\shrink{
\begin{align*}
&\begin{array}[t]{|c|c|c|c|}
\hline \mu=(3, 1) & g=0 & g=1 & g=2 \\
\hline s=1 & 3 & 45 & 483 \\
\hline s=2 & 18 & 255 & 2688 \\
\hline s=3 & 30 & 420 & 4410 \\
\hline s=4 & 15 & 210 & 2205 \\
\hline
\end{array}
\quad
\begin{array}[t]{|c|c|c|c|}
 \hline \mu=(3, 2) & g=0 & g=1 & g=2 \\
\hline s=1 & 6 & 168 & 3402 \\
\hline s=2 & 54 & 1464 & 29058 \\
\hline s=3 & 156 & 4176 & 82212 \\
\hline s=4 & 180 & 4800 & 94260 \\
\hline s=5 & 72 & 1920 & 37704 \\
\hline
\end{array}
\quad
\begin{array}[t]{|c|c|c|c|}
 \hline \mu=(3, 3) & g=0 & g=1 & g=2 \\
\hline s=1 & 9 & 462 & 16443 \\
\hline s=2 & 117 & 5742 & 197559 \\
\hline s=3 & 516 & 24660 & 833472 \\
\hline s=4 & 1008 & 47580 & 1594836 \\
\hline s=5 & 900 & 42300 & 1413720 \\
\hline s=6 & 300 & 14100 & 471240 \\
\hline
\end{array}
\\
&\begin{array}[t]{|c|c|c|c|c|c|}
 \hline \mu=(1, 1, 1) & g=0 & g=1 & g=2 & g=3 & g=4 \\
 \hline s=1 & 4 & 20 & 84 & 340 & 1364  \\
 \hline s=2 & 12 & 60 & 252 & 1020 & 4092  \\
 \hline s=3 & 8 & 40 & 168 & 680 & 2728  \\
\hline \end{array}
\quad
\begin{array}[t]{|c|c|c|c|c|}
 \hline \mu=(3,2,1) & g=0 & g=1 & g=2 & g=3 \\
 \hline s=1 & 42 & 2268 & 81774 & 2498496 \\
 \hline s=2 & 558 & 28248 & 982326 & 29405736 \\
 \hline s=3 & 2472 & 121320 & 4143024 & 122714160 \\
 \hline s=4 & 4836 & 234060 & 7926312 & 233606280 \\
 \hline s=5 & 4320 & 208080 & 7025760 & 206699040 \\
 \hline s=6 & 1440 & 69360 & 2341920 & 68899680 \\
\hline \end{array}
\\
&\begin{array}[t]{|c|c|c|c|c|c|c|}
 \hline \mu=(5, 3, 2) & g=0 & g=1 & g=2 & g=3 & g=4 & g=5 \\
 \hline s=1 & 330 & 98670 & 17117100 & 2288397540 & 262779844470 & 27370788935490 \\
 \hline s=2 & 11790 & 3139530 & 508126980 & 64989626220 & 7244914364850 & 739256601861510 \\
 \hline s=3 & 151140 & 37555800 & 5814501240 & 722008428240 & 78865374260700 & 7932095991173640 \\
 \hline s=4 & 973200 & 231506100 & 34809669720 & 4236585517200 & 456285210221400 & 45429895491347220 \\
 \hline s=5 & 3600180 & 832748640 & 122812524600 & 14745786668160 & 1572851081541420 & 155505293985110400 \\
 \hline s=6 & 8126700 & 1846504080 & 268910866680 & 31999520486160 & 3391243294051140 & 333707416656660000 \\
 \hline s=7 & 11380320 & 2557716000 & 369587047200 & 43733298023520 & 4615886297332800 & 452853891923025600 \\
 \hline s=8 & 9649080 & 2155587000 & 310123401000 & 36581098895880 & 3852087017209200 & 377274782175656400 \\
 \hline s=9 & 4536000 & 1010772000 & 145151092800 & 17098516260000 & 1798743628584000 & 176040872796600000 \\
 \hline s=10 & 907200 & 202154400 & 29030218560 & 3419703252000 & 359748725716800 & 35208174559320000 \\
\hline \end{array}
\\
&\begin{array}[t]{|c|c|c|c|c|c|c|}
 \hline \mu=(1, 1, 1, 1) & g=0 & g=1 & g=2 & g=3 & g=4 & g=5 \\
 \hline s=1 & 30 & 420 & 4410 & 42240 & 390390 & 3554460 \\
 \hline s=2 & 174 & 2364 & 24498 & 233328 & 2151222 & 19565892 \\
 \hline s=3 & 288 & 3888 & 40176 & 382176 & 3521664 & 32022864 \\
 \hline s=4 & 144 & 1944 & 20088 & 191088 & 1760832 & 16011432 \\
\hline \end{array}
\\
&\begin{array}[t]{|c|c|c|c|c|c|c|}
\hline \mu=(2, 2, 1, 1) & g=0 & g=1 & g=2 & g=3 & g=4 & g=5 \\
\hline s=1 & 224 & 11760 & 417648 & 12652640 & 353825472 & 9465041040 \\
\hline s=2 & 2936 & 145560 & 5001792 & 148676240 & 4111488168 & 109250057640 \\
\hline s=3 & 12912 & 623088 & 21061152 & 619916064 & 17042443920 & 451231651728 \\
\hline s=4 & 25176 & 1200264 & 40262736 & 1179630192 & 32339018280 & 854769872184 \\
\hline s=5 & 22464 & 1066464 & 35678592 & 1043606592 & 28581355584 & 754984855584 \\
\hline s=6 & 7488 & 355488 & 11892864 & 347868864 & 9527118528 & 251661618528 \\
\hline \end{array}
\\
&\begin{array}[t]{|c|c|c|c|c|c|c|}
 \hline \mu=(3, 2, 2, 1) & g=0 & g=1 & g=2 & g=3 & g=4 & g=5 \\
 \hline s=1 & 1080 & 142560 & 11891880 & 808030080 & 49030839000 & 2777130588960 \\
 \hline s=2 & 24408 & 2975688 & 236613384 & 15604156944 & 928759785048 & 51934912866648 \\
 \hline s=3 & 195696 & 22833936 & 1764985248 & 114273524448 & 6718979907216 & 372620872120176 \\
 \hline s=4 & 764208 & 86946408 & 6607836864 & 423012867984 & 24682857466608 & 1361716707058488 \\
 \hline s=5 & 1622160 & 181944000 & 13692581280 & 870735528000 & 50576815946160 & 2781487931040000 \\
 \hline s=6 & 1911600 & 212829120 & 15934474080 & 1009718844480 & 58506896866320 & 3212163320083200 \\
 \hline s=7 & 1175040 & 130440960 & 9745954560 & 616691715840 & 35698249900800 & 1958572008345600 \\
 \hline s=8 & 293760 & 32610240 & 2436488640 & 154172928960 & 8924562475200 & 489643002086400 \\
\hline \end{array}
\\
&\begin{array}[t]{|c|c|c|c|c|c|}
 \hline \mu=(3, 3, 3, 3) & g=0 & g=1 & g=2 & g=3 & g=4 \\
 \hline s=1 & 14742 & 6781320 & 1863064476 & 397980044280 & 73027276324002 \\
 \hline s=2 & 684774 & 286543656 & 73938326364 & 15124478632344 & 2690423275640562 \\
 \hline s=3 & 11927088 & 4700315952 & 1162209509712 & 230530176869328 & 40089332784598560 \\
 \hline s=4 & 108506304 & 41049414576 & 9847619855856 & 1910059732782864 & 326635075616752080 \\
 \hline s=5 & 591049872 & 217264375440 & 50997568912848 & 9730568084094000 & 1643434518194147520 \\
 \hline s=6 & 2065978224 & 744104821680 & 171941934622896 & 32417467690208400 & 5425295582074933440 \\
 \hline s=7 & 4798180800 & 1703613513600 & 389301061256640 & 72772493528332800 & 12099023079466665600 \\
 \hline s=8 & 7485955200 & 2632114958400 & 596891523260160 & 110918372096491200 & 18356651181359395200 \\
 \hline s=9 & 7754940000 & 2709582840000 & 611410862412000 & 113177279163888000 & 18674140608815688000 \\
 \hline s=10 & 5114988000 & 1780691688000 & 400648862930400 & 73995902520393600 & 12187705122917006400 \\
 \hline s=11 & 1944000000 & 675695520000 & 151836376608000 & 28014789102336000 & 4610660182447564800 \\
 \hline s=12 & 324000000 & 112615920000 & 25306062768000 & 4669131517056000 & 768443363741260800 \\
\hline \end{array}
\end{align*}
}

\subsection{Tables of some positive LUE Correlators}
We write the correlators in terms of $N$ and the parameter $\a=N(c-1)$.

\shrink{
\begin{alignat*}{1}
\langle \tr X\, \tr X \rangle_{c} ={}& N (\a +N), \\
 \langle \tr X^2\tr X \rangle_{c} ={}& 2 N(\a+N)(\a+2N), \\    
\langle \tr X^2\tr X^2\rangle_{c} ={}& 2 \a \left(1+2 \a^2\right) N+2 \left(1+11 \a^2\right) N^2+36 \a N^3+18 N^4, \\
\langle \tr X^3\tr X^1\rangle_{c} ={}& 3 N \left (\a+N\right ) \left (1+\a^2+5 \a N+5 N^2\right )\\
\langle \tr X^3\tr X^2\rangle_{c} ={}& 6 N \left (\a+N\right ) \left (\a+2 N\right ) \left (3+\a^2+6 \a N+6 N^2\right ) \\
\langle \tr X^3\tr X^3\rangle_{c} ={}& 3 N \left (\a+N\right ) \left (12+25 \a^2+3 \a^4+4 \a \left (29+9 \a^2\right ) N+4 \left (29+34 \a^2\right ) N^2+200 \a N^3+100 N^4\right ),
\\
\langle \tr X\, \tr X\, \tr X \rangle_{c} ={}& 2N(\a+N),
\\
\langle \tr X^2\tr X\, \tr X \rangle_{c} ={}& 6 N(\a+N)(\a+2N),
\\
    \langle \tr X^2\tr X^2\tr X \rangle_{c} ={}& 8 \a \left(1+2 \a^2\right) N+8 \left(1+11 \a^2\right) N^2+144 \a N^3+72 N^4,
\\
    \langle \tr X^2\tr X^2\tr X^2\rangle_{c} ={}& 40 \a^2 \left(2+\a^2\right) N+48 \left(5 \a+7 \a^3\right) N^2+16 \left(10+59 \a^2\right) N^3+1080 \a N^4+432 N^5,  \\
     \langle \tr X^4\tr X^3\tr X^2\rangle_{c} ={}& 24 \a \left(328+1092 \a^2+252 \a^4+8 \a^6\right) N+24 \left(328+1092 \a^2+252 \a^4+8 \a^6+\a \left(4826 \a+2765 \a^3+169 \a^5\right)\right) N^2+ \\ & + 24 \left(4826 \a+2765 \a^3+169 \a^5+\a \left(4826+9935 \a^2+1239 \a^4\right)\right) N^3+24 \left(4826+9935 \a^2+1239 \a^4+ \right. \\ & \left.+\a \left(14340 \a+4240 \a^3\right)\right) N^4+24 \left(14340 \a+4240 \a^3+\a \left(7170+7370 \a^2\right)\right) N^5+24 \left(7170+13670 \a^2\right) N^6+ \\ & +201600 \a N^7+50400 N^8,
\\
\langle \tr X\, \tr X\, \tr X\, \tr X \rangle_{c} ={}& 6 N ( \a+N),
\\
  \langle \tr X^2\tr X\, \tr X\, \tr X \rangle_{c} ={}& 24 N(\a+N)(\a+2N),
\\
\langle \tr X^2\tr X^2\tr X\, \tr X \rangle_{c} ={}& 40 \a \left(1+2 \a^2\right) N+40 \left(1+11 \a^2\right) N^2+720 \a N^3+360 N^4,   
\\
    \langle \tr X^2\tr X^2\tr X^2\tr X \rangle_{c} ={}& 240 \a^2 \left(2+\a^2\right) N+288 \left(5 \a+7 \a^3\right) N^2+96 \left(10+59 \a^2\right) N^3+6480 \a N^4+2592 N^5,    \\
    \langle \tr X^2\tr X^2\tr X^2\tr X^2\rangle_{c} ={}& 48 \a \left(21+14 \a^2 \left(5+\a^2\right)\right) N+48 \left(21+377 \a^2+163 \a^4\right) N^2+96 \left(307 \a+338 \a^3\right) N^3+  \\ & +48 \left(307+1283 \a^2\right) N^4+54432 \a N^5+18144 N^6, 
\\
\langle \tr X\, \tr X\, \tr X\, \tr X\, \tr X \rangle_{c}={}& 24 N(\a+N), \\
    \langle \tr X^2\tr X\, \tr X\, \tr X\, \tr X \rangle_{c} ={}& 120 N(\a+N)(\a+2N), 
\\
 \langle \tr X^2\tr X^2\tr X\, \tr X\, \tr X \rangle_{c} ={}& 240 \a \left(1+2 \a^2\right) N+240 \left(1+11 \a^2\right) N^2+4320 \a N^3+2160 N^4,   \\
      \langle \tr X^3\tr X^2\tr X^2\tr X^2\tr X \rangle_{c} ={}& 6048 \a^2 \left(11+\a^2\right) \left(3+2 \a^2\right) N+432 \left(42 \a \left(11+\a^2\right) \left(3+2 \a^2\right)+3 \a^3 \left(611+121 \a^2\right)\right) N^2+ \\ & +432 \left(28 \left(11+\a^2\right) \left(3+2 \a^2\right)+39 \a^2 \left(47+37 \a^2\right)+9 \a^2 \left(611+121 \a^2\right)\right) N^3+ 432 \left(2160 \a^3+  \right. \\ & \left. +117 \a \left(47+37 \a^2\right)+6 \a \left(611+121 \a^2\right)\right) N^4+2592 \left(611+1741 \a^2\right) N^5+3265920 \a N^6+933120 N^7, 
\end{alignat*}
}

\subsection{Tables of some negative LUE Correlators}

In the following formul\ae\ we denote $a_j :=(\a-j)_{2j+1}$,  $\a=N(c-1)$.

\shrink{
\begin{alignat*}{1}
    \langle \tr X^{-1}\tr X^{-1}\rangle_{c} ={}& \frac{1}{a_1 a_0} N (\a+N),   
    \\
    \langle \tr X^{-2}\tr X^{-1}\rangle_{c} ={}& \frac{1}{a_2 a_0} 2 N (\a+N) (\a+2 N), \\
    \langle \tr X^{-2}\tr X^{-2}\rangle_{c} ={}& \frac{1}{a_3 a_1} 2 N (\a+N) \left(2 \a^4+9 \a^3 N+\a^2 \left(9 N^2-5\right)-21 \a N-21 N^2+3\right),   \\
    \langle \tr X^{-3}\tr X^{-1}\rangle_{c} ={}& \frac{1}{a_3 a_0} 3 N (\a+N) \left(\a^2+5 \a N+5 N^2+1\right),  \\
    \langle \tr X^{-3}\tr X^{-2}\rangle_{c} ={}& \frac{1}{a_4 a_1} 6 N \left (\a+N\right ) \left (\a+2 N\right ) \left (2+\a^4-26 \a N+6 \a^3 N-26 N^2+\a^2 \left (-3+6 N^2\right )\right ), \\
   \langle \tr X^{-3}\tr X^{-3}\rangle_{c} ={}& \frac{1}{a_5 a_2} 3 N \left (\a+N\right ) \left (320-444 \a^2+147 \a^4-26 \a^6+3 \a^8+4 \a \left (-200+411 \a^2-100 \a^4+9 \a^6\right ) N \right. \\ & \left. +4 \left (-200+1731 \a^2-425 \a^4+34 \a^6\right ) N^2+40 \a \left (264-65 \a^2+5 \a^4\right ) N^3+20 \left (264-65 \a^2+5 \a^4\right ) N^4\right ), 
\\
    \langle \tr X^{-1}\tr X^{-1}\tr X^{-1}\rangle_{c} ={}& \frac{1}{ a_2 a_{0}^2 }  4 N (\a+N) (\a+2 N),   \\
    \langle \tr X^{-2}\tr X^{-1}\tr X^{-1}\rangle_{c} ={}& \frac{1}{ a_{3} a_{1} a_{0} } 2 N (\a+N) \left(5 \a^4+24 \a^3 N+\a^2 \left(24 N^2-5\right)-36 \a N-36 N^2\right) ,   \\
    \langle \tr X^{-2}\tr X^{-2}\tr X^{-1}\rangle_{c} ={}& \frac{1}{ a_{4} a_{1} a_{0} }24 N (\a+N) (\a+2 N) \left(\a^4+6 \a^3 N+\a^2 \left(6 N^2-3\right)-26 \a N-26 N^2+2\right),
\\
 \langle \tr X^{-2}\tr X^{-2}\tr X^{-2}\rangle_{c} ={}& \frac{1}{ a_{5} a_{1}^2 } 8 N (\a+N) \left(7 \a^2 \left(\a^2-7\right) \left(\a^2-1\right)^2+36 \left(6 \a^4-71 \a^2+125\right) N^4+ 72 \a \left(6 \a^4-71 \a^2+125\right) N^3+ \right. \\ & \left.9 (\a-1) \a (\a+1) \left(9 \a^4-89 \a^2+100\right) N+9 \left(33 \a^6-382 \a^4+689 \a^2-100\right) N^2\right),
       \end{alignat*}
}

\shrink{
 \begin{alignat*}{1}
    \langle \tr X^{-3}\tr X^{-2}\tr X^{-2}\rangle_{c} ={}& \frac{1}{a_6 a_2 a_1} 24 N (\a+N) (\a+2 N) \left(2 \a^2 \left(\a^2-1\right)^2 \left(2 \a^4-25 \a^2+68\right)+60 \left(3 \a^6-71 \a^4+488 \a^2-840\right) N^4+ \right. 
    \\     
    & \left. +120 \a \left(3 \a^6-71 \a^4+488 \a^2-840\right) N^3+(\a-1) \a (\a+1) \left(57 \a^6-1105 \a^4+6148 \a^2-7200\right) N+ \right. \\ & \left. +\left(237 \a^8-5422 \a^6+36533 \a^4-63748 \a^2+7200\right) N^2\right), 
\\
    \langle \tr X^{-1}\tr X^{-1}\tr X^{-1}\tr X^{-1}\rangle_{c} ={}& \frac{1}{a_3 a_1 a_0^2} 6 N (\a+N) \left(5 \a^4+24 \a^3 N+\a^2 \left(24 N^2-5\right)-36 \a N-36 N^2\right),   \\
    \langle \tr X^{-2}\tr X^{-1}\tr X^{-1}\tr X^{-1}\rangle_{c} ={}& \frac{1}{a_4 a_1 a_0^2}12 N (\a+N) (\a+2 N) \left(7 \a^4+44 \a^3 N+\a^2 \left(44 N^2-7\right)-144 \a N-144 N^2\right),   \\
    \langle \tr X^{-2}\tr X^{-2}\tr X^{-1}\tr X^{-1}\rangle_{c} ={}& \frac{1}{a_5 a_2 a_1 a_0} 8 N (\a+N) \left(14 \a^2 \left(\a^2-1\right)^2 \left(2 \a^4-13 \a^2+20\right)+12 \left(78 \a^6-1055 \a^4+4237 \a^2-3800\right) N^4 + \right. \\ & \left.+24 \a \left(78 \a^6-1055 \a^4+4237 \a^2-3800\right) N^3+3 (\a-1) \a (\a+1) \left(113 \a^6-1205 \a^4+3632 \a^2-800\right) N+  \right. \\ & \left.+ 3 \left(425 \a^8-5538 \a^6+21785 \a^4-19632 \a^2+800\right) N^2\right),   \\
\end{alignat*}
}

\subsection{Tables of some mixed LUE Correlators.}

In the following formul\ae\ we denote $a_j :=(\a-j)_{2j+1}$,  $\a=N(c-1)$. 

\shrink{
\begin{alignat*}{1}
 \langle \tr X^{-1}\tr X^{1}\rangle_{c} ={}& -\frac{1}{a_0}N,
 \\
 \langle \tr X^{-2}\tr X^{1}\rangle_{c} ={}& -\frac{1}{a_1}2 N (\a+N),  \\
 \langle \tr X^{-1}\tr X^{2}\rangle_{c} ={}& -\frac{1}{a_0}2 N (\a+N),  \\
 \langle \tr X^{-2}\tr X^{2}\rangle_{c} = {}&-\frac{1}{a_1} (\a+2 N) \left(\a^2+2 \a N+2 N^2-1\right), \\
\langle \tr X^{-3}\tr X^{1}\rangle_{c} ={}& -\frac{1}{a_2} 3 N \left (\a+N\right ) \left (\a+2 N\right ), \\
\langle \tr X^{-3}\tr X^{2}\rangle_{c} ={}& -\frac{1}{a_2} 6 N \left (\a+N\right ) \left (-2+\a^2+2 \a N+2 N^2 \right ), \\
\langle \tr X^{-1}\tr X^{3}\rangle_{c} ={}& -\frac{1}{a_0} 3 N \left (\a+N\right ) \left (\a+2 N\right ), \\
\langle \tr X^{-2}\tr X^{3}\rangle_{c} ={}& -\frac{1}{a_1} 6 N \left (\a+N\right ) \left (-1+\a^2+2 \a N+2 N^2\right ), \\
\langle \tr X^{-3}\tr X^{3}\rangle_{c} ={}& -\frac{1}{a_2} 3 \left (2 \a \left (4-5 \a^2+\a^4\right )+\left (4-11 \a^2+3 \a^4\right ) N+6 \a \left (-3+2 \a^2\right ) N^2+4 \left (-3+7 \a^2\right ) N^3+30 \a N^4+12 N^5\right ), 
\\
 \langle \tr X^{-1}\tr X^{1}\tr X^{1}\rangle_{c} ={}& 0, 
\\
 \langle \tr X^{-1}\tr X^{-1}\tr X^{1}\rangle_{c} ={}& -\frac{1}{a_0 a_1}2 N (\a+N),  \\
   \langle \tr X^{-2}\tr X^{-1}\tr X^{1}\rangle_{c} ={}&  -\frac{1}{a_2 a_0} 6 N (\a+N) (\a+2 N),  
\\ \langle \tr X^{-1}\tr X^{2}\tr X^{1}\rangle_{c} ={}& -\frac{1}{a_0}2 N (\a+N),  \\
  \langle \tr X^{-2}\tr X^{1}\tr X^{1}\rangle_{c} ={}&  \frac{1}{a_1}2 N (\a+N), \\
\langle \tr X^{-1}\tr X^{-1}\tr X^{2}\rangle_{c} ={}& -\frac{1}{a_0 a_1} 2 N (\a+N) (\a+2 N), \\
   \langle \tr X^{-1}\tr X^{2}\tr X^{2}\rangle_{c} ={}& -\frac{1}{a_0}8 N (\a+N) (\a+2 N), \\
\\ 
\langle \tr X^{-2}\tr X^{2}\tr X^{1}\rangle_{c} ={}& 0, \\
 \langle \tr X^{-2}\tr X^{-2}\tr X^{1}\rangle_{c} ={}& -\frac{1}{a_3 a_1}8 N (\a+N) \left(2 \a^4+9 \a^3 N+\a^2 \left(9 N^2-5\right)-21 \a N-21 N^2+3\right), \\
 \langle \tr X^{-2}\tr X^{-1}\tr X^{2}\rangle_{c} ={}& -\frac{1}{a_2 a_0}8 N (\a + N) (-1 + \a^2 + 3 \a N + 3 N^2), \\
 \langle \tr X^{-1}\tr X^{-1}\tr X^{3}\rangle_{c} ={}& -\frac{1}{a_1 a_0} 12 N^2 (\a+N)^2, \\
 \langle \tr X^{-1}\tr X^{1}\tr X^{3}\rangle_{c} ={}& -\frac{1}{a_0} 6 N (\a+N) (\a+2 N), \\
 \langle \tr X^{-3}\tr X^{-1}\tr X^{1}\rangle_{c} ={}& -\frac{1}{a_3 a_0} 12 N (\a+N) \left(\a^2+5 \a N+5 N^2+1\right), \\
 \langle \tr X^{-3}\tr X^{1}\tr X^{1}\rangle_{c} ={}&  \frac{1}{a_2}6 N (\a+N) (\a+2 N),
\\
 \langle \tr X^{-3}\tr X^{-2}\tr X^{1}\rangle_{c} ={}& -\frac{1}{a_4 a_1}30 N (\a+N) (\a+2 N) \left(\a^4+6 \a^3 N+\a^2 \left(6 N^2-3\right)-26 \a N-26 N^2+2\right),
  \end{alignat*}
}

\shrink{
\begin{alignat*}{1}
  \langle \tr X^{-3}\tr X^{-2}\tr X^{2}\rangle_{c} ={}& -\frac{1}{a_4 a_1}24 N (\a+N) \left(2 \a^6-17 \a^4+10 \left(3 \a^2-13\right) N^4+20 \left(3 \a^2-13\right) \a N^3+23 \a^2+\right. \\ & \left. +2 \left(23 \a^4-120 \a^2+57\right) N^2+ 2 \left(8 \a^4-55 \a^2+57\right) \a N-8\right), 
\\
\langle \tr X^{-1}\tr X^{1}\tr X^{1}\tr X^{1}\rangle_{c} ={}& 0,\\
\langle \tr X^{-1}\tr X^{-1}\tr X^{1}\tr X^{1}\rangle_{c} ={}& \frac{1}{a_1 a_0}2 N (\a+N),\\
\langle \tr X^{-1}\tr X^{-1}\tr X^{-1}\tr X^{1}\rangle_{c} ={}& -\frac{1}{a_2 a_0^2} 12 N (\a + N) (\a + 2 N), 
\\ \langle \tr X^{-2}\tr X^{1}\tr X^{1}\tr X^{1}\rangle_{c} =&{} 0,  \\
 \langle \tr X^{-2}\tr X^{-2}\tr X^{1}\tr X^{1}\rangle_{c} ={}&  \frac{1}{a_3 a_1}24 N (\a+N) \left(2 \a^4+9 \a^3 N+\a^2 \left(9 N^2-5\right)-21 \a N-21 N^2+3\right),  \\
  \langle \tr X^{-2}\tr X^{-2}\tr X^{-1}\tr X^{1}\rangle_{c} ={}& -\frac{1}{a_4 a_1 a_0} 120 N (\a+N) (\a+2 N) \left(\a^4+6 \a^3 N+\a^2 \left(6 N^2-3\right)-26 \a N-26 N^2+2\right),  \\
\langle \tr X^{-2}\tr X^{-1}\tr X^{-1}\tr X^{2} \rangle_{c} ={}& -\frac{1}{a_3 a_1 a_0} 24  N (\a + N) (-1 + \a + 2 N) (\a + 2 N) (1 + \a + 2 N) (2 \a^2-3), \\
\langle \tr X^{-2}\tr X^{-1}\tr X^{2}\tr X^{1} \rangle_{c} ={}& \frac{1}{a_2 a_0} 8 N (\a + N) (-1 + \a^2 + 3 \a N + 3 N^2), \qquad \langle \tr X^{-2}\tr X^{2}\tr X^{1}\tr X^{1} \rangle_{c} =  0, \\
\langle \tr X^{-1}\tr X^{2}\tr X^{2}\tr X^{1} \rangle_{c} ={}& -\frac{1}{a_0}24 N (\a + N) (\a + 2 N),
\\
 \langle \tr X^{-1}\tr X^{-1}\tr X^{2}\tr X^{2} \rangle_{c} ={}& -\frac{1}{a_1 a_0} 24 N^2 (\a + N)^2, \\
 \langle \tr X^{-3}\tr X^{-1}\tr X^{-1}\tr X^{2}\rangle_{c} ={}& -\frac{1}{a_4 a_1 a_0}24 N (\a+N) \left(100 \left(\a^2-2\right) N^4+200 \a \left(\a^2-2\right) N^3+ \right. \\ & \left. +2 \left(73 \a^4-165 \a^2+52\right) N^2+2 \a \left(23 \a^4-65 \a^2+52\right) N+5 \a^2 \left(\a^4-3 \a^2+2\right)\right),  
\end{alignat*}
}

\subsection{Topological Expansion of some mixed correlators}

We compute the first terms in the large $N$ expansion for some of the above mixed correlators; compare with Proposition \ref{prop13}.

\shrink{
\begin{align*}
\langle \tr X^{-2}\tr X^{3} \rangle_{c} ={}& -N \left(\frac{12}{(c-1)^3}+\frac{24}{(c-1)^2}+\frac{18}{c-1}+6\right)-\frac{1}{N}\left(\frac{12}{(c-1)^5}+\frac{24}{(c-1)^4}+\frac{12}{(c-1)^3}\right)+ \\ 
 & -\frac{1}{N^3}\left(\frac{12}{(c-1)^7}+\frac{24}{(c-1)^6}+\frac{12}{(c-1)^5}\right)-\frac{1}{N^5}\left(\frac{12}{(c-1)^9}+\frac{24}{(c-1)^8}+\frac{12}{(c-1)^7}\right)+O\left(\frac{1}{N^7}\right), 
 \\ \langle \tr X^{-4}\tr X^{4} \rangle_{c} ={}& -\left(\frac{400}{(c-1)^7}+\frac{1400}{(c-1)^6}+\frac{1968}{(c-1)^5}+\frac{1420}{(c-1)^4}+\frac{560}{(c-1)^3}+\frac{120}{(c-1)^2}+\frac{16}{c-1}+2\right)+
\\
& - \frac{1}{N^2} \left( \frac{5600}{(c-1)^9}+\frac{19600}{(c-1)^8}+\frac{26920}{(c-1)^7}+\frac{18300}{(c-1)^6}+\frac{6320}{(c-1)^5}+\frac{980}{(c-1)^4}+\frac{40}{(c-1)^3}\right)+
\\
& - \frac{1}{N^4} \left( \frac{58800}{(c-1)^{11}}+\frac{205800}{(c-1)^{10}}+\frac{280448}{(c-1)^9}+\frac{186620}{(c-1)^8}+\frac{61560}{(c-1)^7}+\frac{8620}{(c-1)^6}+\frac{232}{(c-1)^5}\right)+O\left(\frac{1}{N^6}\right),
\\
\langle \tr X^{-1}\tr X^{2} \tr X^{2} \rangle_{c} ={}&  -N^2\left(\frac{16}{c-1}+24+8 (c-1)\right), 
\\
 \langle \tr X^{-2}\tr X^{-2} \tr X^{1} \rangle_{c} ={}& -\frac{1}{N^4}\left(\frac{72}{(c-1)^8}+\frac{144}{(c-1)^7}+\frac{88}{(c-1)^6}+\frac{16}{(c-1)^5}\right) +
\\
& -\frac{1}{N^6}\left(\frac{912}{(c-1)^{10}}+\frac{1824}{(c-1)^9}+\frac{1112}{(c-1)^8}+\frac{200}{(c-1)^7}\right)+\\ 
&-\frac{1}{N^8}\left(\frac{9144}{(c-1)^{12}}+\frac{18288}{(c-1)^{11}}+\frac{11160}{(c-1)^{10}}+\frac{2016}{(c-1)^9}\right) + O\left(\frac{1}{N^{10}}\right), 
\\
\langle \tr X^{-3}\tr X^{-2}\tr X^{2}\tr X^{2}\rangle_{c} ={}& \frac{1}{N^3}\left(\frac{720}{(c-1)^{10}}+\frac{2160}{(c-1)^9}+\frac{2544}{(c-1)^8}+\frac{1488}{(c-1)^7}+\frac{432}{(c-1)^6}+\frac{48}{(c-1)^5}\right)+
\\
& + \frac{1}{N^5}\left(\frac{19200}{(c-1)^{12}}+\frac{57600}{(c-1)^{11}}+\frac{66864}{(c-1)^{10}}+\frac{37728}{(c-1)^9}+\frac{10344}{(c-1)^8}+\frac{1080}{(c-1)^7}\right)+
\\
& + \frac{1}{N^7}\left( \frac{377040}{(c-1)^{14}}+\frac{1131120}{(c-1)^{13}}+\frac{1304688}{(c-1)^{12}}+\frac{724176}{(c-1)^{11}}+\frac{193056}{(c-1)^{10}}+\frac{19488}{(c-1)^9}\right)+O\left(\frac{1}{N^9}\right),
\\
\langle \tr X^{-3}\tr X^{-1}\tr X^{-1}\tr X^{2}\rangle_{c}={}&  -\frac{1}{N^5}\left(\frac{2400}{(c-1)^{11}}+\frac{7200}{(c-1)^{10}}+\frac{8304}{(c-1)^9}+\frac{4608}{(c-1)^8}+\frac{1224}{(c-1)^7}+\frac{120}{(c-1)^6}\right)+
\\
& - \frac{1}{N^7} \left( \frac{69600}{(c-1)^{13}}+\frac{208800}{(c-1)^{12}}+\frac{239904}{(c-1)^{11}}+\frac{131808}{(c-1)^{10}}+\frac{34464}{(c-1)^9}+\frac{3360}{(c-1)^8}\right)+
\\
& - \frac{1}{N^9} \left( \frac{1430400}{(c-1)^{15}}+\frac{4291200}{(c-1)^{14}}+\frac{4923408}{(c-1)^{13}}+\frac{2694816}{(c-1)^{12}}+\frac{700248}{(c-1)^{11}}+\frac{68040}{(c-1)^{10}}+\right)+O\left(\frac{1}{N^{11}}\right).
\end{align*}
}

\subsection*{Acknowledgements}
We are grateful to Marco Bertola, John Harnad, and Di Yang for very useful discussions.
This project has received funding from the European Union's H2020 research and innovation programme under the Marie Sk\l odowska--Curie grant No. 778010 {\em  IPaDEGAN}. G.R. wishes to thank the School of Mathematical Sciences at the University of Science and Technology of China in Hefei for hospitality during which part of this work was completed; the research of G.R. is supported by the Fonds de la Recherche Scientifique-FNRS under EOS project O013018F.

The authors wish to thank the anonymous referee for careful reading the manuscript and suggesting important improvements to the presentation.

\end{document}